\documentclass[12pt]{article}
\makeatletter
\def\maxwidth{ %
  \ifdim\Gin@nat@width>\linewidth
    \linewidth
  \else
    \Gin@nat@width
  \fi
}
\makeatother

\RequirePackage{amsmath,amssymb}
\RequirePackage{amsthm}
\RequirePackage{natbib}
\RequirePackage[colorlinks,allcolors=blue]{hyperref}

\usepackage{fullpage}
\usepackage[nodisplayskipstretch]{setspace}

\usepackage{bbm,graphicx,bm} %
\usepackage{enumitem}
\usepackage{multirow}
\usepackage[all,import]{xy}
\usepackage{appendix}
\usepackage{comment}
\usepackage{color}
\usepackage{soul}
\usepackage{pdflscape}
\usepackage{subcaption}

\usepackage{booktabs}
\usepackage{longtable}
\usepackage{array}
\usepackage[table]{xcolor}
\usepackage{wrapfig}
\usepackage{float}
\usepackage{colortbl}
\usepackage{pdflscape}
\usepackage{tabu}
\usepackage{threeparttable}
\usepackage{threeparttablex}
\usepackage[normalem]{ulem}
\usepackage{makecell}
\usepackage{afterpage}
\usepackage{thmtools}
\usepackage[compact]{titlesec}

\theoremstyle{definition}
\newtheorem{assumption}{Assumption}

\newtheorem{theorem}{Theorem}
\newtheorem{lemma}{Lemma}
\newtheorem{example}{Example}

\newtheorem{corollary}{Corollary}

\newcommand{\R}{\ensuremath{\mathbb{R}}}

\newcommand{\bbone}{\ensuremath{\mathbbm{1}}}
\newcommand{\E}{\ensuremath{\mathbb{E}}}

\newcommand{\Var}{\text{Var}}

\newcommand{\argmax}{\text{argmax}}
\newcommand{\argmin}{\text{argmin}}

\newcommand{\diag}{\text{diag}}

\def\super{\textsuperscript}

\newcommand{\<}{\langle}
\renewcommand{\>}{\rangle}

\newcommand\pderiv[2]{\frac{\partial #1}{\partial #2}}

\def\b1{\boldsymbol{1}}

\def\spacingset#1{\renewcommand{\baselinestretch}%
{#1}\small\normalsize} \spacingset{1}

\definecolor{RED}{RGB}{255,0,0}

\usepackage{soul}
\usepackage[utf8]{inputenc}

\begin{document}

\newcommand{\blind}{0}

\newcommand{\tit}{\LARGE Partial identification via conditional linear programs: estimation and policy learning}

\if0\blind
{\title{\tit
  }
  \author{Eli Ben-Michael\thanks{Assistant Professor, Department of Statistics \& Data Science and Heinz College of Information Systems \& Public Policy, Carnegie Mellon University. 4800 Forbes Avenue, Hamburg Hall, Pittsburgh PA 15213.  Email: \href{mailto:ebenmichael@cmu.edu}{ebenmichael@cmu.edu} URL:
  \href{https://ebenmichael.github.io}{ebenmichael.github.io}.\\
  We thank Erin Gabriel, Kosuke Imai, Zhichao Jiang, Edward Kennedy, Lihua Lei, Alexander Levis, and participants at the 2025 American Causal Inference Conference for helpful comments and discussions.
  }}

\date{}

\maketitle
\thispagestyle{empty}
}\fi

\if1\blind
\title{\bf \tit}
\maketitle
\thispagestyle{empty}
\fi

\pagenumbering{gobble}

\begin{abstract}

Many quantities of interest are only partially identified from observable data: the data can limit them to a set of plausible values, but not uniquely determine them. This paper develops a framework for covariate-assisted estimation, inference, and decision making in partial identification problems where the parameter of interest satisfies linear constraints conditional on covariates. Bounds on the parameter can be written as expectations of solutions to conditional linear programs with covariate-dependent objectives and constraints that must be estimated from data. Examples include estimands involving joint distributions of potential outcomes, policy learning with inequality-aware value functions, and instrumental variable settings. We propose two de-biased estimators. The first solves the conditional linear programs with plugin estimates and de-biases using standard output from linear program solvers, avoiding computationally demanding vertex enumeration. The second uses entropic regularization to create a smooth approximation, trading a small amount of approximation error for smoothness and computational efficiency. We establish asymptotic normality for both estimators, show robustness to first-order errors in estimating nuisance functions, construct Wald-type confidence bounds, and extend to policy learning problems where the value of a decision policy is only partially identified. We apply our methods to study the effects of Medicaid enrollment.

\end{abstract}

\clearpage
\pagenumbering{arabic}
\spacingset{1.65}
\section{Introduction}

Important estimands are oftentimes only partially identified from the observable data.
Even with an infinite amount of data we may only be able to limit the estimand of interest to a set of plausible values, rather than uniquely determine it.
For example, in causal inference settings it may not be possible to point-identify average treatment effects if treatment assignment is confounded; even with auxiliary information, such as an instrumental variable, the average treatment effect is not necessarily point identified.
There is also often interest in complex counterfactual estimands that go beyond average treatment effects, for instance the proportion of individuals who benefit from treatment or characteristics of the spread of treatment effects across the population.

This paper develops a framework for covariate-assisted estimation, inference, and decision-making in a class of partial identification problems where the parameter of interest satisfies a series of linear constraints.
In such settings, bounds for the parameter of interest are solutions to conditional linear programs (LPs) that optimize a linear function of unknown parameters subject to linear constraints, where both the objective function and the constraints may depend on additional covariates and need to be estimated.
This encompasses a wide range of settings; as illustrative examples we consider estimands related to the joint distribution of potential outcomes, policy learning with inequality-aware value functions, and instrumental variable settings (Section~\ref{sec:partial_id}).

We develop two de-biased estimators of bounds on the parameter of interest that are robust to errors in estimating the nuisance functions that define the objective and constraints (Section~\ref{sec:method}).
The first approach directly solves the optimization problems using plugin estimates of the conditional constraints and objectives and leverages properties of optimal solutions to de-bias these estimates using the output of standard LP solvers. In contrast to prior work, this estimator does not require deriving analytical solutions or enumeration of all possible solutions, which can be computationally intractable for large problems, but still allows for flexible estimation of the covariate-conditional constraints and objective vectors.
The second approach leverages recent advances in the theory of regularized linear programs to create smooth approximations to the conditional LPs using entropic regularization and estimates the smoothed bounds in a de-biased manner.
This approach trades off a small amount of approximation error for smoothness and computational efficiency.
In contrast to alternative smoothing schemes, entropic regularization introduces approximation error that can rapidly decay with the choice of hyper-parameter, and, crucially, does not require explicit enumeration of all feasible solutions.

For both estimators, we establish conditions for asymptotic normality and construct standard, Wald-type approximate confidence bounds.
These de-biased estimators are robust to first-order errors in estimating the covariate-conditional constraints and objective vectors, allowing for valid inference even when these nuisance functions are estimated using flexible machine learning methods.
We also extend these results to decision-making problems where the goal is to estimate an optimal decision policy that optimizes the lower or upper bound of a partially identified quantity (Section~\ref{sec:policy_learning}).
We demonstrate these estimators by applying them to a study of the effects of Medicaid enrollment on health outcomes (Section~\ref{sec:ohie}).

\paragraph{Related literature.}

This paper builds on several strands of literature on partial identification.
One line of work focuses on derivation of analytic or symbolic bounds \citep[e.g.][]{balke_bounds_1997, sachs_general_2023, gabriel_sharp_2024}.
This work typically characterizes sharp bounds on partially identified parameters as solutions to linear programs,
then enumerates all feasible solutions (i.e. all vertices of the feasible set) to obtain the bounds.
Because the number of solutions can grow combinatorially with the size of the problem, this approach can be computationally intractable for larger problems.
The approaches we develop do not require vertex enumeration and so can computationally scale to larger problems.
\citet{duarte_automated_2024} consider polynomial programs
and use numerical approaches to bound the parameter of interest.
Our approach similarly arrives at the bounds computationally, rather than analytically, but for a more limited class of problems for which we are able to provide de-biased estimators that allow us to conduct straightforward asymptotic inference.

A second line of work focuses on estimation and inference problems involving linear programs \citep[e.g.][]{freyberger_identification_2015, mogstad_using_2018, hsieh_inference_2022, andrews_inference_2023, fang_inference_2023, cox_simple_2023, cho_simple_2024, goff_inference_2025}.
The methods we develop are related to this line of work in that they directly use solutions to linear programs as part of the estimation and inference steps, and make use of similar properties of linear programs in the theoretical analysis. However, this literature primarily focuses on settings with aggregated constraints (that may hold conditional on instruments). In contrast, we consider covariate-assisted estimation and inference, where the constraints and objective function vary with auxiliary covariates and are nuisance functions that need to be flexibly estimated from data.
Similar to our proposed entropic regularized estimator, \citet{gafarov_simple_2025} works with a regularized version of the underlying linear program, though the form of regularization differs and again we consider the covariate-assisted case.

A related line of work focuses on de-biased, covariate-assisted estimation and inference for partially identified parameters.
Some of this work develops de-biased estimators of aggregations of conditional versions of analytic bounds on the parameter of interest \citep[e.g.][]{Kallus2022_harm,levis_assisted_2023, semenova_aggregated_2024, benmichael_ai_2024}, or
targets smooth approximations to these bounds \citep[e.g.][]{levis_assisted_2023, lanners_data_2025,whitehouse_inference_2026}.
Again, the methods we develop here differ in that they do not require analytical bounds, nor do they require enumeration of all possible solutions.
Closest to our proposal are \citet{ji_model-agnostic_2024} and \citet{khandamiryan_adaptive_2026}.
\citet{ji_model-agnostic_2024} consider estimating covariate-assisted bounds on estimands involving the joint distribution of potential outcomes \citep[see also][]{jiang_semiparametric_2025},
allowing for continuous potential outcomes unlike in our proposal; however, we consider a general setup that allows for more generic forms of partial identification problems.
In contemporaneous work to our proposal, \citet{khandamiryan_adaptive_2026} also consider conditional linear programs; however, they limit to cases where the objective vector is known, and do not consider a regularized estimator nor policy learning.

A final line of work focuses on learning optimal decision rules in settings where the value of a decision policy is only partially identified. This work builds on ideas from decision theory \citep[e.g.][]{Manski2005_partial}, and considers settings where there are unmeasured confounding variables, treatment-control overlap violations, or the value function involves counterfactual quantities \citep[e.g.][]{Kallus2021,Pu2021,Han2021_partial, Cui2021_partial,benmichael2021_safe,Zhang2022_safe,benmichael_asymmetric_2024,DAdamo2023}.
The approach to policy learning here builds directly on this work and generalizes it.

\section{Bounds defined by conditional linear programs}
\label{sec:partial_id}
We begin by describing the generic estimand and associated bound defined by conditional linear programs.
We then discuss several specific examples of such estimands related to partially identified parameters in causal inference and policy learning settings.

\subsection{Conditional linear programs}
\label{sec:clp}

Consider a set of covariates $X \in \mathcal{X}$. Our primary interest is in a linear functional of some true non-negative function $\bm{p}:\mathcal{X} \to \R^K_+$ defined by a function $\bm{c}:\mathcal{X} \to \R^K$, i.e.
\[
  \theta \equiv \E[\<\bm{c}(X), \bm{p}(X)\>] = \E\left[\sum_{k=1}^K c_k(X) p_k(X)\right].
\]
Our focus will be on settings where the objective vector $\bm{c}(x)$ is not necessarily known but is point-identifiable from observable data while $\bm{p}(x)$ is not necessarily point-identified for all $x \in \mathcal{X}$. Rather, we can only observe $J \leq K$ linear functions of $\bm{p}(x)$, i.e. $A \bm{p}(x) = \bm{b}(x)$ for some known matrix $A \in \R^{J \times K}$ and some unknown function $\bm{b}:\mathcal{X} \to \R^J$ that is identifiable from observable data, such as conditional moment restrictions.
For each value of the covariates $x \in \mathcal{X}$, these linear constraints along with non-negativity imply that the true function $\bm{p}(x)$ lies in a convex set $\mathcal{P}(x) \equiv \{\bm{p} \in \R^K_+ \mid A \bm{p} = \bm{b}(x)\}$. Note that this framework also allows for inequality constraints via standard LP transformations.
In Appendix~\ref{sec:examples} we provide examples of causal inference problems involving potential outcomes and instrumental variables that fit into this form.

In many cases, the conditional feasible set $\mathcal{P}(x)$ contains more than one element (e.g. because there are more unknowns than constraints), so $\bm{p}(x)$ is not point-identified the parameter of interest $\theta$ may also not be point identified. We will instead focus on finding a lower bound $\theta_L$ such that $\theta \geq \theta_L$.
To construct this, we define $X$-conditional linear programs that minimize the linear function of $\bm{p}(x)$ over the conditional feasible set $\mathcal{P}(x)$:
\begin{equation}
  \label{eq:partial_id}
  \theta_L(x) = \min_{\bm{p} \in \mathcal{P}(x)} \<\bm{c}(x), \bm{p}\>.
\end{equation}
Note that to find a corresponding upper bound, we can simply negate the objective vector.
In order for this bound to be meaningful, we will assume that the conditional feasible set $\mathcal{P}(x)$ is both non-empty and bounded for all $x \in \mathcal{X}$ and that both $\|\bm{b}\|_\infty < \infty$ and $\|\bm{c}\|_\infty < \infty$, where for a vector-valued function $\bm{f}(x)$, we denote $\|\bm{f}\|_\infty = \max_i \sup_{x} |f_i(x)|$.
As an additional bit of notation, for a scalar-valued function $f: \mathcal{X} \to \R$ we will denote $\|f\|_2^2 = \E[f(X)^2]$, and for a vector-valued function $\bm{f}:\mathcal{X} \to \R^p$ we will denote $\|\bm{f}\|_2^2 = \sum_{i=1}^p \|f_i\|_2^2$.
To simplify the exposition, we will also assume that the constraint matrix $A$ is sufficiently pre-processed to remove redundant constraints and so has full row rank equal to $J$.
In this setup, the conditional feasible set $\mathcal{P}(x)$ is compact, so the minimum is finite and is attained at some minimizer $\bm{p}^\ast(x)$
and the point-wise bound $\theta_L(x)$ is sharp given the point-wise constraints.
We will also work with the Lagrangian dual of the conditional linear program, where the conditional dual feasible set is $\Lambda(x) \equiv \{\bm{\lambda} \in \R^J \mid A' \bm{\lambda} \leq \bm{c}(x)\}$ and the dual linear program is
\begin{equation}
\label{eq:dual_partial_id}
\theta_L(x) = \max_{\bm{\lambda} \in \Lambda(x)} \<\bm{\lambda}, \bm{b}(x)\>,
\end{equation}
because an optimal solution to the primal linear program exists and strong duality holds. We will denote $\bm{\lambda}^\ast(x)$ as a maximizer of the dual linear program .

We can then construct an overall, valid bound as the expected solution to the conditional linear program, i.e. $\theta_L = \E[\theta_L(X)]$.
Here we do not assume any additional global constraints, so the overall feasible set is the set of all functions that satisfy the point-wise constraints and the optimization problems for the overall bounds on $\theta$ over $\bm{p}(\cdot)$ separate across different values of the covariates $x$, absent measurability issues in the constraint and objective vectors.

\section{Estimation}
\label{sec:method}
We now turn to estimating the bounds.
Throughout, we will assume that we have iid samples of random variables $O$ that include the covariates $X$ as well as variables that give information about the constraint and objective vectors, respectively.
Our principal assumption is that we can estimate the constraint vector $\bm{b}(x)$ and the objective vector $\bm{c}(x)$ in a de-biased manner using the information contained in the observed data $O$.

\begin{assumption}
  \label{a:riesz}
  For each constraint $j=1,\ldots,J$ there exists a de-biasing function $\varphi^{(b)}_{j}(O; b_j)$
  such that $\E[\varphi^{(b)}_{j}(O; b_j) \mid X] = 0$, $\E[\varphi_{j}^{(b)}(O; b_j)^2] < \infty$, and for any bounded function $\bar{\bm{b}}:\mathcal{X} \to \R^J$,  $\E[\varphi_{j}^{(b)}(O; \bar{b}_j) \mid X = x] = b_j(x) - \bar{b}_j(x)$  for all $x \in \mathcal{X}$.
  Similarly, for each decision variable $k=1,\ldots,K$ there exists a de-biasing function $\varphi_{k}^{(c)}(O; c_k)$ such that $\E[\varphi^{(c)}_{k}(O; c_k) \mid X] = 0$, $\E[\varphi_{k}^{(c)}(O; c_k)^2] < \infty$, and for any bounded function $\bar{\bm{c}}:\mathcal{X} \to \R^K$, $\E[\varphi_{k}^{(c)}(O; \bar{c}_k) \mid X = x] = c_k(x) - \bar{c}_k(x)$  for all $x \in \mathcal{X}$.
\end{assumption}

This structure allows us to estimate the bound in a de-biased manner that will be robust to errors in estimating the constraint and objective vectors.
We will assume that we have estimates of the constraint and objective vectors $\hat{\bm{b}}(\cdot)$ and $\hat{\bm{c}}(\cdot)$, along with estimates of the de-biasing functions $\hat{\bm{\varphi}}^{(b)}(\cdot)$ and $\hat{\bm{\varphi}}^{(c)}(\cdot)$, that are fit on a separate, independent sample. When it is clear from context we will use
$\varphi_j^{(b)}(O) \equiv \varphi_j^{(b)}(O; b_j)$ and $\varphi_k^{(c)}(O) \equiv \varphi_k^{(c)}(O, c_k)$ to denote the true de-biasing functions using the true constraint or objective functions, and
$\hat{\varphi}_j^{(b)}(O) \equiv \hat{\varphi}_j^{(b)}(O, \hat{b}_j)$ and $\hat{\varphi}_k^{(c)}(O) \equiv \hat{\varphi}_k^{(c)}(O, \hat{c}_k)$ to denote the estimated de-biasing functions using the estimated constraint or objective functions.
Note that while we consider estimates fit out of sample, our analysis can also be extended to cross-fitting the estimates, which we use in practice to avoid losing sample size.

With the estimated constraint and objective vectors, we can define the estimated conditional primal and dual feasible sets $\widehat{\mathcal{P}}(x) \equiv \{\bm{p} \in \R^K_+ \mid A \bm{p} = \hat{\bm{b}}(x)\}$, and $\widehat{\Lambda}(x) \equiv \{\bm{\lambda} \in \R^J \mid A' \bm{\lambda} \leq \hat{\bm{c}}(x)\}$.
We assume that both of these estimated feasible sets are non-empty and that $\widehat{\mathcal{P}}(x)$ is bounded for all $x \in \mathcal{X}$ (this can be guaranteed via preprocessing the estimated constraint and objective vectors) and define a minimizer of the estimated primal problem as $\hat{\bm{p}}(x) \in \argmin_{\bm{p} \in \widehat{\mathcal{P}}(x)} \<\hat{\bm{c}}(x), \bm{p}\>$ and a maximizer of the estimated dual problem as $\hat{\bm{\lambda}}(x) \in \argmax_{\bm{\lambda} \in \widehat{\Lambda}(x)} \<\bm{\lambda}, \hat{\bm{b}}(x)\>$.
The value of the estimated program is then $\hat{\theta}_L(x) = \<\hat{\bm{c}}(x), \hat{\bm{p}}(x)\> = \<\hat{\bm{\lambda}}(x), \hat{\bm{b}}(x)\>$.
We also assume the estimates of the constraint and objective vectors and their de-biasing functions are consistent, and that the rate of convergence for de-biased estimates of linear functions of the constraint and objective vectors is controlled by a rate $r_n$.

\begin{assumption}
  \label{a:rates}
  $\sup_{x\in\mathcal{X}}\E[\varphi^{(b)}_j(O)^2 \mid X = x] < \infty$ and $\sup_{x\in\mathcal{X}}\E[\varphi^{(c)}_k(O)^2 \mid X = x] < \infty$ for all  $j=1,\ldots,J$ and $k=1,\ldots,K$.
   $\hat{\bm{b}}$, $\hat{\bm{c}}$, $\hat{\bm{\varphi}}^{(b)}$ and $\hat{\bm{\varphi}}^{(c)}$ are fit on a separate, independent sample, with $\|\hat{\bm{b}}\|_\infty < \infty$ and $\|\hat{\bm{c}}\|_\infty < \infty$. The estimated feasible sets $\widehat{\mathcal{P}}(x)$ and $\widehat{\Lambda}(x)$ are non-empty and $\widehat{\mathcal{P}}(x)$ is bounded for (almost) all $x \in \mathcal{X}$. $\|\hat{\bm{b}} - \bm{b}\|_\infty = o_p(1)$, $\|\hat{\bm{c}} - \bm{c}\|_\infty = o_p(1)$, $\|\hat{\bm{\varphi}}^{(b)} - \bm{\varphi}^{(b)}\|_\infty = o_p(1)$, and $\|\hat{\bm{\varphi}}^{(c)} - \bm{\varphi}^{(c)}\|_\infty = o_p(1)$.
   Furthermore, there is some rate $r_n$ such that for any linear function of $\bm{b}(X)$ defined by $f:\mathcal{X} \to \R^J$,  we have that $\left|\E\left[\<\bm{f}(X), \hat{\bm{b}}(X) + \hat{\bm{\varphi}}^{(b)}(O) - \bm{b}(X)\>\right]\right| \lesssim \sup_{x \in \mathcal{X}} \|\bm{f}(x)\|_2 \times r_n$.
  Similarly, for any linear function of $\bm{c}(X)$ defined by $f:\mathcal{X} \to \R^K$, we have
   that $\left|\E\left[\<\bm{f}(X), \hat{\bm{c}}(X) + \hat{\bm{\varphi}}^{(c)}(O) - \bm{c}(X)\>\right]\right| \lesssim \sup_{x \in \mathcal{X}} \|\bm{f}(x)\|_2 \times r_n$,
  where the notation $\lesssim$ means that the left hand side is bounded by a constant multiple of the right.
\end{assumption}

In Appendix~\ref{sec:examples} we give the forms of these de-biasing functions for several causal inference examples.
Broadly,
these de-biasing functions can be written in terms of weighted residuals of the form $\varphi_{j}^{(b)}(O; b_j) = g_{j}^{(b)}(O; b_j)(Y_j - b_j(X))$ and $\varphi_{k}^{(c)}(O; c_k) = g_{k}^{(c)}(O; c_k)(C_k - c_k(X))$ where  $Y_j$ and $C_k$ are  elements of $O$ and $g_{j}^{(b)}(\cdot)$, $g_{k}^{(c)}(\cdot)$ are corresponding Riesz representers (e.g. inverse propensity weights).
The uniform conditional second-moment bound assumption on the de-biasing functions is thus akin to a standard overlap assumption on the Riesz representers.
To estimate the de-biasing functions, we can estimate the Riesz representers, giving a double robust or double machine learning-style estimate \citep[e.g.][]{Robins1994, chernozhukov_doubledebiased_2018}. In this case the rates will depend on the product of the errors in the estimates of the constraints $\bm{b}(\cdot)$ and objectives $\bm{c}(\cdot)$ with the errors in their respective Riesz representers---e.g. $r_n = \max\{\|\hat{\bm{b}} - \bm{b}\|_2\|\hat{\bm{g}}^{(b)} - \bm{g}^{(b)}\|_2, \|\hat{\bm{c}} - \bm{c}\|_2\|\hat{\bm{g}}^{(c)} - \bm{g}^{(c)}\|_2\}$ where $\hat{\bm{g}}^{(b)}$ and $\hat{\bm{g}}^{(c)}$ are estimates of the Riesz representers.
Finally, note that in our setting the dimension $K$ and number of constraints $J$ are fixed and do not grow with the sample size.

\subsection{De-biased estimation using primal and dual solutions}
\label{sec:primal_estimation}

\paragraph{De-biased estimator.} To begin, consider a \emph{plugin} estimator of the lower bound $\theta_L$ that uses a solution to the estimated primal linear program: $\frac{1}{n}\sum_{i=1}^n \hat{\theta}_L(X_i)$.
While appealingly simple, plugin estimators typically have a bias that scales with the error in estimating nuisance functions, the constraint and objective vectors in our case. Therefore, if we use flexible machine learning methods to estimate the objective and constraint vectors, the plugin estimator may have a high degree of bias.
To address this, we will estimate and correct for the bias using the de-biasing functions in Assumption~\ref{a:riesz}.

Roughly speaking, if the solution to the conditional linear program $\theta_L(x)$ was smooth with respect to the constraint and objective vectors $\bm{b}(x)$ and $\bm{c}(x)$, then via a Taylor expansion the first-order bias of the plugin estimator would be approximately
\[
  -\E\left[\left\<\bm{\varphi}^{(c)}(O, \hat{\bm{c}}), \nabla_{\bm{c}}\theta_L(X)\right\> + \left\<\nabla_{\bm{b}}\theta_L(X),\bm{\varphi}^{(b)}(O, \hat{\bm{b}})\right\> \right].
\]
This would lead to a procedure to construct a one-step bias corrected (or double machine learning) estimator: create a plugin estimate of the bias using the estimated de-biasing functions $\hat{\bm{\varphi}}^{(b)}$ and $\hat{\bm{\varphi}}^{(c)}$ and the gradients of the value of the conditional linear program $\nabla_{\bm{b}}\hat{\theta}_L(x)$ and $\nabla_{\bm{c}}\hat{\theta}_L(x)$, and then subtract this from the plugin estimator.

Unfortunately the value of the conditional linear program $\theta_L(x)$ is not necessarily smooth. Nonetheless, we can still use the intuition above to construct a de-biased estimator for the lower bound $\theta_L$, though its properties will be different from usual one-step or double machine learning-style estimators due to the lack of smoothness.
To do so, note that by Danskin's theorem the subgradients of the value of the conditional linear program with respect to the constraint vector $\bm{b}(x)$ are the set of optimal dual solutions and the supergradients with respect to the objective vector $\bm{c}(x)$ are the set of optimal primal solutions. So we construct a de-biased estimator using solutions to the estimated primal and dual linear programs  $\hat{\bm{p}}(x)$ and $\hat{\bm{\lambda}}(x)$ as
\begin{equation}
  \label{eq:bound_est}
  \hat{\theta}_L = \frac{1}{n}\sum_{i=1}^n \hat{\theta}_L(X_i) + \left\< \hat{\bm{\varphi}}^{(c)}(O_i), \hat{\bm{p}}(X_i)\right\>  + \left\<\hat{\bm{\lambda}}(X_i), \hat{\bm{\varphi}}^{(b)}(O_i)\right\>.
\end{equation}

This primal-dual estimator avoids analytically characterizing the bounds, which would require enumerating and evaluating all of the up to $\binom{K}{J}$ vertices of the feasible set $\mathcal{P}(x)$ and so is typically only tractable for problems with a a number of variables $K$ and constraints $J$ in the low single digits \citep[as in, e.g.,][]{Kallus2022_harm,levis_assisted_2023, semenova_aggregated_2024, benmichael_ai_2024}.
Instead, we can use a linear program optimizer to compute the primal and dual variables, such as the simplex algorithm \citep{dantzig_maximization_1951}, which has polynomial run-time in typical cases and is efficient in practice, though it has exponential run-time in the worst-case \citep{spielman_smoothed_2004}.

\paragraph{Basic feasible solutions.}To understand the properties of this estimator, we can view the primal and dual solutions in terms of the \emph{basic feasible solutions} to the conditional LP.
Because the matrix of marginal constraints $A \in \R^{J \times K}$ has rank equal to $J$, and the feasible set  $\mathcal{P}(x)$ is non-empty by assumption, there must exist a solution to this LP that lies on the boundary of the feasible set $\mathcal{P}(x)$.
In  particular there exists a solution that is a basic feasible solution (BFS) with $J$ active, non-negative variables and $K - J$ inactive variables on the boundary at 0.

For a set of $J$ unique indices $B = \{i_1,\ldots, i_{J}\}$, let $A_B$ denote the sub-matrix of $A$ that includes all rows and the $J$ columns defined by $B$, and similarly for a vector $v$ let $v_B$ denote the sub-vector of $v$ that includes the $J$ elements defined by $B$.
Then, a basic solution to the LP is given by $\bm{p}_B = A_B^{-1} \bm{b}(x)$, and $p_i = 0$ for all $i \not \in B$. If $\bm{p}_B \geq 0$, then it is also a basic \emph{feasible solution}.
Immediately abusing this notation, we will also use $A_B^{-1}$ to denote the re-ordered inverse matrix that places the rows of the inverse of the sub-matrix $A_B$ into $i_1,\ldots,i_J$, and sets the other rows equal to zero.
We will denote the set of all possible bases as $\mathcal{B} = \{B \subset \{1,\ldots,K\} \mid |B| = J,  \text{rank}(A_B) = J\}$, the set of primal-feasible bases at covariate value $x$ as $\mathcal{B}(x) = \{B \in \mathcal{B} \mid A_B^{-1} \bm{b}(x) \geq 0\}$, and the set of dual-feasible bases at covariate value $x$ as $\mathcal{D}(x) \equiv \{B \in \mathcal{B}  \mid A'A_B^{-1\prime} \bm{c}(x) \leq \bm{c}(x)\}$.
Finally, we can define the set of primal and dual feasible, and hence optimal, bases as $\mathcal{B}^\ast(x) = \mathcal{B}(x) \cap \mathcal{D}(x)$, which is non-empty because an LP with a finite solution has at least one such basis.

Thus, we can write a primal-dual solution pair to the conditional LP as $\bm{p}^\ast(x) = A_{B^\ast(x)}^{-1} \bm{b}(x)$ and $\bm{\lambda}^\ast(x) = A_{B^\ast(x)}^{-1\prime} \bm{c}(x)$ for some basis $B^\ast(x) \in \mathcal{B}^\ast(x)$, and the value of the conditional LP as $\theta_L(x) = \<\bm{c}(x), A_{B^\ast(x)}^{-1} \bm{b}(x)\>$.
Similarly, we will assume that we are using an LP solver that returns an optimal basis (such as the simplex method), so that we can write the estimated primal-dual solution pair as $\hat{\bm{p}}(x) = A_{\widehat{B}(x)}^{-1} \hat{\bm{b}}(x)$ and $\hat{\bm{\lambda}}(x) = A_{\widehat{B}(x)}^{-1\prime} \hat{\bm{c}}(x)$ for some basis $\widehat{B}(x)$ that is primal and dual feasible using the estimated constraint and objective vectors.

Despite the estimated conditional LP $\hat{\theta}_L(x)$ using estimated constraint and objective vectors, it is still possible that the resulting estimated basis $\widehat{B}(x)$ is optimal for the true LP and so $\theta_L(x) = \<\bm{c}(x), A_{\widehat{B}(x)}^{-1} \bm{b}(x)\>$. The following lemma disentangles the contribution to the bias of $\hat{\theta}_L$ from the estimated basis and the estimated constraint and objective vectors.
\begin{lemma}
  \label{lem:direct_est_bias_basis}
    Under Assumptions~\ref{a:riesz} and \ref{a:rates},
    \begin{align*}
        \E[\hat{\theta}_L] - \theta_L = \E\left[\left\<\bm{c}(X), \left(A^{-1}_{\widehat{B}(X)} - A^{-1}_{B^\ast(X)}\right) \bm{b}(X)\right\>\right] + O_p\left(r_n + \|\hat{\bm{b}} - \bm{b}\|_2 \|\hat{\bm{c}} - \bm{c}\|_2\right).
    \end{align*}
\end{lemma}
\noindent Lemma~\ref{lem:direct_est_bias_basis} shows that when the estimated basis is optimal, non-smoothness of the value of the conditional LP will not be an issue and the de-biased estimator will have the usual second-order bias properties of one-step or double machine learning-style estimators.
However, when the estimated basis $\widehat{B}(x)$ is not optimal, non-smoothness can lead to additional first order bias that depends on the suboptimality of the estimated basis.

Even if the estimated basis $\widehat{B}(x)$ is not optimal, by weak duality the estimated basis will still provide a valid lower bound if it is feasible for the dual problem. If it is also optimal, then the bound is sharp (\citet{ji_model-agnostic_2024} also make use of this fact).
Thus, we can characterize the difficulty of finding a valid lower bound as the difficulty of finding a dual-feasible basis.
To do so, we define the \emph{dual infeasibility gap} as $\mathcal{E}_{\lambda}(x) = \min_{B \in \mathcal{B} \setminus \mathcal{D}(x)} \|[A'A^{-1\prime}_B \bm{c}(x) - \bm{c}(x)]_+\|_\infty$, where $[z]_+ = \max\{z,0\}$ is the positive part of $z$ and we define $\mathcal{E}_\lambda(x)  = \infty$ when all basic dual solutions are feasible and so $\mathcal{B} \setminus \mathcal{D}(x) = \emptyset$.
We then make the following \emph{dual infeasibility margin condition} assumption about how often there is a nearly-feasible solution.

\begin{assumption}[Dual infeasibility margin]
  \label{a:infeasible_gap}
  For some $\alpha > 0$, $P(\mathcal{E}_\lambda(X) \leq t) \lesssim t^\alpha $ for all $t > 0$.
\end{assumption}
The \emph{margin parameter} $\alpha$ determines the difficulty of the problem. If $\alpha$ is small then there is a high probability across the covariate space $\mathcal{X}$ that there are very nearly feasible dual solutions and so it is more likely that we mistakenly select an infeasible dual solution. If $\alpha$ is large then this is less likely.

If a lower bound is valid but not sharp, this is because the estimated basis $\widehat{B}(x)$ is feasible for the dual problem but not an optimal solution to the dual problem. We will denote the set of such bases as $\mathcal{D}^{-}(x) = \{B \in \mathcal{D}(x) \mid \<\bm{c}(x), A^{-1}_B \bm{b}(x)\> < \theta_L(x)\}$.
By complementary slackness, a sufficient way to certify that a dual-feasible basis is optimal is to check whether it is also feasible for the primal problem.
Therefore, we can consider the difficulty of detecting that a dual-feasible basis is not optimal by defining the \emph{optimality detection gap} as $\mathcal{E}_p(x) = \min_{B \in \mathcal{D}^{-}(x)} \|[A_B^{-1} \bm{b}(x)]_-\|_\infty$, where $[z]_- = \max\{-z,0\}$ is the negative part of $z$ and we define $\mathcal{E}_p(x) = \infty$ when all dual-feasible bases are also optimal, such as in point-identified settings, and so $\mathcal{D}^{-}(x)= \emptyset$.
We then can define an analogous \emph{optimality detection margin condition}.
\begin{assumption}[Optimality detection margin]
  \label{a:primal_infeasible_gap}
  For some $\alpha > 0$, $P(\mathcal{E}_p(X) \leq t) \lesssim t^\alpha $ for all $t > 0$.
\end{assumption}

Here again small values of the margin parameter $\alpha$ correspond to more difficult problems where there are many cases with sub-optimal bases  that are nearly feasible primal solutions and so difficult to detect. In such cases even though we may get a valid lower bound by selecting a dual-feasible solution, it may not be sharp.
Note that in point-identified settings all valid dual variables are also optimal and so we will always have a sharp lower bound.
Margin conditions such as these have been used to analyze the performance of
estimators for classification in prediction and causal inference settings \citep{Audibert2007, Luedtke2016, kennedy_sharp_2020}, and estimating bounds on partially identified parameters \citep{levis_assisted_2023,DAdamo2023, semenova_aggregated_2024}. In Appendix~\ref{sec:margin_ex}, we work through a simple example with binary potential outcomes where these margin conditions involve the conditional average treatment effect and reduce to standard assumptions made in the literature.
Finally, note that it is possible for the margin parameter to be different for the two conditions, but for simplicity we will choose the smaller of the two.

\paragraph{Convergence rate and inference.}
These margin conditions allow us to characterize the rate of convergence of the de-biased estimator in Eqn~\eqref{eq:bound_est} to either a valid lower bound (with Assumption~\ref{a:infeasible_gap}) or a sharp lower bound (with Assumptions~\ref{a:infeasible_gap} and \ref{a:primal_infeasible_gap}).
\begin{theorem}
  \label{thm:direct_est_rate}
  Define $\tilde{\theta}_L  = \frac{1}{n} \sum_{i=1}^n \left\<\bm{c}(X_i) + \bm{\varphi}^{(c)}(O_i), \hat{\bm{p}}(X_i)\right\> + \left\<\hat{\bm{\lambda}}(X_i), \bm{\varphi}^{(b)}(O_i)\right\>$ as the estimator using the true nuisance functions and the estimated primal and dual variables.
  Under Assumptions~\ref{a:riesz}, \ref{a:rates}, and \ref{a:infeasible_gap},
  \begin{align*}
    \hat{\theta}_L  & \leq \theta_L + \tilde{\theta}_L - \E[\tilde{\theta}_L]  + O_p\left(r_n + \|\hat{\bm{b}} - \bm{b}\|_2 \|\hat{\bm{c}} - \bm{c}\|_2 + \|\hat{\bm{c}} - \bm{c}\|_\infty^{1 + \alpha} \right) + o_p(n^{-1/2}).
  \end{align*}
  If in addition Assumption~\ref{a:primal_infeasible_gap} holds, then
  \begin{align*}
    \hat{\theta}_L - \theta_L &= \tilde{\theta}_L - \E[\tilde{\theta}_L]  + O_p\left(r_n + \|\hat{\bm{b}} - \bm{b}\|_2 \|\hat{\bm{c}} - \bm{c}\|_2 + \|\hat{\bm{b}} - \bm{b}\|_\infty^{1 + \alpha}  + \|\hat{\bm{c}} - \bm{c}\|_\infty^{1 + \alpha}\right) + o_p(n^{-1/2}).
  \end{align*}
\end{theorem}

Theorem~\ref{thm:direct_est_rate} shows that $\hat{\theta}_L$ lower bounds $\theta_L$ plus a mean-zero term  $\tilde{\theta}_L - \E[\tilde{\theta}_L]$ and additional bias terms.
First, we see the de-biased rate $r_n$ and the second-order bias term $\|\hat{\bm{b}} - \bm{b}\|_2 \|\hat{\bm{c}} - \bm{c}\|_2$ that are standard in one-step or double machine learning-style estimators and discussed in Lemma~\ref{lem:direct_est_bias_basis}.
Next there is a bias term due to selecting a dual-infeasible basis.
Under the dual-infeasibility margin condition in Assumption~\ref{a:infeasible_gap},  this error decreases at a faster rate than the worst-case error of the objective vector by a factor of the margin parameter $\alpha$.
Roughly, this is because in cases where we have an infeasible solution, the infeasibility gap must be smaller than the estimation error in the objective vector. Note that when the objective vector $\bm{c}(x)$ is known, Theorem~\ref{thm:direct_est_rate} shows that $\hat{\theta}_L$ will give us an asymptotically valid, though potentially conservative, lower bound because we will always select a dual-feasible basis.
Under the optimality detection margin condition in Assumption~\ref{a:primal_infeasible_gap} we can characterize when this bound will be sharp by lower bounding the negative bias by a term that scales with $\|\hat{\bm{b}} - \bm{b}\|_\infty^{1 + \alpha}$, where similar to the dual-infeasibility case, if we select a sub-optimal dual solution, and hence a primal-infeasible solution, the level of infeasibility must be smaller than the estimation error in the constraint vector.

The challenge with using Theorem~\ref{thm:direct_est_rate} to conduct inference is that even if the estimates of the nuisance functions converge quickly enough so that the bias terms are $o_p(n^{-1/2})$, the mean-zero term $\tilde{\theta}_L - \E[\tilde{\theta}_L]$ may not be asymptotically normal because of the non-smoothness of the value of the conditional LP.
In particular, the mean zero term may not stabilize because there can be multiple primal and dual solutions, with the particular solution selected depending on vagaries of the error structure of the estimated nuisance functions.
However, if the primal and dual solutions are unique then this issue does not arise.
So from Theorem~\ref{thm:direct_est_rate} we can ensure that $\hat{\theta}_L$ is asymptotically normal under the following assumptions.

\begin{assumption}[Rate and regularity conditions for asymptotic normality]
  \label{a:rate_bfs}
    \mbox{}
    \begin{enumerate}[label=(\roman*)]
      \item \label{a:rate_bfs_ii} Either (a) for almost all $x \in \mathcal{X}$ there exists a unique primal solution and $\bm{b}(x)$ is known up to a scalar function so that $\bm{\varphi}^{(b)}(o) = f(o) \bm{b}(x)$ for some function $f$ such that $E[f(O)^2] < \infty$; (b) for almost all  $x \in \mathcal{X}$ there exists a non-degenerate primal solution and $\bm{c}(x)$ is known up to a scalar function so that $\bm{\varphi}^{(c)}(o) = f(o) \bm{c}(x)$ for some function $f$ such that $E[f(O)^2] < \infty$; or (c) for almost all $x \in \mathcal{X}$ there exists a unique, non-degenerate primal solution.
      \item \label{a:rate_bfs_i} $\|\hat{\bm{b}} - \bm{b}\|_\infty^{1+\alpha} = o_p(n^{-\frac{1}{2}})$, $\|\hat{\bm{c}} - \bm{c}\|_\infty^{1 + \alpha} = o_p(n^{-\frac{1}{2}})$, ${\|\hat{\bm{b}} - \bm{b}\|_2 \|\hat{\bm{c}} - \bm{c}\|_2} = o_p(n^{-\frac{1}{2}})$ and $r_n = o(n^{-\frac{1}{2}})$.
      \item \label{a:rate_bfs_iii} Define $g(O) \equiv \left\<\bm{c}(X) + \bm{\varphi}^{(c)}(O), \bm{p}^\ast(X)\right\> + \left\<\bm{\lambda}^\ast(X), \bm{\varphi}^{(b)}(O)\right\>$ for an optimal primal and dual solution pair $\bm{p}^\ast(X)$ and $\bm{\lambda}^\ast(X)$.  $V_L \equiv \Var(g(O)) > 0$ and for every $\varepsilon > 0$, $\lim_{n \to \infty}\E\left[\left(g(O) - \theta_L\right)^2 \bbone\{|g(O) - \theta_L| \geq \varepsilon \sqrt{n V_L}\}\right] = 0$.
    \end{enumerate}

\end{assumption}

Assumption~\ref{a:rate_bfs}\ref{a:rate_bfs_ii} generalizes the above discussion on uniqueness of primal and dual solutions to be weaker for cases where the objective or constraint vectors are known (up to a scalar function) and is similar to uniqueness conditions in other settings \citep[e.g.][]{Luedtke2016}. If the constraints are known, we only require that the primal solution is unique, while if the objective is known we only require that the primal solution is non-degenerate and so the dual solution is unique.\footnote{
The objective function is often known up to a scalar function, e.g. the propensity score.}
Assumption~\ref{a:rate_bfs}\ref{a:rate_bfs_i}
characterizes the rates of convergence for the estimated constraint and objective vectors that enable asymptotic normality.
Easier optimization problems with larger margin parameters $\alpha$ allow for slower rates of convergence, while more difficult problems with smaller margin parameters require faster rates.
Note that if the objective vector $\bm{c}(x)$ is known, then we only require the de-biased rate condition $r_n = o(n^{-1/2})$ in order to have a valid, but possibly conservative, lower confidence bound.
Assumption~\ref{a:rate_bfs}\ref{a:rate_bfs_iii} is a standard Lindeberg regularity condition.

\begin{corollary}
  \label{cor:direct_est_normal}
  Under Assumptions~\ref{a:riesz}, \ref{a:rates}, \ref{a:infeasible_gap}, \ref{a:primal_infeasible_gap}, and \ref{a:rate_bfs}, $\frac{\sqrt{n}\left(\hat{\theta}_L - \theta_L\right)}{\sqrt{V_L}} \Rightarrow N(0, 1)$.
\end{corollary}

Corollary~\ref{cor:direct_est_normal} shows that under the rate and regularity conditions in Assumption~\ref{a:rate_bfs}, the de-biased estimator $\hat{\theta}_L$ is asymptotically normally distributed around the true bound $\theta_L$.
To construct a lower confidence bound, we estimate the variance
with
\begin{align*}
  \hat{V}_L & = \frac{1}{n}\sum_{i=1}^n \left(\hat{\theta}_L(X_i) + \left\< \hat{\bm{\varphi}}^{(c)}(O_i), \hat{\bm{p}}(X_i)\right\>  + \left\<\hat{\bm{\lambda}}(X_i), \hat{\bm{\varphi}}^{(b)}(O_i)\right\>- \hat{\theta}_L\right)^2.
\end{align*}
We can then create a Wald-type approximate $(1 - \alpha)$ lower confidence bound for $\theta_L$ as $\hat{\theta}_L - z_{1-\alpha} \sqrt{\hat{V}_L / n}$, where $z_{1-\alpha}$ is the $1-\alpha$ quantile of a standard normal distribution.\footnote{Using the analogous estimator for the upper bound and its variance, we can also construct an approximate $(1 - \alpha)$ confidence interval for the partially identified set by combining approximate $(1-\alpha/2)$ lower and upper confidence bounds. We could also construct an approximate $(1 - \alpha)$ confidence interval for the partially identified parameter $\theta$ by following the procedure in \citet{imbens_confidence_2004}, though this requires uniform joint convergence of the bound estimators, which we have not directly established in Corollary~\ref{cor:direct_est_normal}.}
Finally, note that Corollary~\ref{cor:direct_est_normal} is point-wise in the underlying distribution. It is also possible to show convergence uniformly across a class of distributions satisfying stronger, uniform versions of Assumptions~\ref{a:riesz}-\ref{a:rate_bfs}. In Appendix~\ref{sec:uniform_convergence} we state these assumptions and give a uniform result.

\subsection{Estimating bounds with entropic regularization}
\label{sec:entropic}
\paragraph{Entropic estimator.} We have seen that the non-smoothness of the conditional LPs creates two key challenges.
First, small errors in either the constraint or objective vectors can lead to large changes in the optimal solution of the conditional LP.
Second, if there are multiple solutions then it may not be possible to conduct standard inference based on asymptotic normality.
As an alternative, we propose to target a smoothed version of the estimand that adds an entropy penalty to the conditional LP.
This will induce some approximation error, leading to lower bounds that are somewhat positively biased, but we will typically be able to tune the level of regularization so that the approximation error will be small relative to the standard error.

For constraint and objective vectors $\bm{b}$ and $\bm{c}$, define the solution to the \emph{entropic conditional linear program} as
\begin{equation}
  \label{eq:condl_entropic_primal}
  \bm{p}^{\eta}(\bm{b}, \bm{c}) = \underset{A\bm{p} = \bm{b}}{\argmin}\;\<\bm{c}, \bm{p}\> + \frac{1}{\eta}\sum_{k=1}^K p_k(\log p_k - 1).
\end{equation}
As shorthand, we will denote the entropic solution at covariate value $x$ as $\bm{p}^\eta(x) \equiv \bm{p}^\eta(\bm{b}(x), \bm{c}(x))$.
We then define the \emph{conditional entropic regularized lower bound} as the value using the entropic solutions, $\theta^\eta_L(\bm{b}(x), \bm{c}(x)) \equiv \<\bm{c}(x), \bm{p}^\eta(\bm{b}(x), \bm{c}(x))\>$, with shorthand $\theta_L^\eta(x) \equiv \<\bm{c}(x), \bm{p}^\eta(x)\>$, and the aggregated bound as its expectation,
$\theta_L^\eta \equiv \E\left[\theta_L^\eta(X)\right]$.
Note that this lower bound only incorporates regularization through the solution $\bm{p}^\eta(X)$, and does not include an additional regularization term in the bound itself.

The hyperparameter $\eta > 0$ controls the level of regularization. As $\eta \to 0$, the entropy penalty in Eqn~\eqref{eq:condl_entropic_primal} dominates the objective function, and so the entropic solution $\bm{p}^\eta(X)$  will converge to the maximum entropy distribution that satisfies the constraints $A\bm{p} = \bm{b}(X)$. This corresponds to point-identifying the estimand $\theta$ by making a maximum entropy assumption; e.g. in the potential outcomes setting in Appendix~\ref{sec:examples}, this assumes that the potential outcomes are mutually independent given the covariates.
On the other hand, as $\eta \to \infty$, the entropic solution will converge to the solution of the unregularized program with maximum entropy \citep{alvarez_hessian_2004,muller_fisher-rao_2025}.
Throughout, we will assume a sufficient condition for strong duality and regularity of the entropic solutions to hold.

\begin{assumption}[Strict feasibility]
  \label{a:strict_feasibility}
  There exists a $\delta > 0$ such that for almost all $x \in \mathcal{X}$, there exists a strictly feasible $\bm{p} \in \mathcal{P}(x)$ such that $A\bm{p} = \bm{b}(x)$ and $\bm{p}_k \geq \delta $ for all $k=1,\ldots,K$, and a strictly feasible $\bm{p} \in \widehat{\mathcal{P}}(x)$ such that $A\bm{p} = \hat{\bm{b}}(x)$ and $\bm{p}_k \geq \delta$ for all $k=1,\ldots, K$.
\end{assumption}
\noindent
In the potential outcomes setting in Appendix~\ref{sec:examples}, a sufficient condition for Assumption~\ref{a:strict_feasibility} is that there is at least a $\delta$ probability of observing any outcome given covariates and treatment and the estimated probabilities respect this as well.

This regularization approach is related to smoothing approaches that have been considered elsewhere in the literature. For example, \citet{levis_assisted_2023} consider a smoothed version of the intersection bounds from \citet{semenova_aggregated_2024} that uses the smooth soft-max function in place of the non-smooth max function.
\citet{whitehouse_inference_2026} instead use a log-sum-exp approximation that is equivalent to entropic regularization on the corresponding linear program where the only constraint is that the primal values sum to one (and so does not need to be estimated).
Separately, in the aggregate, non-covariate assisted case, \citet{gafarov_simple_2025} considers including an $\ell_2$ penalty in the LP that also smooths the problem and leads to a unique solution.

One benefit of entropic regularization is that for certain problems entropic regularization can lead to significant computational speedups by using the  \citet{sinkhorn_diagonal_1967} algorithm \citep{cuturi_sinkhorn_2013}.
Another key benefit is that the regularizer is strongly convex, so there exists a unique solution for all values of the covariates $x \in \mathcal{X}$.
In addition, entropic regularization guarantees non-negativity, leaving only the linear equality constraints. This means that the unique solution to the entropic regularized program is a smooth function of the nuisance parameters $\bm{b}(x)$ and $\bm{c}(x)$,
removing the ambiguity of multiple solutions and the non-smoothness issues related to selecting the optimal basis.
Therefore, we can construct a de-biased estimator using standard semi-parametric theory
by computing the Jacobian of the entropic regularized solutions with respect to the vectors $\bm{b}(x)$ and $\bm{c}(x)$
and using a Taylor expansion of the true solution around the plugin estimate $\hat{\bm{p}}^\eta(x) \equiv \bm{p}^\eta(\hat{\bm{b}}(x), \hat{\bm{c}}(x))$.
Then, we can construct a de-biased estimate as
\begin{equation}
  \label{eq:entropic_theta_est}
  \begin{aligned}
    \hat{\theta}_L^\eta &= \frac{1}{n}\sum_{i=1}^n \hat{\theta}^\eta_L(X_i) + \<\nabla_{\bm{b}} \hat{\theta}_L^\eta(X_i) ,  \hat{\bm{\varphi}}^{(b)}(O_i) \> + \<\nabla_{\bm{c}} \hat{\theta}^\eta_L(X_i), \hat{\bm{\varphi}}^{(c)}(O_i)\>,
  \end{aligned}
\end{equation}
where $\hat{\theta}^\eta_L(x) = \theta^\eta_L(\hat{\bm{b}}(x), \hat{\bm{c}}(x))$, and $\nabla_{\bm{b}} \hat{\theta}_L^\eta(x) $ and $\nabla_{\bm{c}} \hat{\theta}_L^\eta(x)$ are shorthand for the Jacobians of the conditional entropic regularized lower bound with respect to the constraint and objective vectors evaluated at the plugin estimates; we will use $\nabla_{\bm{b}} \theta_L^\eta(x) $ and $\nabla_{\bm{c}} \theta_L^\eta(x)$ to denote the same gradients evaluated at the true constraint and objective vectors.
Lemma~\ref{lem:jacobian} in the Appendix gives the explicit form of these Jacobians.

\paragraph{Convergence rate and inference.}
We first show that the asymptotic bias of $\hat{\theta}_L^\eta$ relative to the true entropic regularized bound $\theta_L^\eta$ behaves as we would expect, decreasing with the de-biased rate $r_{n}$ and second order errors in the nuisance functions $\hat{\bm{b}}$ and $\hat{\bm{c}}$, and so there is no first-order dependence on the nuisance functions.
\begin{theorem}
  \label{thm:entropic_est_rate}
  Define
   $\tilde{\theta}_L^\eta \equiv \frac{1}{n}\sum_{i=1}^n \theta^\eta_L(X_i) + \<\nabla_{\bm{b}} \theta_L^\eta(X_i) ,  \bm{\varphi}^{(b)}(O_i) \> + \<\nabla_{\bm{c}} \theta^\eta_L(X_i), \bm{\varphi}^{(c)}(O_i)\>.$
    Under Assumptions~\ref{a:riesz}, \ref{a:rates}, and \ref{a:strict_feasibility}, for a fixed $\eta >0$,
  \begin{align*}
     \hat{\theta}^\eta_L - \theta^\eta_L &= \tilde{\theta}^\eta_L - \theta^\eta_L  +  O_p\left(r_n + \|\hat{\bm{b}} - \bm{b}\|_2^2 + \|\hat{\bm{c}} - \bm{c}\|_2^2\right) + o_p(n^{-1/2}).
  \end{align*}
\end{theorem}
\noindent Of course, changing the estimand by including the entropy penalty will introduce approximation error relative to the original bounds because $\theta_L^\eta \geq \theta_L$.
The following result incorporates the approximation error in a setting where we choose the hyperparameter $\eta$ to grow with the sample size.
A key component of the approximation error relates to the size of the hyperparameter $\eta$ relative to the \emph{primal sub-optimality gap} $\Delta_p(x) \equiv \min_{B \in \mathcal{B}(x) \setminus \mathcal{B}^\ast(x)} \<\bm{c}(x), A_B^{-1} \bm{b}(x)\> - \theta_L(x)$, where we define $\Delta_p(x) = 0$ if $\mathcal{B}(x) = \mathcal{B}^\ast(x)$, i.e. all feasible bases are optimal.
This measures the gap between the minimal value $\theta_L(x)$ and the value at a next-best basic feasible solution.
We make a form of margin condition on this gap, limiting the number of cases where there exist very close next-best values; in principle, this margin parameter can be different from the margin parameters in Assumptions~\ref{a:infeasible_gap} and \ref{a:primal_infeasible_gap}.
\begin{assumption}[Primal sub-optimality margin]
  \label{a:primal_margin}
  For some $\alpha > 0$, $P(0 < \Delta_p(X) \leq t) \lesssim t^\alpha $ for all $t > 0$, where $\alpha = \infty$ denotes that there exists some $\Delta_0 > 0$ such that $P(0 < \Delta_p(X) < \Delta_0) = 0$.
\end{assumption}
In the binary potential outcome setting we consider in Appendix~\ref{sec:margin_ex}, this margin condition limits how close the conditional probability of each potential outcome level can be to zero; Assumption~\ref{a:strict_feasibility} of strict feasibility implies the potential outcomes are bounded away from being deterministic so in this special case the margin parameter $\alpha = \infty$.

\begin{theorem}
  \label{thm:entropic_est_approx}
    Define $\rho_\alpha(\eta) = \eta^{-(1 + \alpha)}$ for $\alpha < \infty$ and $\rho_\infty(\eta) = e^{-C\eta}$, where $C$ is a constant that does not depend on $n$ or $\eta$.
    Under Assumptions~\ref{a:riesz}, \ref{a:rates}, \ref{a:strict_feasibility}, and \ref{a:primal_margin}, if
    for some $0 < q \leq 1$, $\sup_{\eta > 0}\E\left[\sigma_\text{min}\left(A \text{diag}(\bm{p}^\eta(X)) A'\right)^{-q}\right] \lesssim 1$,   where $\sigma_\text{min}(\cdot)$ denotes the minimum singular value, and $\eta\left(\|\hat{\bm{b}} - \bm{b}\|_2  + \|\hat{\bm{c}} - \bm{c}\|_2\right)= o_p(1)$,  then
    \[
      \hat{\theta}^\eta_L - \theta_L = \tilde{\theta}^\eta_L - \theta^\eta_L  +  O_p\left(\rho_\alpha(\eta) + r_n + \eta \|\hat{\bm{b}} - \bm{b}\|_2\|\hat{\bm{c}} - \bm{c}\|_2 + \eta \|\hat{\bm{c}} - \bm{c}\|_2^2  + \|\hat{\bm{b}} - \bm{b}\|_\infty^{1+q} \right) + o_p\left(n^{-1/2}\right).
    \]
\end{theorem}

\noindent Theorem~\ref{thm:entropic_est_approx} includes the approximation error induced by regularization.
If the overall level of regularization $\eta^{-1}$ is small relative to the size of the sub-optimality gaps $\Delta_p(x)$, \citet{weed_explicit_2018} shows that the approximation error induced via entropic regularization decreases exponentially with $\eta$; otherwise, the approximation error decreases linearly with the level of regularization $\eta^{-1}$. Under the margin condition in Assumption~\ref{a:primal_margin} with $\alpha < \infty$, $\eta^{-1}$ will be too large with at most probability $\eta^{-\alpha}$, leading to an overall approximation error of $\eta^{-(1+\alpha)}$.
Under a ``hard margin'' setting where the sub-optimality gap is lower-bounded, the level of regularization $\eta^{-1}$ will eventually be small enough to achieve an exponentially decaying approximation error.

On the other hand, as the level of regularization $\eta^{-1}$ decreases, the entropic solutions become similar to the non-smooth solutions to the unregularized program and the Hessians of the entropic solutions with respect to the nuisance functions $\bm{b}(x)$ and $\bm{c}(x)$ increase. This increases the scale of the bias and variance due to error in the nuisance functions.
The condition that $\eta\left(\|\hat{\bm{b}} - \bm{b}\|_2  + \|\hat{\bm{c}} - \bm{c}\|_2\right)= o_p(1)$ ensures that this increase in scale does not make the variance diverge, while the impact on the bias terms is included directly in the rate.
A particular concern is the Hessian with respect to the constraints $\bm{b}(x)$: in settings where there is a unique, degenerate primal solution to the unregularized program the second derivatives can grow exponentially as the level of regularization decreases. The additional regularity condition in Theorem~\ref{thm:entropic_est_approx} rules out this case by requiring the entropic optimization problem to be strongly convex as $\eta$ grows. Note that this allows for cases where the unregularized program has multiple, possibly degenerate solutions, as long as the columns of $A$ corresponding to non-zero primal variables has full column rank.
This changes the rate with respect to estimation error in the constraints $\|\hat{\bm{b}} - \bm{b}\|_\infty^{1 + q}$. If $q$ is close to 1 so that the Hessian is well behaved, the rate will be close to second order in the estimation error; conversely, if $q$ is close to zero and the Hessian is poorly behaved then the rate will be close to first order in the error.
Other rates involving the $L^2$ norm error are possible at the expense of stricter conditions on $q$.

Finally, we can use Theorem~\ref{thm:entropic_est_approx} to conduct asymptotic inference when the level of regularization $\eta^{-1}$  and the estimation error for the nuisance functions shrink.
\begin{assumption}[Rate and regularity conditions for entropic estimator]
  \label{a:rate_entropic}
  \mbox{}
  \begin{enumerate}[label=(\roman*)]
    \item \label{a:rate_entropic_i} $\eta = o(n^{1/2})$, $\rho_\alpha(\eta)\sqrt{n}\to 0$, $\|\hat{\bm{b}} - \bm{b}\|_\infty = o_p\left(n^{-\frac{1}{2(1+q)}}\right)$, $\|\hat{\bm{b}} - \bm{b}\|_2\|\hat{\bm{c}} - \bm{c}\|_2 = o_p\left(\frac{1}{n^{1/2}\eta}\right)$, $\|\hat{\bm{c}} - \bm{c}\|_2  = o_p\left(\frac{1}{n^{1/4} \eta^{1/2}}\right)$ , and $r_{n} = o\left(\frac{1}{n^{1/2}}\right)$.
    \item \label{a:rate_entropic_ii} Define $g^\eta(O) \equiv \theta^\eta_L(X) + \<\nabla_{\bm{b}} \theta_L^\eta(X) ,  \bm{\varphi}^{(b)}(O) \> + \<\nabla_{\bm{c}} \theta^\eta_L(X), \bm{\varphi}^{(c)}(O)\>$. $V^\eta \equiv \Var(g^\eta(O)) > 0$ and for every $\varepsilon > 0$, $\lim_{n \to \infty}\E\left[\left(g^\eta(O) - \theta_L^\eta\right)^2 \bbone\{|g^\eta(O) - \theta_L^\eta| \geq \varepsilon \sqrt{n V^\eta}\}\right] = 0$.
  \end{enumerate}
\end{assumption}

Assumption~\ref{a:rate_entropic}\ref{a:rate_entropic_i} summarizes sufficient conditions for this increase in magnitude of the bias to be outweighed by the decrease in approximation error so that the overall bias is asymptotically smaller than the standard error and will not affect our confidence intervals.
The faster $\eta$ grows, the more quickly the nuisance function errors need to converge.
Therefore, for $\alpha < \infty$ the growth rate for $\eta$ must be between $n^\frac{1}{2(1 + \alpha)}$ and $n^{1/2}$, where the $n^{1/2}$ rate corresponds to estimating the nuisance functions at a parametric $n^{-1/2}$ rate; in the hard margin case where $\alpha = \infty$, $\eta$ can grow logarithmically with $n$.
With this, and the Lindeberg condition in Assumption~\ref{a:rate_entropic}\ref{a:rate_entropic_ii}, we can conduct inference for the unregularized bound $\theta_L$ using the de-biased entropic estimator $\hat{\theta}_L^\eta$.
\begin{corollary}
  \label{cor:entropic_est_normal}
  Under the conditions of Theorem~\ref{thm:entropic_est_approx}  and Assumption \ref{a:rate_entropic}, $\frac{\sqrt{n}\left(\hat{\theta}_L^\eta - \theta_L\right)}{\sqrt{V^\eta}} \Rightarrow N(0, 1)$.
\end{corollary}
\noindent We can construct the confidence bounds by first estimating the variances with
\[
    \hat{V}_L^\eta \equiv \frac{1}{n}\sum_{i=1}^n \left(\hat{\theta}^\eta_L(X_i) + \<\nabla_{\bm{b}} \hat{\theta}_L^\eta(X_i) ,  \hat{\bm{\varphi}}^{(b)}(O_i) \> + \<\nabla_{\bm{c}} \hat{\theta}^\eta_L(X_i), \hat{\bm{\varphi}}^{(c)}(O_i)\> - \hat{\theta}_L^\eta\right)^2,
\]
then constructing $1-\alpha$ level confidence bounds as described at the end of Section~\ref{sec:primal_estimation}.

In Appendix~\ref{sec:sim} we inspect the impact of different growth rates for $\eta$ via simulation and find that the de-biased entropic estimator is not particularly sensitive to the choice of hyper-parameter.
A baseline heuristic that we suggest in practice and use empirically in Section~\ref{sec:ohie} is to compute $\hat{\theta}_L^\eta$ for a range of values of $\eta$,  and choose a value of $\eta$ that is large enough so that the estimate stabilizes. This value can also be compared to the de-biased primal-dual estimate $\hat{\theta}_L$ to ensure that the difference is not substantively meaningful (e.g. an order of magnitude or two smaller than the estimates).

\section{Decision making with de-biased bound estimators}
\label{sec:policy_learning}

We now turn to applying the de-biased primal-dual and entropic estimators to the problem of decision making, where the goal is to find a decision rule that optimizes a parameter that is partially identified via conditional LPs.
In particular, consider a decision rule $\pi: \mathcal{X} \to [0,1]$ that maps covariates to a continuous action between 0 and 1 (e.g. a treatment probability).
We will focus on settings where the covariate-conditional objective function is given by $\bm{c}(x, \pi(x))$, so that the conditional lower bound is
$\theta_L(x, \pi) = \min_{\bm{p} \in \mathcal{P}(x)} \<\bm{c}(x, \pi(x)), \bm{p}\>,$
and the overall expected lower bound on the objective value for a given decision rule $\pi$ is given by $\theta_L(\pi) = \E[\theta_L(X, \pi)]$.

Our goal is to find a decision rule that maximizes this lower bound across all decision rules $\pi$ in a policy class $\Pi$, i.e. $\pi^\ast \in  \argmax_{\pi \in \Pi} \ \theta_L(\pi)$. This setup accommodates a wide variety of decision making problems under uncertainty \citep{Manski2011} including maximin value and  minimax regret rules.\footnote{For minimax regret rules, we can maximize the minimum negative regret relative to the unconstrained oracle policy (see Example~\ref{ex:policy_learning} in Appendix~\ref{sec:examples}). However, this does not directly accommodate the problem of finding a minimax regret rule relative to the best-in-class policy. We leave this problem for future work.}
In this section we will consider finding such policies by maximizing an estimate of the lower bound, using either the de-biased primal-dual or entropic estimators from the previous section.
Throughout, we will focus on settings where the objective function $\bm{c}(x, \pi(x))$ is a known function of the covariates $x$ and the decision rule $\pi(x)$ and we will make several regularity assumptions.
\begin{assumption}[Regularity conditions]
  \label{a:regularity}
  For any $x \in \mathcal{X}$, $d \in [0,1]$, and $o \in \mathcal{O}$, the constraint and objective functions are bounded, i.e. $\|\bm{b}(x)\|_\infty \leq B_b$ and $\|\bm{c}(x,d)\|_\infty \leq B_c$ for some constants $B_b, B_c < \infty$, as is the de-biasing function for the constraints, i.e. $\|\bm{\varphi}^{(b)}(o)\|_\infty \leq B_\varphi$ for some constant $B_\varphi < \infty$. The estimates also satisfy these bounds. Furthermore, the objective function $\bm{c}(x, \cdot)$ is Lipschitz in the decision $\pi(x)$.
\end{assumption}
As a final piece of setup, we will characterize the complexity of the policy class $\Pi$ via the Rademacher complexity
$\mathcal{R}_n(\Pi) = \E\left[\sup_{\pi \in \Pi} \left|\frac{1}{n}\sum_{i=1}^n \pi(X_i) \varepsilon_i\right|\right],$
where $\varepsilon_i \in \{-1,1\}$ are independent Rademacher random variables. For policy classes $\Pi$ with finite VC dimension $\nu$, $\mathcal{R}_n(\Pi) \lesssim \sqrt{\nu / n}$ \citep{wainwright_2019}.
\subsection{Policy learning with the de-biased primal-dual estimator}
\label{sec:policy_learning_primal}

For a policy $\pi$, we can estimate the conditional bound using the estimated constraint vector $\hat{\bm{b}}(x)$, to get $\hat{\theta}_L(x, \pi(x))$, with estimated primal variables $\hat{\bm{p}}(x, \pi(x)) \in \argmin_{\bm{p} \in \widehat{\mathcal{P}}(x)} \<\bm{c}(x, \pi(x)), \bm{p}\>$ and dual variables $\hat{\bm{\lambda}}(x, \pi(x)) \in \argmax_{\bm{\lambda} \in \Lambda(x, \pi(x))} \<\bm{\lambda}, \hat{\bm{b}}(x)\>$,
where $\Lambda(x, \pi(x)) = \{\bm{\lambda} \in \R^J \mid   A'\bm{\lambda} \leq \bm{c}(x, \pi(x))\}$ is the dual feasible set for a given policy $\pi$.
Then, we estimate a policy $\hat{\pi}$ by maximizing the de-biased primal-dual   estimator of the lower bound on the expected objective value:
\begin{equation}
  \label{eq:primal_policy_est}
  \begin{aligned}
    \hat{\pi} \in \underset{\pi \in \Pi}{\argmax}  & \frac{1}{n}\sum_{i=1}^n \hat{\theta}_L(X_i, \pi(X_i))+  \left\<\hat{\bm{\lambda}}(X_i, \pi(X_i)), \hat{\bm{\varphi}}^{(b)}(O_i)\right\>.
  \end{aligned}
\end{equation}
To evaluate the quality of the estimated policy $\hat{\pi}$, we compare it to the best-in-class policy that maximizes the true lower bound on the expected objective value, i.e. $\pi^\ast \in \argmax_{\pi \in \Pi} \theta_L(\pi)$, using the \emph{excess regret}, $\theta_L(\pi^\ast) - \theta_L(\hat{\pi})$.
Controlling this excess regret involves (i) linking the excess regret to the bias of the de-biased primal-dual estimator, and (ii) applying standard results on empirical risk minimization.

To do so, note that in this setting the objective vector is known, and so we can always be sure to have dual-feasible, and thus valid, solutions.
To control the excess regret, we will also need to control the conservativeness of the lower bound; we do so
with a stronger version of the optimality detection margin condition in Assumption~\ref{a:primal_infeasible_gap}.
To describe this assumption, we define $\mathcal{D}(x, \pi(x)) \equiv \{B \in \mathcal{B}  \mid A'A_B^{-1\prime} \bm{c}(x, \pi(x)) \leq \bm{c}(x, \pi(x))\}$ and $\mathcal{D}^{-}(x, \pi(x)) = \{B \in \mathcal{D}(x, \pi(x)) \mid \<\bm{c}(x, \pi(x)), A^{-1}_B \bm{b}(x)\> < \theta_L(x, \pi(x))\}$ as the set of dual-feasible and dual-feasible but not optimal bases for a given covariate value $x$ and decision $\pi(x)$. We then define the optimality detection gap for a covariate-decision pair as $\mathcal{E}_p(x, \pi(x)) = \min_{B \in \mathcal{D}^{-}(x, \pi(x))} \|[A_B^{-1} \bm{b}(x)]_-\|_\infty$.
We then extend Assumption~\ref{a:primal_infeasible_gap} to hold over all possible policies $\pi \in \Pi$.
\begin{assumption}[Strong optimality detection margin]
  \label{a:primal_infeasible_gap_strong}
  For some $\alpha > 0$, $P(\mathcal{E}_p(X, \pi(X)) \leq t) \lesssim t^\alpha $ for all $t > 0$ and all $\pi \in \Pi$.
\end{assumption}
\noindent This is primarily an additional requirement on the form of the objective function $\bm{c}(x, \pi(x))$.
It restricts the number of cases where the policy value $\pi(x)$ can change the dual feasibility set enough so that there is a nearly feasible primal solution.
A sufficient condition is that a margin condition holds for the optimality detection gap when considering all primal-infeasible bases, which does not depend on $\pi$, rather than only the dual feasible but suboptimal ones.

\begin{theorem}
  \label{thm:primal_policy}
  Under Assumptions~\ref{a:riesz}, \ref{a:rates}, \ref{a:regularity}, and \ref{a:primal_infeasible_gap_strong}, if for (almost) all $x \in \mathcal{X}$ and any $d \in [0,1]$ $\hat{\bm{p}}(x,d)$ is non-degenerate, then the excess regret of the estimated policy $\hat{\pi}$ satisfies
  $\theta_L(\pi^\ast) - \theta_L(\hat{\pi}) = O_p\left(\mathcal{R}_n(\Pi) +\|\hat{\bm{b}} - \bm{b}\|_\infty^{1 + \alpha} + r_n \right)$.
\end{theorem}
Theorem~\ref{thm:primal_policy} shows that the excess regret of the estimated policy $\hat{\pi}$ inherits most of the properties of the de-biased primal-dual estimator in Theorem~\ref{thm:direct_est_rate}, with an additional Rademacher complexity term that is standard for empirical risk minimization problems and controllable by the analyst by choosing a simpler or more expressive policy class.
Comparing to policy learning settings with point identification \citep{Athey2021}, there is an additional term in the excess regret bound arising from error in the estimated constraint vector $\hat{\bm{b}}(x)$, with the impact of the error lessened by the margin parameter $\alpha$.
Theorem~\ref{thm:primal_policy} requires that the estimated primal values be non-degenerate, leading to a unique estimated dual solution and ensuring that the estimate of the value is Lipschitz in the objective function. Note that this does not require the true primal values to be non-degenerate, so it can be viewed as a stability condition on the estimator that can be enforced by perturbing the nuisance function estimates if needed.

\subsection{Policy learning with the entropic estimator}
\label{sec:policy_learning_entropic}

An alternative is to use the entropic estimator to estimate the lower bound on the expected objective value for a policy $\pi$, and then find the policy that maximizes this estimate.
Fixing a regularization hyperparameter $\eta$, and denoting $\hat{\bm{p}}^\eta(x, \pi(x)) \equiv \bm{p}^\eta(\hat{\bm{b}}(x), \bm{c}(x, \pi(x)))$ as the solution to the entropic conditional LP with policy value $\pi(x)$ and $\hat{\theta}_L^\eta(x, \pi(x)) = \<\bm{c}(x, \pi(x)), \hat{\bm{p}}^\eta(x, \pi(x))\>$ as the lower bound estimate, we can find a policy $\hat{\pi}^\eta$ by solving
\begin{equation}
  \label{eq:entropic_bound_est}
  \begin{aligned}
    \hat{\pi}^\eta \in \underset{\pi \in \Pi}{\argmax} & \frac{1}{n}\sum_{i=1}^n \hat{\theta}_L^\eta(X_i, \pi(X_i)) + \left\< \nabla_{\bm{b}} \hat{\theta}_L^\eta(X_i, \pi(X_i)), \hat{\bm{\varphi}}^{(b)}(O_i) \right\>,
  \end{aligned}
\end{equation}
where again we use $\nabla_{\bm{b}} \hat{\theta}_L^\eta(X_i, \pi(X_i))$ as shorthand for the gradient of $\theta_L^\eta(\bm{b}, \bm{c})$ with respect to $\bm{b}$, evaluated at the plugin estimates of the constraints and objectives.
In contrast with the primal-dual approach in Section~\ref{sec:policy_learning_primal}, the entropic estimator of the lower bound $\hat{\theta}_L^\eta(\pi)$ is differentiable with respect to the policy value $\pi(x)$, so we can optimize it with standard gradient-based methods. Note, however, that the overall objective is non-convex and each objective value and gradient computation requires solving the entropic conditional LP for each unique value of the covariates $x$ in the sample, which can be computationally expensive for iterative optimization methods.

We will compare this estimated decision policy $\hat{\pi}^\eta$ to two notions of an optimal policy: (i) the best-in-class policy that maximizes the true lower bound on the expected objective value, $\pi^\ast \in \argmax_{\pi \in \Pi} \ \theta_L(\pi)$, and (ii) the best-in-class policy that maximizes the entropic lower bound on the expected objective value, $\pi^\ast_\eta \in \argmax_{\pi \in \Pi} \ \theta_L^\eta(\pi)$.
To do so, we will extend the definition of the primal sub-optimality gap to include the policy value $\pi(x)$, $\Delta_p(x, \pi(x)) \equiv \min_{B \in \mathcal{B}(x) \setminus \mathcal{B}^\ast(x, \pi(x))} \<\bm{c}(x, \pi(x)), A_B^{-1} \bm{b}(x)\> - \theta_L(x, \pi(x))$ where $\mathcal{B}^\ast(x, \pi(X))$ is the set of optimal basic feasible solutions, and again we define $\Delta_p(x, \pi(x)) = 0$ if $\mathcal{B}(x) = \mathcal{B}^\ast(x, \pi(x))$, i.e. all feasible solutions are optimal for a given policy value $\pi(x)$.
We then define a version of the primal sub-optimality margin Assumption~\ref{a:primal_margin}  that holds for all candidate policies $\pi \in \Pi$.
\begin{assumption}[Strong primal sub-optimality margin]
  \label{a:primal_margin_strong}
  For some $\alpha > 0$, $P(0 < \Delta_p(X, \pi(X)) \leq t) \lesssim t^\alpha $ for all $t > 0$ and $\pi \in \Pi$, where $\alpha = \infty$ denotes that there exists some $\Delta_0 > 0$ such that $P(0 < \Delta_p(X, \pi(X)) < \Delta_0) = 0$ for all $\pi \in \Pi$.
\end{assumption}
\begin{theorem}
  \label{thm:entropic_policy}
  Under Assumptions~\ref{a:riesz}, \ref{a:rates}, \ref{a:strict_feasibility},  \ref{a:regularity}, and \ref{a:primal_margin_strong}, for a fixed $\eta$ the excess regret of the estimated policy $\hat{\pi}^\eta$  relative to $\pi_\eta^\ast$ measured via the entropic lower bound satisfies
  \[
  \theta_L^\eta(\pi^\ast_\eta) - \theta_L^\eta(\hat{\pi}^\eta) = O_p\left(\mathcal{R}_n(\Pi) + r_n + \|\hat{\bm{b}} - \bm{b}\|_2^2\right).
  \]
  Furthermore,
  if $\sup_{\pi \in \Pi}\sup_{\eta > 0}\E\left[\sigma_\text{min}\left(A \text{diag}(\bm{p}^\eta(\bm{b}(X), \bm{c}(X, \pi(X)))) A'\right)^{-q}\right] \lesssim 1$ for some $0 < q \leq 1$ and
  $\eta$ is allowed to grow,
  then the excess regret of the estimated policy $\hat{\pi}^\eta$ relative to $\pi^\ast$ measured via the unregularized lower bound satisfies
  \[
    \theta_L(\pi^\ast) - \theta_L(\hat{\pi}^\eta) = O_p\left(\rho_\alpha(\eta) + \eta \times \mathcal{R}_n(\Pi)  + r_n +  \|\hat{\bm{b}} - \bm{b}\|_\infty^{1 + q}\right),
  \] where recall that  $\rho_\alpha(\eta) = \eta^{-(1 + \alpha)}$ for $\alpha < \infty$ and $\rho_\infty(\eta) = e^{-C\eta}$.
\end{theorem}
Theorem~\ref{thm:entropic_policy} shows that relative to the optimal decision policy for the true entropic lower bound, $\pi^\ast_\eta$, the excess regret of the estimated policy $\hat{\pi}^\eta$ is impacted by estimation errors in the constraint vector $\hat{\bm{b}}(x)$ only via the de-biased rate $r_{n}$ and squared errors $\|\hat{\bm{b}} - \bm{b}\|_2^2$. This means that the estimated policy is relatively robust to such errors and we would typically expect the complexity of the policy class $\Pi$ to drive the excess regret.
The same is true for the excess regret relative to the optimal policy for the true lower bound, $\pi^\ast$, along with the additional scaling with the hyperparameter $\eta$ as with Theorem~\ref{thm:entropic_est_approx} under the additional regularity condition that the singular value moment bound holds uniformly over $\pi \in \Pi$.

\section{Example: impact of Medicaid access or enrollment}
\label{sec:ohie}
In this section we apply the general procedure outlined in the previous sections to measure the impact of Medicaid access and enrollment on emergency department visits \citep{finkelstein_oregon_2012, taubman_medicaid_2014}.
In Appendix~\ref{sec:sim} we also conduct a simulation study to understand the finite sample behavior of the estimators for both point- and partially-identified parameters.

This natural experiment occurred in 2008 when the state of Oregon offered a Medicaid program to individuals by randomly selecting names off a waiting list. Here, being chosen off the waitlist is an instrument $Z$, while enrolling in Medicaid is the treatment $D$. We will focus on the number of primary care-treatable emergency department (ED) visits as the outcome $Y$. This is a modeled, continuous outcome; we round it to correspond to counts and top-code at 10 visits\footnote{Only 22 cases had more than 10 visits, comprising 0.17\% of the sample. 10 cases have missing values.}.
Our dataset consists of the 13,009 individuals who lived in areas served by one of 12 hospitals for which the authors received ED records, were sampled to be part of an in-person follow-up survey, and had no other members in the household.
For this group, the probability of being selected off the waitlist was uniform across individuals (with 48\% selected).
We include basic demographic variables, pre-period measures of health and health-care utilization, and pre-period enrollment in public assistance programs in our covariate vector $X$. Here, strong ignorability of $Z$ is satisfied by design; we also assume the exclusion restriction that being selected from the waitlist does not affect the outcome except through Medicaid enrollment.

\paragraph{Impacts of Medicaid access.}
First, we focus on the impacts of access to Medicaid by being chosen off the waitlist ($Z = 1$) versus not ($Z = 0$).
We compare the observed randomization procedure to an oracle policy that could observe each individual's potential outcomes (i.e. $Y(D(0))$ and $Y(D(1))$) and can select individuals from the waitlist in order to minimize the number of PC-treatable ED visits (i.e. when $Y(D(1)) \leq Y(D(0))$).
We do this by estimating bounds on (i) the number of additional PC-treatable ED visits under the observed distribution relative to the counterfactually optimal policy, $\E[Y] - \E[\min_z Y(D(z))]$, (ii) the proportion of individuals who experience any additional PC-treatable ED visits under the observed distribution relative to the oracle optimal policy---$P(Y \neq \min_{z} Y(D(z)))$, and (iii) the regret under powerlaw value functions, defined below.
We estimate the nuisance functions using multinomial gradient boosted decision trees and use 3-fold cross-fitting.
To make the hyperparameter $\eta$ comparable across estimands, we scale the objective of each LP to have a maximum value of 1 i.e. $\|\bm{c}\|_\infty = 1$.

Recall that computing the estimates requires solving an optimization problem for each observed covariate value $X_i$ in the sample. In this case each covariate vector is unique so this becomes 13,009 optimization problems with $K = 121$ decision variables and $J=21$ constraints.
However, each of these optimization problems is quick to solve, and they can be solved in parallel across the sample.
Parallel processing on 12 cores on a 2023 Mac Studio with an M2 Max chip and 64GB of RAM, computing the de-biased primal-dual estimator took 1.63 seconds and computing the de-biased entropic estimate with $\eta = 10$ took 1.37 seconds. In contrast, cross-fitting the nuisance functions using gradient boosted decision trees took an order of magnitude longer at 72 seconds.
Thus, for large problems, flexibly estimating the nuisance functions is likely to be the (unavoidable) computational bottleneck, rather than solving the conditional LPs.

\begin{figure}[t]
  \centering
  \begin{subfigure}[t]{0.45\textwidth}
    \includegraphics[width=0.7\textwidth]{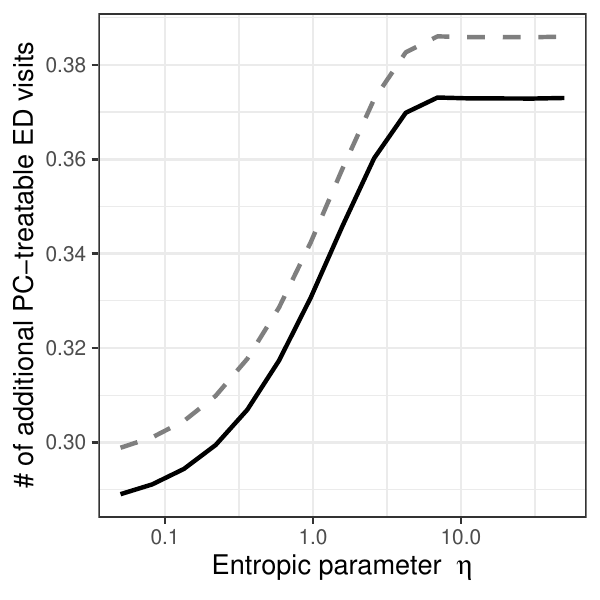}
  \end{subfigure}
  \begin{subfigure}[t]{0.45\textwidth}
    \includegraphics[width=0.7\textwidth]{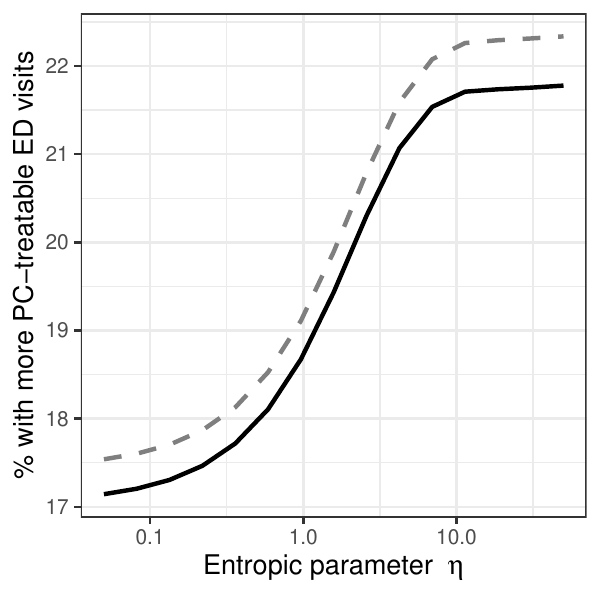}
  \end{subfigure}
  \caption{Estimated upper bounds on the expected number of additional PC-treatable ED visits and the percentage of individuals with any more PC-treatable ED visits under the randomization policy versus the ED-minimizing waitlist policy, as the entropic hyperparameter $\eta$ increases. Dashed lines correspond to upper 95\% confidence bounds.}
  \label{fig:ent_penalty}
\end{figure}

Figure~\ref{fig:ent_penalty} shows the estimated upper bounds and 95\% upper confidence bounds on the first two estimands
using the de-biased entropic estimator when increasing the entropic hyper-parameter $\eta$.
For both estimands, the bounds stabilize when $\eta \approx 10$.
We estimate an upper bound of an additional 0.37 PC-treatable ED visits per person (95\% bound: 0.386) under the randomization policy vs the ED-minimizing waitlist policy (relative to an observed average of $0.4 \pm 0.018$ visits per person).
We also estimate that at most 21.9\% of individuals experience any more PC-treatable ED visits under the randomization policy vs the ED-minimizing waitlist policy (95\% bound: 22.2\%).
The de-biased primal dual estimates are nearly identical at 0.374 and 21.8\%, respectively.

Next we consider decision-making, where we evaluate a policy based on the  power-law collective utility function \citep{moulin_axioms_1988},
\[
  \begin{aligned}
    V^\lambda(\pi) & \equiv \E\left[\frac{1}{\lambda}\left(\left\{u(Y(D(0))) + \pi(X) (u(Y(D(1))) - u(Y(D(0))))\right\}^\lambda - 1\right)\right],
  \end{aligned}
\]
where we set the utility of $Y$ PC-treatable ED visits to be the maximum number of PC-treatable ED visits minus $Y$, $u(Y) = 10 - Y$.
Here the value of the power-law parameter $\lambda$ determines the level of emphasis placed on individuals with lower or higher utility under $\pi$. When $\lambda = 1$, the value function is simply the expected individual utility under $\pi$.
Values of $\lambda < 1$ will lead to inequality-averse value functions.
When $\lambda = 0$, the  $V^0(\pi)$ is the logarithm of the geometric mean of the individual utilities, and as $\lambda \to -\infty$, $V^\lambda(\pi)$ goes towards prioritizing only the individuals with the smallest utilities under $\pi$; choosing $\pi$ to maximize this results in the \emph{egalitarian rule} that maximizes the minimum utility.
Finally, values of $\lambda > 1$ will lead to inequality-seeking value functions.
In Appendix~\ref{sec:examples} we show how to include the regret relative to the oracle policy in our framework.

Figure~\ref{fig:powerlaw} shows estimated upper bounds as the power parameter $\lambda$ changes.
We see that the upper bound on the regret is larger the closer the utility function is to the expected utility ($\lambda \to 1$), but decreases as $\lambda$ decreases to $-1$ and is more inequality-averse.
These bounds indicate that randomization incurs less potential regret for more inequality-averse utility functions, as we might expect, but that there is also substantial potential room for improvement on these policies at all but the lowest levels of inequality aversion.

\begin{figure}
  \centering
  \begin{subfigure}[t]{0.45\textwidth}
  \includegraphics[width=0.7\textwidth]{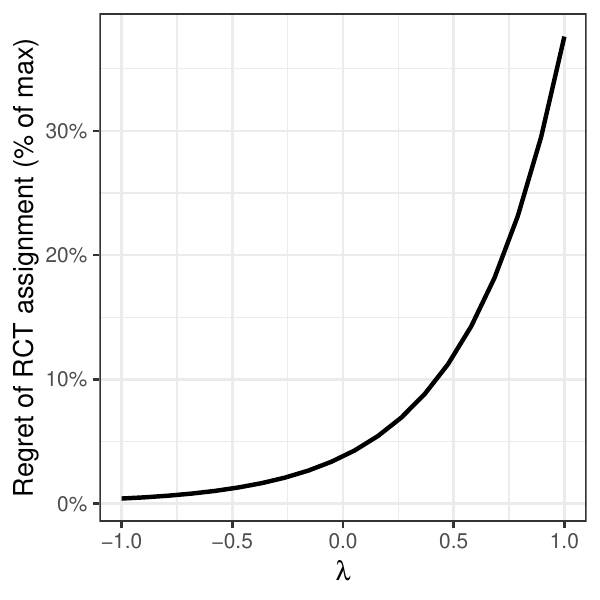}
  \caption{Regret of randomization policy.}
  \label{fig:powerlaw}
  \end{subfigure}
  \begin{subfigure}[t]{0.45\textwidth}
    \includegraphics[width=0.7\textwidth]{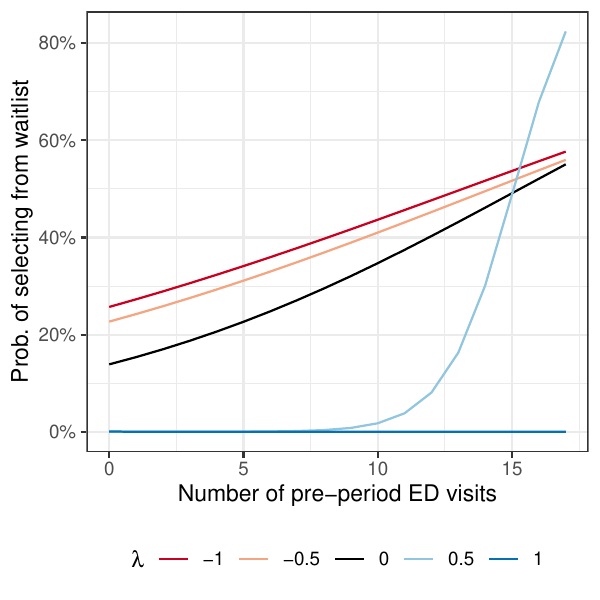}
  \caption{Minimax regret waitlist allocation rules.}
  \label{fig:policies}
  \end{subfigure}

  \caption{(a) Estimated upper bound on the regret of the randomization policy versus the ED-minimizing waitlist policy, using the power-law collective utility function as the power parameter $\lambda$ changes. Upper 95\% confidence bounds are not visible relative to the y-axis scale. (b) Minimax regret waitlist allocation rules as a function of the number of pre-period ED visits.}
\end{figure}

\begin{figure}[t]
  \centering
  \begin{subfigure}[t]{0.45\textwidth}
    \includegraphics[width=\textwidth]{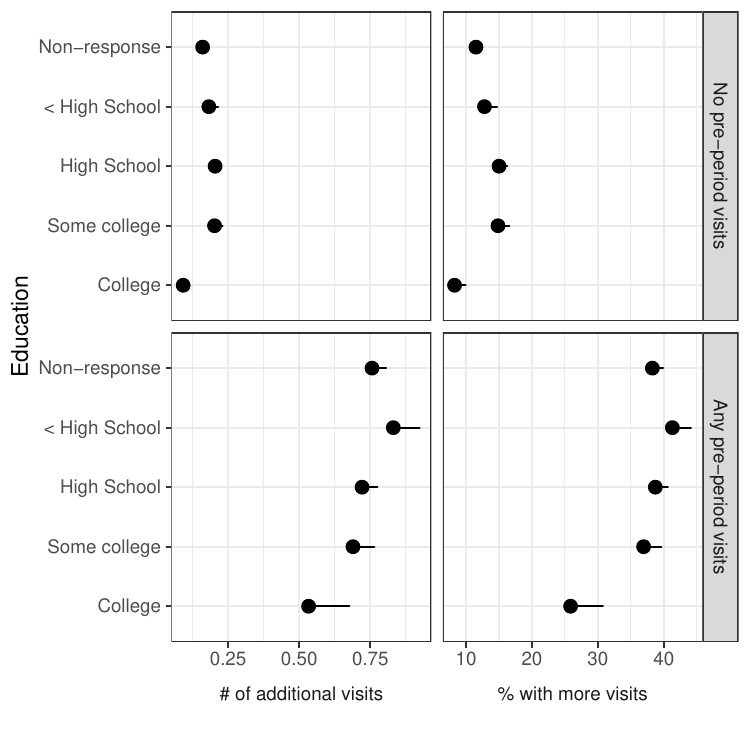}
    \caption{Randomizing access to Medicaid.}
    \label{fig:ed_het_plot_1}
  \end{subfigure}
  \begin{subfigure}[t]{0.45\textwidth}
    \includegraphics[width=\textwidth]{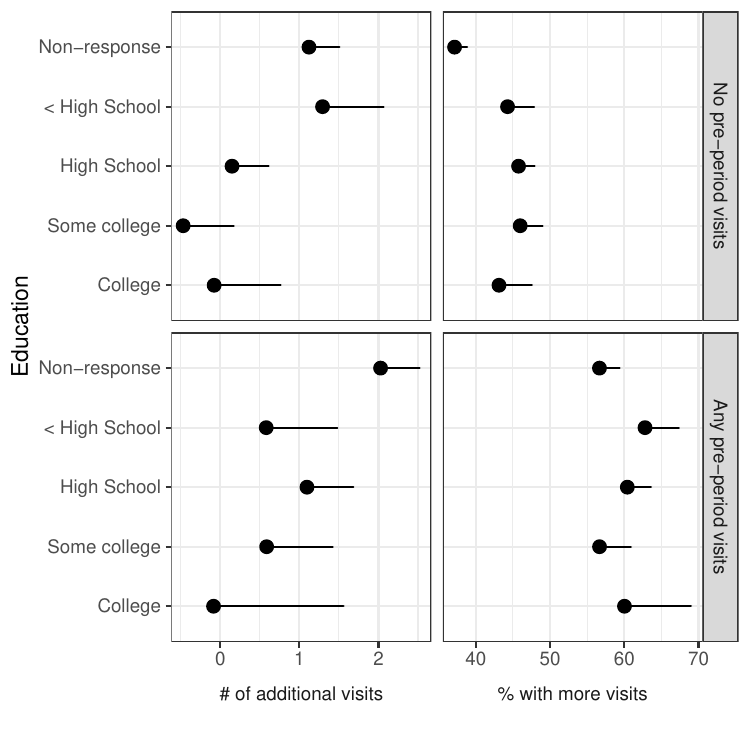}
    \caption{Natural off-waitlist Medicaid enrollment.}
    \label{fig:ed_het_plot_1_iv}
  \end{subfigure}
  \caption{Estimated upper bound on the expected number of additional PC-treatable ED visits and the percentage of individuals with any more PC-treatable ED visits for (a) randomizing access to Medicaid vs the ED-minimizing waitlist policy, and (b) natural Medicaid uptake conditional on access vs the ED-minimizing coverage policy. Upper end of line represents 95\% upper confidence bound. ``Non-response'' indicates that the individual did not respond to the in-person survey.}
  \label{fig:het_plot1}
\end{figure}

These ranges also mask considerable variation in the level of uncertainty.
Figure~\ref{fig:ed_het_plot_1} shows the estimated bounds using the least amount of regularization, stratified by education and whether the individual had any pre-randomization ED visits.
We see that for individuals that had no ED visits in the pre-period, the upper bound on the expected number of additional PC-treatable ED visits is much smaller than for individuals that had at least one visit; a similar trend holds for the proportion of individuals with any additional visits. There is limited heterogeneity by level of education.
In Appendix~\ref{sec:condl_bounds} we consider estimating bounds conditional on continuous measures using a pseudo-outcome regression approach.

Taken together, these results imply that there is potential for an optimized waitlist rule.
To evaluate this, we consider a stylized policy learning problem where the goal is to reduce PC-treatable ED visits (this is not the only goal of Medicaid).
We learn a minimax regret waitlist allocation rule relative to the oracle policy under power-law collective utility functions
using the de-biased entropic estimator of the regret relative to the oracle with $\eta = 100$. We parameterize the rule as a logistic function of the number of pre-period ED visits, and find the minimax optimal parameters by solving Eqn~\eqref{eq:entropic_bound_est}.
Figure~\ref{fig:policies} shows the estimated minimax allocation rules using a range of power-law collective utility functions.
With $\lambda = 1$, the regret is minimized by minimizing the conditional average effect of $Z$ on $Y$; this selects nobody off the waitlist because the average effect is negative across the range of pre-period ED visits.
As $\lambda$ decreases, the minimax regret rule is no longer point-identifiable, and the estimated minimax regret waitlist policies assign a higher selection probability to individuals with a higher number of pre-period ED visits, with the slope flattening for more inequality-averse utility functions.

\paragraph{Impacts of Medicaid enrollment.}

Next, we consider the same question but for enrollment in Medicaid $D$ rather than selection from the waitlist $Z$. In particular, if there were no capacity constraints, and we could enroll individuals into Medicaid directly based off their potential outcomes ($Y(0)$ and $Y(1)$), how would this compare to simply taking everyone off the waitlist (i.e. $Y(D(1))$)? We make this comparison in two ways: (i) the number of additional PC-treatable ED visits $E[Y(D(1))] - \E[\min_d Y(d)]$, and (ii) the proportion of individuals with additional PC-treatable ED visits under the no-waitlist policy, $P(Y(D(1)) \neq \min_{d} Y(d))$.
Here we use the de-biased primal-dual estimator with nuisance functions again fit using multinomial gradient boosted decision trees and 3-fold cross-fitting.
In the same computing environment as before, computing the estimates took under 5 seconds after cross-fitting the nuisance functions. For $7.3\%$ of units in the sample, the LP was not feasible, implying that either the exclusion restriction does not hold
or there is too much estimation error in the nuisance functions for those units. For these units we project the corresponding $\hat{\bm{b}}(x)$ vectors onto the feasible set.

Overall, we estimate that if everyone were to be taken off the waitlist
there would be at most an additional 0.81 PC-treatable ED visits per person relative to the oracle policy that enrolls individuals into Medicaid based on their potential outcomes (95\% bound: 1), and at most 47.9\% of individuals would have any more PC-treatable ED visits relative to the oracle policy (95\% bound: 48.8\%).
Again there is substantial heterogeneity in these upper bounds.
Figure~\ref{fig:ed_het_plot_1_iv} shows the estimated bounds stratified by education and whether the individual had an ED visit in the 14 months pre-randomization.
We see that the upper bound on the percentage of individuals with any more PC-treatable ED visits is generally lower for individuals with no pre-period visits, while the upper bounds for the number of additional PC-treatable ED visits are mixed.

\section{Discussion}

In this paper we proposed a general framework for estimating bounds on partially-identified parameters that satisfy a series of linear constraints, conditional on auxiliary covariates, leading to a series of linear programs that are conditional on the covariates.
We developed two estimators for these bounds.
The de-biased primal-dual estimator directly solves the conditional LPs for each unique covariate value in the data, then extracts readily available information from standard LP solvers to construct de-biased estimates of the bounds.
This estimator is related to existing methods for covariate-assisted bounds, but eschews the need to analytically derive the bound or fully enumerate a combinatorial number of potential solutions.

A key issue with solving the LPs directly is that the solutions are not smooth in the constraint and objective vectors, so small errors in these nuisance parameters can lead to large changes in the estimated bounds.
The entropic regularized estimator instead targets a smoothed version of the bounds, adding an entropy penalty to the conditional LPs that ensures that the solutions are smooth in the nuisance parameters and so we can estimate them in a de-biased manner.
Although this induces approximation error,
we can likely expect the approximation error to be small relative to the standard error of the estimates if we allow the level of regularization to shrink appropriately.
Finally, we extend these results to decision making problems where the value of a decision policy is only partially identified by a set of conditional linear constraints. We show that taking an empirical risk minimization approach and optimizing bounds estimated via the de-biased primal-dual or entropic estimators results in a policy with excess regret that has an analogous structure to the bias of the estimators.

An important limitation of these approaches is that they require a finite number of decision variables and so, for instance, cannot accommodate continuous outcomes.
While it may be possible to approximate some continuous outcomes with a finite number of discrete levels, this may not be feasible in all cases.
One potential path forward is to leverage the duality of the LPs to adapt the de-biased BFS and entropic estimators to return valid bounds even if an approximation is used, similar in spirit to the dual estimator proposed by \cite{ji_model-agnostic_2024}.
In addition, under our asymptotic framework the number of decision variables $K$ and constraints $J$ are fixed as the sample size grows. In some instances, e.g. when discretizing continuous outcomes, it may be more appropriate to consider an asymptotic framework where $K$ and $J$ grow with the sample size. This would likely require additional regularity conditions, and could be particularly challenging because the number of nuisance functions ($K + J$) would also grow with the sample size.
While we have used an entropy penalty to regularize the LPs, other forms of regularization may have different benefits. For example, it may be possible to use barrier-based interior point methods to solve the LPs, and directly account for regularization bias via a double de-biased estimator following recent results in \citet{liu_beyond_2025}. Such an approach could lessen the need to select a regularization hyperparameter.
Another limitation is that we have sidestepped the question whether and when conditioning on covariates information leads to narrower bounds.

\clearpage

\appendix
\allowdisplaybreaks

\onehalfspacing

\renewcommand\thefigure{\thesection.\arabic{figure}}
\setcounter{figure}{0}
\numberwithin{figure}{section}
\renewcommand\theassumption{\thesection.\arabic{assumption}}
\setcounter{assumption}{0}
\renewcommand\theexample{\thesection.\arabic{example}}
\setcounter{example}{0}
\numberwithin{equation}{section}
\setcounter{theorem}{0}
\numberwithin{theorem}{section}
\setcounter{proposition}{0}
\numberwithin{proposition}{section}
\setcounter{lemma}{0}
\numberwithin{lemma}{section}
\setcounter{corollary}{0}
\numberwithin{corollary}{section}
\setcounter{table}{0}
\numberwithin{table}{section}

\appendix

\renewcommand\thefigure{\thesection.\arabic{figure}}
\setcounter{figure}{0}
\numberwithin{figure}{section}
\renewcommand\theassumption{\thesection.\arabic{assumption}}
\setcounter{assumption}{0}
\renewcommand\theexample{\thesection.\arabic{example}}
\setcounter{example}{0}
\numberwithin{equation}{section}
\setcounter{theorem}{0}
\numberwithin{theorem}{section}
\setcounter{proposition}{0}
\numberwithin{proposition}{section}
\setcounter{lemma}{0}
\numberwithin{lemma}{section}
\setcounter{corollary}{0}
\numberwithin{corollary}{section}
\setcounter{table}{0}
\numberwithin{table}{section}

\section{Examples of conditional linear programs}
\label{sec:examples}

\subsection{Examples in general form}
\begin{example}[Functions of the joint distributions of potential outcomes]
\label{ex:joint_po}
Consider a setting where we are interested in understanding the causal relationship between a discrete treatment variable $D \in \{0,\ldots, M-1\}$ and a discrete  outcome variable  $Y \in \mathcal{Y}$ where $|\mathcal{Y}| = L$. For example, $D$ may represent a multi-level treatment or different doses.
We posit the existence of $M$ potential outcomes $(Y(0),\ldots,Y(M-1))$, one for each treatment level.
In this setting, our goal is to estimate some moment of the joint distribution of potential outcomes, i.e. $\theta = \E[f(Y(0), \ldots, Y(M-1); X)]$ for some function $f:\mathcal{Y}^M \times \mathcal{X} \to \R$. We observe iid data $(X, D, \{\bbone\{D = d\} Y(d) \}_{d = 0}^{M-1})$ from some distribution $P$.
\end{example}

Using the tower property, we can write the estimand as
\begin{equation}
  \label{eq:joint_po}
    \theta = \E\left[\sum_{y_0,\ldots,y_{M-1} \in \mathcal{Y}} f(y_0, \ldots, y_{M-1}; X) P(Y(0) = y_0, \ldots, Y(M-1) = y_{M-1} \mid X) \right].
\end{equation}
In this case the unknown function $\bm{p}^\ast(x)$ corresponds to the joint distribution of the potential outcomes given that the covariates $X = x$, and the function $\bm{c}(x)$ is a vector collecting the function values $f(y_0, \ldots, y_{M-1}; x)$.
For example, $f(y_0, \ldots, y_{M-1}; x) = y_d - y_{d'}$ would correspond to the average treatment effect of $d$ versus $d'$ and $f(y_0, \ldots, y_{M-1}; x) = \max_{d} y_d$ would correspond to average outcome under an oracle treatment assignment rule that maximizes the potential outcomes.
Another parameter we consider in Section~\ref{sec:ohie} is the proportion of individuals who are not assigned an optimal treatment, i.e. $P(Y \neq \max_{d} Y(d))$.
This corresponds to choosing $f(y_0, \ldots, y_{M-1}; x) = \sum_{d'}P(D = d'\mid X = x)\bbone\{y_{d'} \neq \max_{d} y_d\}$.
Note that if the propensity score $P(D = d' \mid X = x)$ is unknown, then the objective vector $\bm{c}(x)$ is also unknown but can be estimated.

Under the standard strong ignorability assumption that there is treatment overlap and the potential outcomes are independent of the treatment given the covariates \citep{Rosenbaum1983}, the marginal distribution of each potential outcome given the covariates is identified as $P(Y(d) = y_\ell \mid X = x) = P(Y = y_\ell \mid D = d, X = x)$ for all $d \in \{0,\ldots,M-1\}$ and $\ell \in \mathcal{Y}$. This gives a set of marginal constraints on the conditional joint distribution of potential outcomes:
\[
  \sum_{\substack{y_0,\ldots,y_{d - 1}, \\ y_{d + 1}, \ldots, y_{M-1} \in \mathcal{Y}}} P(Y(0) = y_0, \ldots, Y(d) = y_\ell,\ldots Y(M - 1) = y_{M-1} \mid X = x) = P(Y = y_\ell \mid D = d, X = x).
  \]
Since the probabilities sum to one,
each margin gives $L - 1$ constraints. There is also an overall sum-to-one constraint on the joint distribution. This leads to a total of $J = (L-1)M + 1$ constraints, on $K = L^M$ variables (the number of elements in the joint distribution) which can be collected into an integral constraint matrix $A \in \{0,1\}^{J \times K}$. The vector corresponding to the right hand side of the constraints, $\bm{b}(x)$, consists of the marginal probability distributions for each potential outcome, as well as 1 for the sum-to-one constraint.

\begin{example}[Policy learning with collective utility functions]
  \label{ex:policy_learning}
A special case of the above is when we have a binary treatment $D \in \{0,1\}$ and we want to learn a decision policy $\pi:\mathcal{X} \to [0,1]$ that assigns $D=1$ with probability $\pi(X)$ and achieves some notion of optimality based on a function $u:\mathcal{Y} \to \R^+$ that assigns a utility to an outcome $y$ of $u(y)$.
We can write the expected utility for an individual under decision policy $\pi$ as $\pi(X) \times u(Y(1)) + (1 - \pi(X)) \times u(Y(0))$.

To choose a policy $\pi$, we can use \emph{power-law collective utility functions} that aggregate the individual utilities into a single number \citep{moulin_axioms_1988}:
\[
  \begin{aligned}
    V^\lambda(\pi) & \equiv \E\left[\frac{1}{\lambda}\left(\left\{u(Y(0)) + \pi(X) (u(Y(1)) - u(Y(0)))\right\}^\lambda - 1\right)\right].
  \end{aligned}
\]
Here the value of the power-law parameter $\lambda$ determines the level of emphasis placed on individuals with lower or higher utilities under $\pi$. When $\lambda = 1$, the value function is simply the expected individual utility under $\pi$.
Values of $\lambda < 1$ will lead to inequality-averse value functions.
When $\lambda = 0$, the  $V^0(\pi)$ is the logarithm of the geometric mean of the individual utilities, and as $\lambda \to -\infty$, $V^\lambda(\pi)$ goes towards prioritizing only the individuals with the smallest utilities under $\pi$; choosing $\pi$ to maximize this results in the \emph{egalitarian rule} that maximizes the minimum utility.
Finally, values of $\lambda > 1$ will lead to inequality-seeking value functions.
Note that when $\lambda \neq 1$, the power-law value function is not point-identifiable.

To create a notion of optimality, we can compare the value under the policy $\pi$ to the value under an infeasible oracle policy that has access to the potential outcomes and assigns the treatment that maximizes the individual utility, i.e. $\bbone\{u(Y(1)) \geq u(Y(0))\}$.
The \emph{regret} of policy $\pi$ relative to the oracle is thus:
\[
  R^\lambda(\pi)  = \E\left[\sum_{y_0, y_1 \in \mathcal{Y}} \frac{\left(\max\{u(y_0), u(y_1)\} ^\lambda - \left\{u(y_0) + \pi(X) (u(y_1) - u(y_0))\right\}^\lambda\right)}{\lambda} P(Y(0) = y_0, Y(1) = y_1 \mid X)\right].
\]
This is a special case of Eqn~\eqref{eq:joint_po} with $M = 2$ and $f(y_0, y_1; x) = \frac{1}{\lambda}\max\{u(y_0), u(y_1)\} ^\lambda - \frac{1}{\lambda}\left(u(y_0) + \pi(x) (u(y_1) - u(y_0))\right)^\lambda$.
We will discuss finding optimal policies $\pi$ that minimize the maximum regret or maximize the minimum value in Section~\ref{sec:policy_learning}.
\end{example}

\begin{example}[Instrumental variables]
  \label{ex:iv}
Now consider the case of a binary treatment $D$ and a binary instrumental variable $Z \in \{0,1\}$ that affects takeup of the treatment, where the outcome space $\mathcal{Y}$ still contains $L$ distinct levels (this discussion generalizes to non-binary treatments and instruments). The instrument leads to two potential treatments $D(0)$ and $D(1)$. Under an exclusion restriction that the instrument does not have a direct effect on the outcome variable, we then have two potential outcomes $Y(0)$ and $Y(1)$.
In this case, we are interested in estimating a property of the joint distribution of potential treatments and potential outcomes $\theta = \E\left[f(Y(0), Y(1), D(0), D(1))\right]$ for some function $f:\mathcal{Y}^2\times \{0,1\}^2 \to \R$.
The observable data consists of iid samples of $(X, Z, (1-Z) D(0), Z D(1), (1-D(Z)) Y(0), D(Z) Y(1))$.

We can write the estimand as
\[
\begin{aligned}
  \theta  = \E\left[\sum_{y_0,y_1\in \mathcal{Y}}\sum_{d_0,d_1=0}^1 f(y_0, y_1, d_0, d_1; X) P(Y(0) = y_0, Y(1) = y_1, D(0) = d_0, D(1) = d_1 \mid X) \right].
\end{aligned}
\]
Now the unknown function $\bm{p}^\ast(x)$ corresponds to the joint distribution of the potential outcomes and potential treatments given the covariates $X = x$, and the function $\bm{c}(x)$ is a vector that collects the values of the function $ f(y_0, y_1, d_0, d_1; x) $.
For example, $f(y_0, y_1, d_0, d_1; x) = y_1 - y_0$ would again correspond to the average treatment effect.
Another example is the proportion of individuals who are not assigned an optimal treatment under the treatment distribution induced by a particular level of the instrument: $P(Y(D(z)) \neq \max_{d} Y(d))$. This corresponds to choosing $f(y_0, y_1, d_0, d_1; x) = \sum_{d'} P(D = d' \mid Z = z, X = x)\bbone\{y_{d'} \neq \max_{d} y_d\}$.

Under ignorable assignment of the instrument, the joint distribution of the observed treatment and outcome satisfies
\[
  P(Y = y, D = d \mid Z = z, X) = \sum_{y_{1-d} \in \mathcal{Y}}\sum_{d_{1-z}=0}^1 P(Y(d) = y, Y(1-d) = y_{1-d}, D(z) = d, D(1-z) = d_{1-z} \mid X).
\]
Accounting for the sum-to-one constraints gives $J = 4(L-1) + 1$ constraints with $K=4L^2$ decision variables.
Here the constraint matrix $A \in \{0,1\}^{J \times K}$ encodes the marginal constraints and the vector $\bm{b}(x)$ encodes the conditional joint distribution of the observed outcome and treatment variables along with the sum to one constraint.
\end{example}

\subsection{De-biased rates for examples}
\label{sec:examples_method}
\paragraph{Functions of the joint distributions of potential outcomes.}

Recall that for estimands involving the joint distribution of potential outcomes, the constraints involve the marginal distributions of each outcome. There are various ways to index the constraints, but say that the first $j=1,\ldots,L-1$ constraints correspond to the marginal distribution of $Y(0)$, the next $j=L,\ldots, 2L-2$ constraints correspond to the marginal distribution of $Y(1)$, and so on.
In this case, we can define $b_j(x) = m_j(d_j, x) \equiv P(Y = y_j \mid D = d_j, X = x)$ where $d_j$  and $y_j$ index the treatment and outcome levels corresponding to the $j$\super{th} constraint, respectively.
Defining the propensity score as $e(d_j, x) = P(D = d_j \mid X = x)$,
a de-biasing function can be constructed from the Riesz representor (or efficient influence function) as $\varphi^{(b)}_{j}(O; b_j) = \frac{\bbone\{D = d_j\}}{P(D = d_j \mid X)}\left(\bbone\{Y = y_j\} - b_j(x)\right)$, the inverse propensity weighted residual.
Here the de-biased rate is $r_{n} = \|\hat{\bm{b}} - \bm{b}\|_2 \|\hat{\bm{e}} - \bm{e}\|_2$.

For a case where the objective vector $\bm{c}(x)$ is not known, consider estimating the proportion of individuals not assigned an optimal treatment $P(Y \neq \max_{d} Y(d))$. Recall that in this case the elements of $\bm{c}(x)$ correspond to $\sum_{d'}e(d', x)\bbone\{y_{d'} \neq \max_{d} y_d^{(k)}\}$, where $(y_0^{(k)},\ldots,y_{M-1}^{(k)})$ are the $k$\super{th} combination of possible potential outcomes. Here,
$\bm{c}(x)$ is a function of the propensity score, and we can estimate it as
$\hat{c}_k(x) = \sum_{d'}\hat{e}(d', x)\bbone\{y_{d'} \neq \max_{d} y_d^{(k)}\}$, where $\hat{e}(d', x)$ is the estimated propensity score.
The corresponding de-biasing function is a function of the residuals of the propensity score: $\varphi_k^{(c)}(O; c_k) = \sum_{d'}\left(\bbone\{D = d'\} - e(d', x)\right)\bbone\{y_{d'} \neq \max_{d} y_d^{(k)}\}$.
In this case,
$\hat{c}_k(x) + \hat{\varphi}^{(c)}_k(O)  = \sum_{d'}\bbone\{D = d'\}\bbone\{y_{d'} \neq \max_{d} y_d^{(k)}\}$ and so the de-biased rate for $\bm{c}(\cdot)$ is simply $r_{n} = 0$.

\paragraph{Instrumental variables.}

In the case of instrumental variables, recall that the constraints involve the joint distribution of the outcome $Y$ and treatment $D$, conditioned on the instrument $Z$ and covariates $X$.
The constraint vector is $b_j(x) = P(Y = y_j, D = d_j \mid Z = z_j, X = x)$, where $y_j, d_j$, and $z_j$ index the levels of the outcome, treatment, and instrument corresponding to the $j$\super{th} constraint, respectively.
The de-biasing function is again constructed from the Riesz representor as
$\varphi_j^{(b)}(O; b_j) = \frac{\bbone\{Z = z_j\}}{P(Z = z_j \mid X)}\left(\bbone\{Y = y_j, D = d_j\} - b_j(x)\right)$, the inverse instrument probability weighted residual.
Here the de-biased rate is the product of the errors in the estimates of the joint distribution of $Y$ and $D$ and the instrument assignment probabilities.

Now consider estimating bounds on the proportion of individuals not selecting an optimal treatment given a particular level of the instrument $P(Y(D(z)) \neq \max_{d} Y(d))$.
Above we discussed that the elements of $\bm{c}(x)$ correspond to $\sum_{d'}P(D = d'\mid Z = z, X = x)\bbone\{y_{d'} \neq \max_{d} y_d^{(k)}\}$, where $(y_0^{(k)},\ldots,y_{M-1}^{(k)})$ are the $k$\super{th} combination of possible potential outcomes and treatment levels.
We can estimate it as $\hat{c}_k(x) = \sum_{d'}\hat{P}(D = d'\mid Z = z, X = x)\bbone\{y_{d'} \neq \max_{d} y_d^{(k)}\}$, where $\hat{P}(D = d'\mid Z = z, X = x)$ is the estimated treatment selection probability given the level of the instrument and the value of the covariates.
The corresponding de-biasing function is the instrument probability weighted residual $\varphi_k^{(c)}(O; c_k) = \sum_{d'}\frac{\bbone\{Z = z\}}{P(Z = z \mid X)}\left(\bbone\{D = d'\} - P(D = d' \mid Z = z, X)\right)\bbone\{y_{d'} \neq \max_{d} y_d^{(k)}\}$,
and so the de-biased rate will be a product of the error in the estimates of the treatment and instrument assignment probabilities.

\section{Worked example of margin conditions}
\label{sec:margin_ex}

Here we consider an example with a binary treatment $D \in \{0,1\}$ and binary potential outcomes $Y(0),Y(1) \in \{0,1\}$, where our goal is to estimate the probability that the observed treatment assignment is not optimal: $\theta = P(Y \neq \max_dY(d))$. Letting $e(x) = P(D = 1 \mid X = x)$ be the propensity score, $m(d, x) = P(Y(d) = 1 \mid X = x)$, and ordering the decision variables as $\bm{p}(x) = (P(Y(0) = 0, Y(1) = 0 \mid X = x), P(Y(0) = 0, Y(1) = 1 \mid X = x), P(Y(0) = 1, Y(1) = 0 \mid X = x), P(Y(0) = 1, Y(1) = 1 \mid X = x))$, the constraint matrix is
\[
A = \left[\begin{array}{c c c c}
  1 & 1 & 0 & 0\\
  1 & 0 & 1 & 0\\
  1 & 1 & 1 & 1
\end{array} \right],
\]
the constraint vector is $\bm{b}(x) = (1 - m(0,x), 1 - m(1,x), 1)$ and the objective vector is $\bm{c}(x) = (0, 1-e(x), e(x), 0)$.
Setting each of the decision variables equal to zero and solving for the remaining ones in turn gives that the four basic primal and dual solutions are
\begin{align*}
  \tilde{\bm{p}}_1(x) & = (1 - m(1, x) - m(0, x), m(1, x), m(0,x), 0) & \tilde{\bm{\lambda}}_1(x) & = (-e(x), e(x) - 1, 1)\\
  \tilde{\bm{p}}_2(x) & = (1 - m(1,x), m(1,x) - m(0,x), 0, m(0,x)) & \tilde{\bm{\lambda}}_2(x) & = (1-e(x), e(x) - 1, 0)\\
  \tilde{\bm{p}}_3(x) & = (1 - m(0,x), 0, m(0,x) - m(1,x), m(1,x)) & \tilde{\bm{\lambda}}_3(x) & = (-e(x), e(x), 0)\\
  \tilde{\bm{p}}_4(x) & = (0, 1 - m(0,x), 1 - m(1,x), m(1,x) + m(0,x)- 1) & \tilde{\bm{\lambda}}_4(x) & = (1-e(x), e(x), 0).
\end{align*}
Now assessing dual feasibility:
\begin{align*}
  A'\tilde{\bm{\lambda}}_1(x)  - \bm{c}(x) & = (0, 0, 0, 1)  \not \leq 0\\
  A'\tilde{\bm{\lambda}}_2(x)  - \bm{c}(x) & = (0, 0, - 1, 0)  \leq 0 \\
  A'\tilde{\bm{\lambda}}_3(x)  - \bm{c}(x) & = (0, -1, 0, 0)  \leq 0\\
  A'\tilde{\bm{\lambda}}_4(x) - \bm{c}(x) & = (1, 0, 0, 0)  \not \leq 0.
\end{align*}
Therefore, the two dual infeasible bases are always $\mathcal{B} \setminus \mathcal{D}(x) = \{(1,2,3), (2,3,4)\}$, regardless of the value of the covariate $x$, and the infeasibility gap $\mathcal{E}_\lambda(x) = 1$ for all $x \in \mathcal{X}$, implying that the dual infeasibility margin condition holds for any $\alpha > 0$.

Turning to the optimality detection margin, note that since there are only two dual-feasible bases, i.e. $|\mathcal{D}(x)| = 2$ for all $x \in \mathcal{X}$, the set of dual-feasible but not dual-optimal bases $\mathcal{D}^{-}(x)$ is either empty or contains one of the $(1,2,4)$ or $(1,3,4)$ bases. $\mathcal{D}^{-}(x)$  is empty when
\[
  \<\tilde{\bm{\lambda}}_2(x) - \tilde{\bm{\lambda}}_3(x), \bm{b}(x) \> = m(1,x) - m(0,x) = 0,
\]
so when the Conditional Average Treatment Effect (CATE) is zero ($m(1,x) - m(0,x) = 0$), both of the dual-feasible bases are also dual-optimal, and the optimality detection gap $\mathcal{E}_p(x) = \infty$. Now assessing primal infeasibility for the dual-feasible but not dual-optimal bases, we have that $\tilde{\bm{p}}_2(x)$ is primal infeasible when the CATE is negative, $m(1,x)< m(0,x)$ (so $\tilde{\bm{p}}_3(x)$ is optimal), and conversely $\tilde{\bm{p}}_3(x)$ is primal infeasible when the CATE is positive, $m(1,x) > m(0,x)$ (so $\tilde{\bm{p}}_2(x)$ is optimal). Therefore, the optimality detection gap is $\mathcal{E}_p(x) = |m(1,x) - m(0,x)|$, and the optimality detection margin condition becomes
\[
  P(0 < |m(1,X) - m(0,X)| \leq t) \lesssim t^\alpha.
\]
  This is a standard margin condition on the CATE, which has appeared in other settings \citep[e.g.][]{Luedtke2016,Kallus2022_harm,benmichael_asymmetric_2024}, but it also allows for the possibility that the CATE is exactly zero with positive probability.
  However, note that the primal non-degeneracy condition in Assumption~\ref{a:rate_bfs}\ref{a:rate_bfs_ii} requires that the CATE be almost surely non-zero because $m(1,x) - m(0,x) = 0$ leads to the two possible primal solutions being degenerate; it also requires that $m(1,x)$ and $m(0,x)$ be strictly bounded away from 0 and 1, so it rules out potential outcomes that are deterministically zero or one given the covariates.

  Next we consider the primal sub-optimality margin. The four possible values corresponding to the four bases are
  \begin{align*}
    \<\bm{c}(x), \tilde{\bm{p}}_1(x)\> & = (1 - e(x))m(1,x) + e(x) m(0,x)\\
    \<\bm{c}(x), \tilde{\bm{p}}_2(x)\> & = (1 - e(x)) (m(1,x) - m(0,x))\\
    \<\bm{c}(x), \tilde{\bm{p}}_3(x)\> & = -e(x) (m(1,x) - m(0,x))\\
    \<\bm{c}(x), \tilde{\bm{p}}_4(x)\> & = (1-e(x))(1 - m(0,x)) + e(x)(1 - m(1, x)).
  \end{align*}

  Note that the set of primal feasible bases depends on the signs of $m(1,x) -(1 - m(0,x))$ and $m(1,x) - m(0,x)$.

  When $m(1,x) - (1 - m(0,x)) >0$ and $m(1,x) - m(0,x) > 0$, the two feasible solutions are $\tilde{\bm{p}}_2(x)$ and $\tilde{\bm{p}}_4(x)$, the optimal solution uses $\tilde{\bm{p}}_2(x)$ and the difference in values is $1 - m(1,x)$. When $m(1,x) - (1 - m(0,x)) > 0$ and $m(1,x) - m(0,x) < 0$, the two feasible solutions are $\tilde{\bm{p}}_3(x)$ and $\tilde{\bm{p}}_4(x)$, the optimal solution uses $\tilde{\bm{p}}_3(x)$ and the difference in values is $1 - m(0,x)$ (this also holds when $m(1,x) - m(0,x) = 0$).

  When $m(1,x) - (1 - m(0,x)) < 0$, and $m(1,x) - m(0,x) > 0$  the two feasible solutions are $\tilde{\bm{p}}_1(x)$ and $\tilde{\bm{p}}_2(x)$, the optimal solution uses $\tilde{\bm{p}}_2(x)$ and the difference in values is $m(0,x)$.  When $m(1,x) - (1 - m(0,x)) < 0$, and $m(1,x) - m(0,x) < 0$  the two feasible solutions are $\tilde{\bm{p}}_1(x)$ and $\tilde{\bm{p}}_3(x)$, the optimal solution uses $\tilde{\bm{p}}_3(x)$ and the difference in values is $m(1,x)$ (this also holds when $m(1,x) - m(0,x) = 0$). Therefore, the sub-optimality gap is the minimum distance of any of the probabilities from 0, $\Delta_p(x) = \min_{d \in \{0,1\}} \min\{m(d,x ), 1 - m(d,x)\}$ and the primal sub-optimality margin condition is that
  \[
    P(0 < \min_{d \in \{0,1\}} \min\{m(d,X), 1 - m(d,X)\} \leq t) \lesssim t^\alpha.
  \]
  Note that Assumption~\ref{a:strict_feasibility} of strict feasibility implies a hard margin, i.e. that $\Delta_p(x) \geq \delta > 0$, so the primal sub-optimality condition holds with $\alpha = \infty$ and the entropic estimator in Section~\ref{sec:entropic} has an exponentially decaying approximation error.

\section{Extension to uniform convergence}
\label{sec:uniform_convergence}

In Section~\ref{sec:method}, we present convergence results that hold pointwise in the distribution. However, it is also possible to obtain uniform convergence results over a class of distributions $\mathcal{Q}$ with stronger versions of the assumptions. To explicitly state these assumptions and results, we will briefly recapitulate the setup, assumptions and results from Sections~\ref{sec:partial_id} and \ref{sec:method} keeping track of the dependence on the distribution $Q$. We will subscript expectations and probabilities with $Q$ to make this dependence explicit.
We observe $n$ i.i.d. samples $O_1, \ldots, O_n$ from a distribution $Q \in \mathcal{Q}$, and are interested in estimating $\theta_L(Q) = \E_Q\left[\theta_L(x; Q)\right]$, where $\theta_L(x;Q)$ is defined as the solution to the conditional linear program in \eqref{eq:partial_id} with constraint vector $\bm{b}(x;Q)$ and objective vector $\bm{c}(x;Q)$.

\begin{assumption}
  \label{a:compact_uniform}
  For all distributions $Q \in \mathcal{Q}$, the feasible sets $\mathcal{P}(X; Q)$ and $\Lambda(X; Q)$ are non-empty and $\mathcal{P}(X; Q)$ is bounded with probability 1, and both $\sup_{Q \in \mathcal{Q}}\|\bm{b}(\cdot;Q)\|_\infty < \infty$ and $\sup_{Q \in \mathcal{Q}}\|\bm{c}(\cdot;Q)\|_\infty < \infty$.
\end{assumption}

We next state assumptions that are uniform versions of Assumptions~\ref{a:riesz} and \ref{a:rates}. For these and other rate conditions, we will use the notation $o_p(a_n)$ to mean that the quantity in question divided by $a_n$ converges to zero in probability uniformly over $Q \in \mathcal{Q}$, i.e. $\lim_{n \to \infty} \sup_{Q \in \mathcal{Q}} P_Q(|X_n| / a_n > \epsilon) = 0$ for all $\epsilon > 0$.

\begin{assumption}
  \label{a:riesz_uniform}
  For every $Q \in\mathcal{Q}$, the following holds. For each constraint $j=1,\ldots,J$ there exists a de-biasing function $\varphi^{(b)}_{j}(O; b_j, Q)$
  such that $\E_Q[\varphi^{(b)}_{j}(O; b_j, Q) \mid X] = 0$, $\E_Q[\varphi_{j}^{(b)}(O; b_j, Q)^2] < \infty$, and for any bounded function $\bar{\bm{b}}:\mathcal{X} \to \R^J$  $\E_Q[\varphi_{j}^{(b)}(O; \bar{b}_j, Q) \mid X = x] = b_j(x; Q) - \bar{b}_j(x)$  for all $x \in \mathcal{X}$.
  Similarly, for each decision variable $k=1,\ldots,K$ there exists a de-biasing function $\varphi_{k}^{(c)}(O; c_k, Q)$ such that $\E_Q[\varphi^{(c)}_{k}(O; c_k, Q) \mid X] = 0$, $\E_Q[\varphi_{k}^{(c)}(O; c_k, Q)^2] < \infty$, and for any bounded function $\bar{\bm{c}}:\mathcal{X} \to \R^K$, $\E_Q[\varphi_{k}^{(c)}(O; \bar{c}_k, Q) \mid X = x] = c_k(x; Q) - \bar{c}_k(x)$  for all $x \in \mathcal{X}$.
\end{assumption}

\begin{assumption}
  \label{a:rates_uniform}
  \[\sup_{Q\in\mathcal{Q}}\sup_{x\in\mathcal{X}}\E_Q[\varphi^{(b)}_j(O; b_j, Q)^2 \mid X = x] < \infty\] and
  \[\sup_{Q\in\mathcal{Q}}\sup_{x\in\mathcal{X}}\E_Q[\varphi^{(c)}_k(O; c_k, Q)^2 \mid X = x] < \infty\] for all $j=1,\ldots,J$ and $k=1,\ldots,K$.
   For each $Q \in \mathcal{Q}$, the following holds.
   $\hat{\bm{b}}(\cdot;Q)$, $\hat{\bm{c}}(\cdot;Q)$, $\hat{\bm{\varphi}}^{(b)}(\cdot;Q)$ and $\hat{\bm{\varphi}}^{(c)}(\cdot;Q)$ are fit on a separate, independent sample drawn from $Q(\cdot)$.
   The estimated feasible sets  $\widehat{\mathcal{P}}(X; Q)$ and $\widehat{\Lambda}(X; Q)$ are non-empty and $\widehat{\mathcal{P}}(X; Q)$  is bounded with probability 1.
   $\|\hat{\bm{b}}(\cdot;Q) - \bm{b}(\cdot;Q)\|_\infty = o_p(1)$, $\ \|\hat{\bm{c}}(\cdot;Q) - \bm{c}(\cdot;Q)\|_\infty = o_p(1)$, $\|\hat{\bm{\varphi}}^{(b)}(\cdot;Q) - \bm{\varphi}^{(b)}(\cdot;Q)\|_\infty = o_p(1)$, and $ \|\hat{\bm{\varphi}}^{(c)}(\cdot;Q) - \bm{\varphi}^{(c)}(\cdot;Q)\|_\infty = o_p(1)$.
  Furthermore, there is some rate $r_n$ such that for any linear function of $\bm{b}(X; Q)$ defined by $f:\mathcal{X} \to \R^J$,  we have that $\sup_{Q \in \mathcal{Q}}\left|\E_Q\left[\<\bm{f}(X), \hat{\bm{b}}(X;Q) + \hat{\bm{\varphi}}^{(b)}(O;Q) - \bm{b}(X;Q)\>\right]\right| \lesssim \sup_{x \in \mathcal{X}} \|\bm{f}(x)\|_2 \times r_n$.
  Similarly, for any linear function of $\bm{c}(X; Q)$ defined by $f:\mathcal{X} \to \R^K$, we have
   that
   \[\sup_{Q \in \mathcal{Q}}\left|\E_Q\left[\<\bm{f}(X), \hat{\bm{c}}(X;Q) + \hat{\bm{\varphi}}^{(c)}(O;Q) - \bm{c}(X;Q)\>\right]\right| \lesssim \sup_{x \in \mathcal{X}} \|\bm{f}(x)\|_2 \times r_n.\]
\end{assumption}

We next explicitly state assumptions analogous to the margin conditions and rate and regularity conditions in Assumptions~\ref{a:infeasible_gap},\ref{a:primal_infeasible_gap}, and \ref{a:rate_bfs} that place further restrictions on the class of distributions $\mathcal{Q}$.
To do so, we define distribution $Q$-specific versions of the set of primal-feasible bases, $\mathcal{B}(x; Q) = \{B \in \mathcal{B} \mid A_B^{-1}\bm{b}(x; Q) \geq 0\}$,
the set of dual-feasible bases, $\mathcal{D}(x;Q) = \{B \in \mathcal{B} \mid A'A_B^{-1\prime}\bm{c}(x; Q) \leq \bm{c}(x; Q)\}$, the dual infeasibility gap, $\mathcal{E}_{\lambda}(x; Q) = \min_{B \in \mathcal{B} \setminus \mathcal{D}(x; Q)} \|[A'A^{-1\prime}_B \bm{c}(x; Q) - \bm{c}(x; Q)]_+\|_\infty$ (with $\mathcal{E}_\lambda(x; Q)  = \infty$ when all basic dual solutions are feasible and so $\mathcal{B} \setminus \mathcal{D}(x; Q) = \emptyset$), the set of dual-feasible but not dual-optimal bases, $\mathcal{D}^{-}(x; Q) = \{B \in \mathcal{D}(x; Q) \mid \<\bm{c}(x; Q), A^{-1}_B \bm{b}(x; Q)\> < \theta_L(x; Q)\}$, and the optimality detection gap, $\mathcal{E}_p(x; Q) = \min_{B \in \mathcal{D}^{-}(x; Q)} \|[A_B^{-1} \bm{b}(x; Q)]_-\|_\infty$ (with $\mathcal{E}_p(x; Q) = \infty$ when all dual-feasible bases are also optimal and so $\mathcal{D}^{-}(x; Q) = \emptyset$).

\begin{assumption}[Uniform dual infeasibility margin]
  \label{a:infeasible_gap_uniform}
  For some $\alpha > 0$, $\sup_{Q \in \mathcal{Q}}P_Q(\mathcal{E}_\lambda(X; Q) \leq t) \lesssim t^\alpha $ for all $t > 0$.
\end{assumption}

\begin{assumption}[Uniform optimality detection margin]
  \label{a:primal_infeasible_gap_uniform}
  For some $\alpha > 0$, $\sup_{Q \in \mathcal{Q}}P_Q(\mathcal{E}_p(X; Q) \leq t) \lesssim t^\alpha $ for all $t > 0$.
\end{assumption}

\begin{assumption}[Uniform rate and regularity conditions for asymptotic normality]
  \label{a:rate_bfs_uniform}
    \mbox{}
    \begin{enumerate}[label=(\roman*)]
      \item For all $Q \in \mathcal{Q}$, either (a) for almost all $x \in \mathcal{X}$ there exists a unique primal solution and $\bm{b}(x;Q)$ is known up to a scalar function so that $\bm{\varphi}^{(b)}(o;Q) = f(o;Q) \bm{b}(x;Q)$ for some function $f$ such that $\sup_{Q \in \mathcal{Q}} \E_Q[f(O;Q)^2] < \infty$; (b) for almost all  $x \in \mathcal{X}$ there exists a non-degenerate primal solution and $\bm{c}(x;Q)$ is known up to a scalar function so that $\bm{\varphi}^{(c)}(o;Q) = f(o;Q) \bm{c}(x;Q)$ for some function $f$ such that $\sup_{Q \in \mathcal{Q}} \E_Q[f(O;Q)^2] < \infty$; or (c) for almost all $x \in \mathcal{X}$ there exists a unique, non-degenerate primal solution.
      \item $\|\hat{\bm{b}} - \bm{b}\|_\infty^{1+\alpha} = o_p(n^{-\frac{1}{2}})$, $\|\hat{\bm{c}} - \bm{c}\|_\infty^{1 + \alpha} = o_p(n^{-\frac{1}{2}})$, ${\|\hat{\bm{b}} - \bm{b}\|_2 \|\hat{\bm{c}} - \bm{c}\|_2} = o_p(n^{-\frac{1}{2}})$ and $r_n = o(n^{-\frac{1}{2}})$.
      \item For an optimal primal and dual solution pair $\bm{p}^\ast(X; Q)$ and $\bm{\lambda}^\ast(X; Q)$, define
      \[
      g(O; Q) \equiv \left\<\bm{c}(X;Q) + \bm{\varphi}^{(c)}(O;Q), \bm{p}^\ast(X;Q)\right\> + \left\<\bm{\lambda}^\ast(X;Q), \bm{\varphi}^{(b)}(O;Q)\right\>.
      \]
      For any optimal primal and dual solution pair, $V_L(Q) \equiv \Var_Q(g(O; Q)) > 0$ and for every $\varepsilon > 0$,
      \[
        \lim_{n \to \infty}\sup_{Q \in \mathcal{Q}}\E_Q\left[\left(g(O; Q) - \theta_L(Q)\right)^2 \bbone\{|g(O; Q) - \theta_L(Q)| \geq \varepsilon \sqrt{n V_L(Q)}\}\right] = 0.
      \]
    \end{enumerate}

\end{assumption}

Assumptions~\ref{a:infeasible_gap_uniform} and \ref{a:primal_infeasible_gap_uniform} require that the margin conditions hold over the class of distributions $\mathcal{Q}$, and Assumption~\ref{a:rate_bfs_uniform} requires that the rate and regularity conditions hold uniformly over $\mathcal{Q}$.
One implication is that the uniform versions of these assumptions together rule out the possibility of having a sequence of distributions where both the margin parameter $\alpha$ gets arbitrarily small, and the constraint and objective vectors are estimated non-parametrically. Substantively, this rules out cases where there are increasingly nearly dual-feasible bases and increasingly nearly primal-feasible bases that are dual-feasible, without collapsing to feasibility.

We will also include an additional margin-style assumption on the probability that the primal infeasibility gap for dual feasible but primal infeasible bases, $\mathcal{E}_{\text{primal}}(x;Q) = \min_{B \in \mathcal{D}(x;Q) \setminus \mathcal{B}(x;Q)} \|[A_B^{-1} \bm{b}(x; Q)]_-\|_\infty$ (with $\mathcal{E}_\text{primal}(x; Q) = \infty$ when all dual-feasible bases are also primal-feasible and so $\mathcal{D}(x; Q) = \mathcal{B}(x; Q)$).
\begin{assumption}
  \label{a:primal_infeasible_gap_uniform_2}
  $\sup_{Q \in \mathcal{Q}}P_Q(\mathcal{E}_\text{primal}(X; Q) \leq t) \to 0$ as $t \to 0$.
\end{assumption}
Unlike Assumptions~\ref{a:infeasible_gap_uniform} and \ref{a:primal_infeasible_gap_uniform}, this condition does not require a specific rate.

With these assumptions, we can state a uniform version of Corollary~\ref{cor:direct_est_normal}. We could similarly state a uniform version of Corollary~\ref{cor:entropic_est_normal} following the same style of argument. We will use $d_{PR}(X, Y)$ to denote the Prokhorov distance between the distributions of the random variables $X$ and $Y$.

\begin{corollary}
  \label{cor:primal_est_normal_unique_uniform}
  Under Assumptions~\ref{a:compact_uniform}, \ref{a:riesz_uniform}, \ref{a:rates_uniform}, \ref{a:infeasible_gap_uniform}, \ref{a:primal_infeasible_gap_uniform}, \ref{a:rate_bfs_uniform}, and \ref{a:primal_infeasible_gap_uniform_2}, we have that
  \[
    \lim_{n \to \infty}\sup_{Q \in \mathcal{Q}}d_{PR}\left(\sqrt{n}\left(\hat{\theta}_L - \theta_L(Q)\right), \sqrt{V_L(Q)} Z\right) = 0,
  \]
  where $Z \sim N(0, 1)$.
\end{corollary}
\begin{proof}[Proof of Corollary~\ref{cor:primal_est_normal_unique_uniform}]
  The proof follows the same steps as the proof of Theorem~\ref{thm:direct_est_rate} and Corollary~\ref{cor:direct_est_normal}, but with the bounds and convergence statements holding uniformly over $Q \in \mathcal{Q}$.
  The only step requiring a new argument is the proof of Lemma~\ref{lem:prob_subopt_converges}, where we need to show that the probability of selecting a non-optimal basis converges to zero uniformly over $Q \in \mathcal{Q}$.
  To see this, note that $P_Q(\widehat{B}(X) \not \in \mathcal{B}^\ast(X;Q)) = P_Q(\widehat{B}(X) \not \in \mathcal{D}(X;Q)) + P_Q(\widehat{B}(X) \in \mathcal{D}(X;Q) \setminus \mathcal{B}^\ast(X;Q))$. Following the same argument as in the proof of Lemma~\ref{lem:prob_subopt_converges} with the uniform margin condition in Assumption~\ref{a:infeasible_gap_uniform}, $\sup_{Q\in\mathcal{Q}} P_Q(\widehat{B}(X) \not \in \mathcal{D}(X;Q)) = o_p(1)$.
  Similarly, following the same argument as in the proof of Lemma~\ref{lem:prob_subopt_converges},
  \[
    P_Q(\widehat{B}(X) \in \mathcal{D}(X;Q) \setminus \mathcal{B}^\ast(X;Q)) \leq P_Q(\mathcal{E}_\text{primal}(X; Q) \leq C_0\|\hat{\bm{b}}(\cdot;Q) - \bm{b}(\cdot;Q)\|_\infty),
  \]
  for some positive constant $C_0$ that does not depend on $Q$. Since $\|\hat{\bm{b}} - \bm{b}\|_\infty = o_p(1)$ uniformly over $Q \in \mathcal{Q}$, by Assumption~\ref{a:primal_infeasible_gap_uniform_2}, we have that $\sup_{Q \in \mathcal{Q}}P_Q(\widehat{B}(X) \in \mathcal{D}(X;Q) \setminus \mathcal{B}^\ast(X;Q)) = o_p(1)$.

  This gives that $\hat{\theta}_L - \theta_L(Q) = \check{\theta}_L - \E_Q[\check{\theta}_L] + o_p(n^{-1/2})$ uniformly over $Q \in \mathcal{Q}$, where $\check{\theta}_L$ is defined as in the proof of Corollary~\ref{cor:direct_est_normal}. Under assumption \ref{a:rate_bfs_uniform},
  \[
    \lim_{n \to \infty}\sup_{Q \in \mathcal{Q}}d_{PR}\left(\sqrt{n}\left(\check{\theta}_L - \E_Q[\check{\theta}_L]\right),  \sqrt{V_L(Q)} Z\right) = 0,
  \] using the Prokhorov bound in \citet{vandervaart_weak_1996} Proposition A.5.2. The final result follows from Slutsky's theorem.

\end{proof}

\section{Simulation study}
\label{sec:sim}

For this simulation study, we consider a setting with a single binary treatment $D \in \{0,1\}$ and an ordinal outcome with $L$ levels scaled to be between 0 and 1: $Y \in \{0, 1 / (L-1), 2 / (L-1), \ldots,1\}$.
The goal is to estimate the probability that the observed outcome is worse than the best potential outcome,
$\theta = P(\max\{Y(1), Y(0)\} > Y) = \E[D \bbone\{Y(0) > Y(1)\} + (1 - D) \bbone\{Y(1) > Y(0)\}]$.
Note that this corresponds to a setting where both the conditional constraint and objective vectors are unknown.
Our focus is on understanding the relative performance between the de-biased primal-dual and entropic estimators as
(a) the level of regularization changes, (b) the sample size changes, (c) the estimation error of the nuisance functions changes, and (d) the complexity of the conditional linear program changes.
For the latter, we change the number of outcome levels. More levels of the outcome add more decision variables and constraints into the conditional linear program, which changes the complexity.

\begin{figure}[t]
  \centering
  \includegraphics[width=0.8\textwidth]{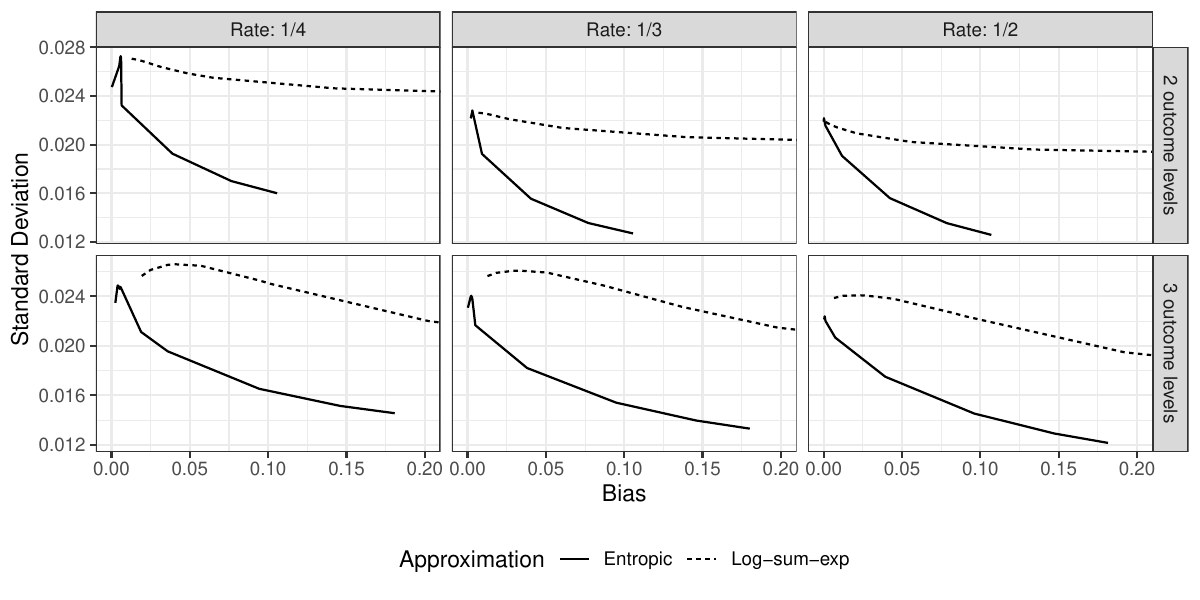}
  \caption{Bias versus standard deviation of the de-biased entropic estimator and log-sum-exp approximation of the lower bound as the hyper-parameter $\eta$ changes, for varying rates $r$ and number of outcome levels $L$, with a sample size of $n = 1000$.}
  \label{fig:bias_variance_plot}
\end{figure}

The data generating process for the simulation study is as follows. For each simulated dataset, a covariate $X_i$ is drawn independently from a standard normal distribution.
Control and treatment potential outcomes are drawn as $Y_i(0) = \varepsilon_i(0)$ and $Y_i(1) = \frac{1}{2.4}X^2 + 1.2X + \varepsilon_i(1)$, where $\varepsilon_i(0), \varepsilon_i(1)$ are multivariate normal random variables with mean zero, variance 1, and covariance $\rho = 0.9$, so the marginal variance of $Y_i(0)$ is 1 and the marginal variance of $Y_i(1)$ is 2.
Both $Y(0)$ and $Y(1)$ are then discretized into $L$ ordinal categories using quantile-based breaks from the standard normal distribution so that a roughly equal proportion of units have $Y_i(0) = \frac{\ell}{L-1}$ for each $\ell=0,\ldots,L-1$ and more units have values of $Y_i(1)$ closer to 0 and 1.
From this, we can analytically compute conditional outcome models $m_j(d, x) = P(Y = y_j \mid D = d, X = x)$ for each outcome level $y_j$.
Treatments are drawn with a probability equal to $e(x) = 1/(1 + \exp(x))$.
To control the estimation error of the nuisance functions in the simulations, we adapt the construction in \citet{levis_assisted_2023} and perturb the true outcome models with random noise $\log \tilde{m}_j(d, X_i) = \log m_j(d, X_i) + \varepsilon_{ijd}$, where $\varepsilon_{ijd}$ are independent Normal random variables with mean equal to $2.25 n^{-r}$ and variance equal to $2.25^2 n^{-2r}$, where $r$ is a parameter that controls the estimation error rate of the outcome model. The ``estimated'' outcome model is then normalized as $\hat{m}_j(d, X_i) = \frac{\exp\left(\log \tilde{m}_j(d, X_i)\right)}{\sum_{j'} \exp\left(\log \tilde{m}_{j'}(d, X_i)\right)}$. Similarly, we generate an ``estimated'' propensity score as $\hat{e}(X_i) = \left(1 + \exp\left(-\log(e(X_i) / (1  - e(X_i))) + \varepsilon_{ie} \right)\right)^{-1}$, where $\varepsilon_{ie}$ are also independent Normal random variables with mean equal to $2.25 n^{-r}$ and variance equal to $2.25^2 n^{-2r}$.

\begin{figure}[t]
  \centering
  \includegraphics[width=0.8\textwidth]{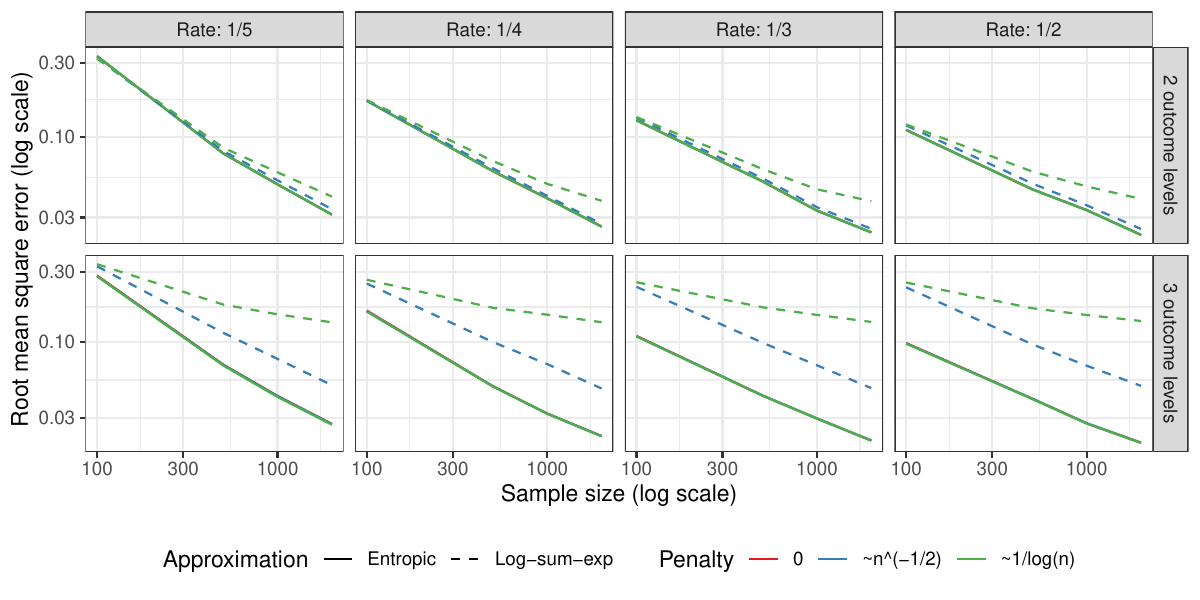}
  \caption{Root mean square error versus sample size of the de-biased primal-dual and entropic estimators, along with the log-sum-exp approximation, with penalties $\eta^{-1}$ and $\xi$ scaling with $\frac{1}{\sqrt{n}}$ and $\frac{1}{\log(n)}$,  for varying rates $r$ and number of outcome levels $L$ ($\eta^{-1} = 0$ corresponds to the de-biased primal-dual estimator).}
  \label{fig:rmse_plot}
\end{figure}

For different combinations of the sample size $n$, number of outcome levels $L$ and rates $r$, we fit the de-biased primal-dual and entropic estimator with a range of hyper-parameters to estimate the lower bound $\theta_L$.
We also compute estimates based on the log-sum-exp approximation discussed in the instrumental variables setting by \citet{levis_assisted_2023}.
In the general setting considered here for the lower bound, this would iterate over all basic feasible solutions $B \in \mathcal{B}(X)$ and compute the approximation as $
-\frac{1}{\xi}\log \sum_{B \in \mathcal{B}(X)} \exp\left(-\xi \left\<\bm{c}(x), A_B^{-1}\bm{b}(X)\right\> \right)$, where  $\xi$  is a hyperparameter that controls the approximation.
In this simulation, we  estimate this smooth approximation in a de-biased manner; however, note that it requires explicitly computing all basic feasible solutions.

Figure~\ref{fig:bias_variance_plot} shows the bias versus standard deviation of the de-biased entropic estimator and log-sum-exp approximation of the lower bound as the hyper-parameters $\eta$ and $\xi$ change, for varying rates $r$ and number of outcome levels $L$, with a sample size of $n = 1,000$.
This exhibits the trade-off between bias and variance for the de-biased entropic estimator as the level of regularization $\eta$ changes.
Across all specifications of the hyper-parameters, the de-biased entropic estimator pareto dominates the log-sum-exp approximation, with a more favorable bias-variance trade-off and lower minimal values of both the bias and variance, especially with $L = 3$ outcome levels.

\begin{figure}[h]
  \centering
  \includegraphics[width=0.8\textwidth]{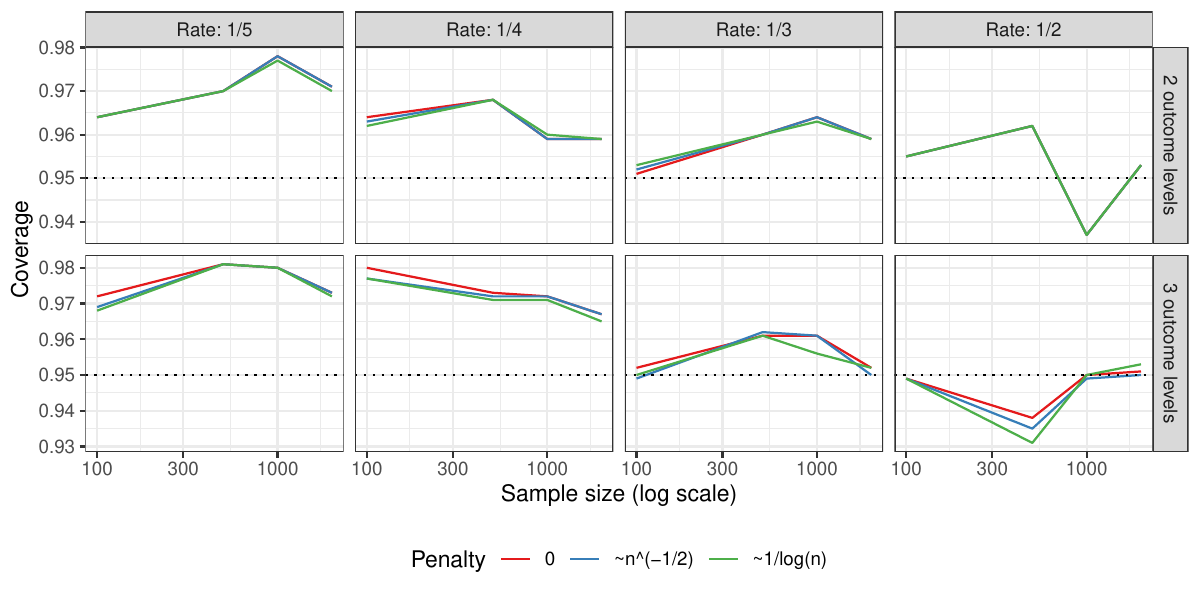}
    \caption{Coverage versus sample size of the de-biased primal-dual and entropic estimators, with penalty $\eta^{-1}$ scaling with $\frac{1}{\sqrt{n}}$ and $\frac{1}{\log(n)}$,  for varying rates $r$ and number of outcome levels $L$ ($\eta^{-1} = 0$ corresponds to the de-biased primal-dual estimator). Dotted line indicates nominal coverage of 95\%.}
  \label{fig:coverage_plot}
\end{figure}

Figure~\ref{fig:rmse_plot} shows the root mean square error (RMSE) of the de-biased primal-dual and entropic estimators, along with the log-sum-exp approximation, as the sample size and nuisance function estimation error rate change, using hyperparameters $\eta^{-1}$ and $\xi$ that scale with $\frac{1}{\sqrt{n}}$ and $\frac{1}{\log(n)}$.
We see that the de-biased entropic estimator is not particularly sensitive to the scaling of the hyper-parameter $\eta$: both choices behave roughly the same as the de-biased primal-dual estimator.
In addition, both the de-biased primal-dual and entropic estimators are robust to the estimation error rate of the nuisance functions, and decrease quickly with the sample size (empirically the RMSE decreases with the sample size at a $\frac{1}{\sqrt{n}}$ rate or faster).
In contrast, the log-sum-exp approximation is sensitive to the scaling of the hyper-parameter $\xi$. With three outcome levels, choosing $\xi \sim \frac{1}{\log n}$ leads to a slower rate of convergence, and while choosing $\xi \sim \frac{1}{\sqrt{n}}$ leads to a comparable rate of convergence as the de-biased primal-dual and entropic estimators, the RMSE is still larger at every sample size.
Figure~\ref{fig:coverage_plot} shows that the empirical coverage of the one-sided confidence intervals for the de-biased primal-dual and entropic estimators is close to the nominal level.

\begin{figure}[t]
  \centering
  \includegraphics[width=0.8\textwidth]{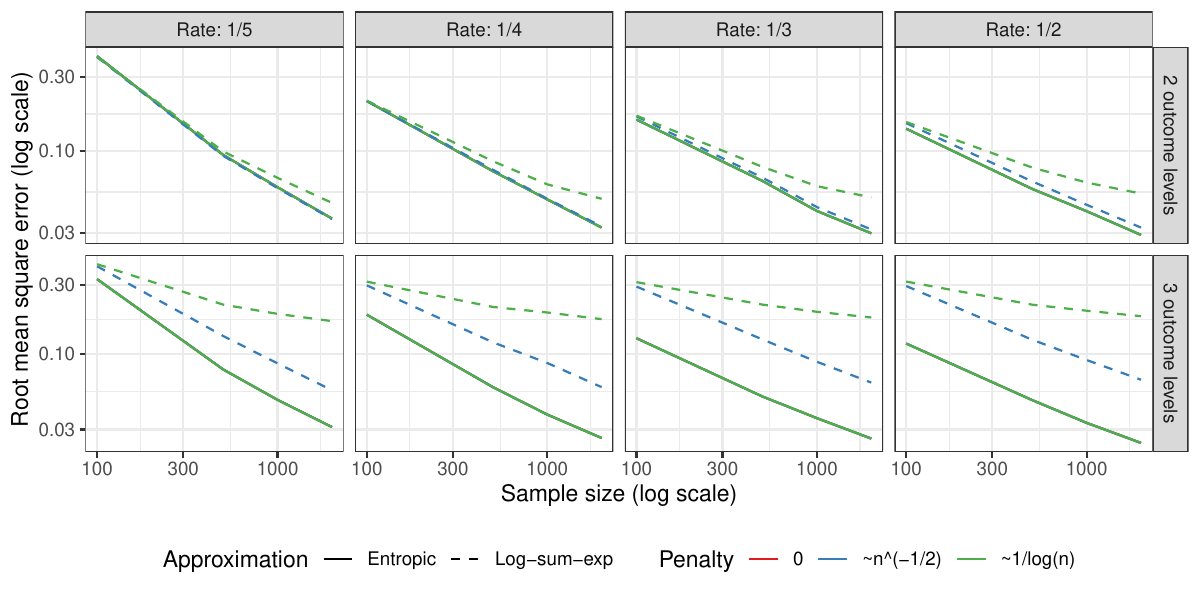}
  \caption{Root mean square error versus sample size of the de-biased primal-dual and entropic estimators, along with the log-sum-exp approximation, with penalties $\eta^{-1}$ and $\xi$ scaling with $\frac{1}{\sqrt{n}}$ and $\frac{1}{\log(n)}$,  for an identified estimand, varying rates $r$ and number of outcome levels $L$ ($\eta^{-1} = 0$ corresponds to the de-biased primal-dual estimator).}
  \label{fig:rmse_plot_id}
\end{figure}

Finally, we also consider an identified estimand, the ATT $\theta = \E[e(X) (Y(1) - Y(0))] / \E[e(X)]$.
Figure~\ref{fig:rmse_plot_id} shows the RMSE versus the sample size for varying rates and number of outcome levels.
Because this estimand is point-identified, any feasible solution is optimal, and so the solutions are not unique.
However, there are no degeneracy issues and the objective vector is known up to a scaling by the propensity score $e(x)$
As we expect, the entropic estimators return numerically the same result as the de-biased primal-dual estimator regardlesss of the level of regularization. We also see that both the de-biased primal-dual and entropic estimators still perform better than the log-sum-exp approximation when the number of outcome levels is larger.

\section{Estimating conditional bounds}
\label{sec:condl_bounds}
In addition to estimating bounds on the partially identified parameter $\theta$, we are often interested in understanding heterogeneity in the parameter across the covariate space. This can be important, for instance, when estimating treatment rules with partial identification, where various forms of optimality (e.g. minimax, maximin) are defined in terms of bounds on a conditional parameter \citep{Pu2021, benmichael_asymmetric_2024}.

Let $V \in \mathcal{V} \subseteq \mathcal{X}$ denote a (subset of) covariates of interest, and let $\theta(v) \equiv \E[\<\bm{c}(X), \bm{p}^\ast(X)\> \mid V = v]$ be the conditional parameter. Then, our interest is in estimating the conditional lower bound:
\begin{equation}
  \label{eq:partial_id_condl}
  \theta_L(v) \equiv \E\left[\min_{\bm{p} \in \mathcal{P}(x)} \<\bm{c}(X), \bm{p}\> \mid V = v\right]
\end{equation}
To estimate the conditional bound, we can adjust the de-biased primal-dual and entropic estimators in Sections~\ref{sec:primal_estimation} and \ref{sec:entropic} to condition on the covariates $V$ by regressing \emph{pseudo-outcomes} on the covariates $V$.
For example, for the primal-dual estimator in Section~\ref{sec:primal_estimation}, we can estimate the conditional lower bound as
\begin{equation}
  \label{eq:primal_bound_est_condl}
  \begin{aligned}
    \hat{\theta}_L(v) &= \widehat{\E}\left[ \hat{\theta}_L(X) + \< \hat{\bm{\varphi}}^{(c)}(O), \hat{\bm{p}}(X)\> + \< \hat{\bm{\lambda}}(X), \hat{\bm{\varphi}}^{(b)}(O)\>\mid V = v\right],
  \end{aligned}
\end{equation}
and for the entropic estimator in Section~\ref{sec:entropic}, we can estimate the conditional lower bound as
\begin{equation}
  \label{eq:entropic_bound_est_condl}
  \begin{aligned}
    \hat{\theta}_L^\eta(v) &= \widehat{\E}\left[\hat{\theta}_L^\eta(X) + \left\< \nabla_{\bm{b}} \hat{\theta}_L^\eta(X), \hat{\bm{\varphi}}^{(b)}(O) \right\> + \<\nabla_{\bm{c}} \hat{\theta}^\eta_L(X), \hat{\bm{\varphi}}^{(c)}(O)\>\mid V = v\right],
  \end{aligned}
\end{equation}
where $\widehat{\E}[f(X) \mid V = v]$ denotes an estimator for the conditional expectation of $f(X)$ given $V = v$.

Broadly, pseudo-outcome-based estimators such as these inherit the properties of the de-biased estimators off which they are based \citep{Kennedy2022_drlearner}.
A consequence of the results in Theorems~\ref{thm:direct_est_rate} and \ref{thm:entropic_est_approx} and the general analysis of pseudo-outcome-based estimators in \citet{Kennedy2022_drlearner} is that the error for the conditional bounds estimators depend on the errors in the nuisance function estimates as in  Theorems~\ref{thm:direct_est_rate}  and \ref{thm:entropic_est_approx} as well as an ``oracle'' term that depends on the estimation error of the conditional expectation $\widehat{\E}[\cdot \mid V = v]$.

\begin{figure}[t]
  \centering
  \begin{subfigure}[t]{0.45\textwidth}
    \includegraphics[width=\textwidth]{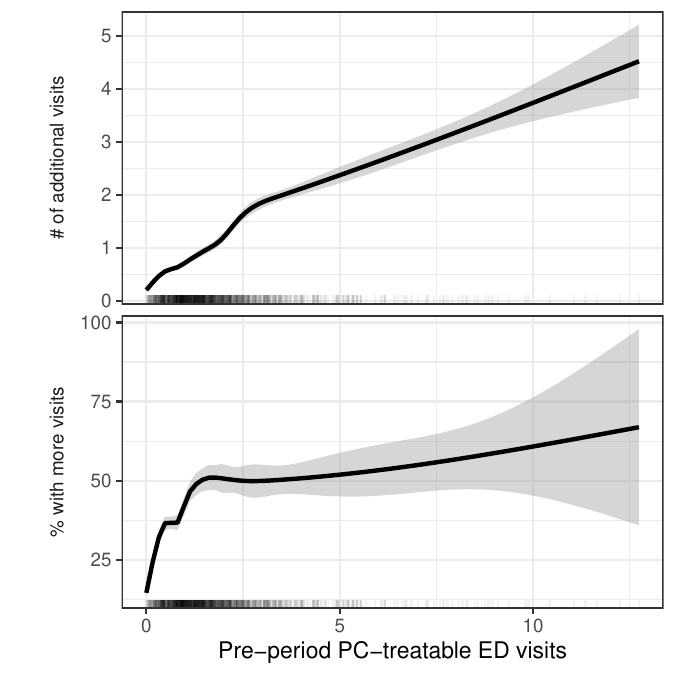}
    \caption{Randomizing access.}
    \label{fig:ed_het_plot_2}
  \end{subfigure}
  \begin{subfigure}[t]{0.45\textwidth}
    \includegraphics[width=\textwidth]{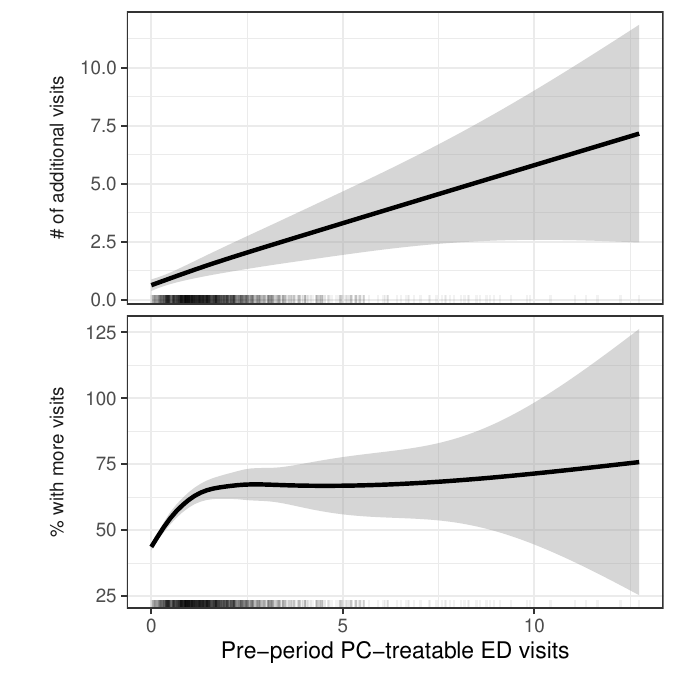}
    \caption{Medicaid enrollment if off waitlist.}
    \label{fig:ed_het_plot_2_iv}
  \end{subfigure}
  \caption{Conditional upper bounds on the expected number of additional PC-treatable ED visits and the percentage of individuals with any more PC-treatable ED visits versus the (modeled) number of PC-treatable ED visits during the pre-period for (a) randomizing access to Medicaid vs providing access based on potential outcomes to minimize ED visits, and (b) natural Medicaid uptake conditional on access vs enforcing Medicaid coverage based on potential outcomes to minimize ED visits. Thick dark lines are GAM estimates of the conditional lower and upper bounds; shaded regions correspond to GAM-based 95\% confidence intervals to show approximate level of uncertainty.}
  \label{fig:het_plot2}
\end{figure}

Figure~\ref{fig:het_plot2} shows estimates of the bounds from the empirical analysis in Section~\ref{sec:ohie} conditional on a modeled estimate of the number of primary care-treatable ED visits in the pre-period, using the pseudo-outcome regression approach. We see that for individuals with few to no modeled primary-care-treatable visits, the upper bounds are much smaller, but increase rapidly.
There is a larger degree of uncertainty---and consequently more potential upside for optimizing the waitlist and enrollment policies---for individuals with more pre-period ED visits.

\section{Additional results and proofs}

\subsection{Auxiliary lemmas and additional results}

\begin{lemma}
  \label{lem:plugin_basis_generic}
  Under Assumption~\ref{a:riesz},
  \[
   -\|\hat{\bm{b}} - \bm{b}\|_\infty \lesssim \left\<\bm{c}(X), \left(A_{\widehat{B}(X)}^{-1} - A_{B^\ast(X)}^{-1}\right) \bm{b}(X)\right\> \lesssim \|\hat{\bm{c}} - \bm{c}\|_\infty
  \]
\end{lemma}

\begin{proof}[Proof of Lemma \ref{lem:plugin_basis_generic}]
  We begin with the lower bound.
  Define the pseudo-true primal and dual variables that use the true objective vector $\bm{c}(x)$ and the estimated constraint vector $\hat{\bm{b}}(x)$ as the solutions to
  \begin{align*}
  \tilde{\bm{p}}(x) \in \argmin_{\bm{p} \geq 0 \mid A\bm{p} = \hat{\bm{b}}(x)} \<\bm{c}(x), \bm{p}\>, \quad
  \tilde{\bm{\lambda}}(x) \in \argmax_{\bm{\lambda} \mid A'\bm{\lambda} \leq \bm{c}(x)} \<\bm{\lambda}, \hat{\bm{b}}(x)\>.
\end{align*}
Now we have the decomposition
\begin{align*}
  \left\<\bm{c}(X), (A_{\widehat{B}(X)}^{-1} - A_{B^\ast(X)}^{-1})\bm{b}(X)\right\> & = \left\<\bm{c}(X), A_{\widehat{B}(X)}^{-1} \hat{\bm{b}}(X)\right\> - \left\<\bm{c}(X), \bm{p}^\ast(X)\right\> + \left\<\bm{c}(X), A_{\widehat{B}(X)}^{-1} \left(\bm{b}(X) - \hat{\bm{b}}(X)\right)\right\>\\
  & = \left\<\bm{c}(X), \hat{\bm{p}}(X)\right\>   - \left\<\bm{\lambda}^\ast(X), \bm{b}(X)\right\>  + \left\<\bm{c}(X), A_{\widehat{B}(X)}^{-1} \left(\bm{b}(X) - \hat{\bm{b}}(X)\right)\right\>.
\end{align*}

 Next, note that
\begin{align*}
  \left\<\bm{\lambda}^\ast(X), \bm{b}(X)\right\>  & = \left\<\bm{\lambda}^\ast(X) - \tilde{\bm{\lambda}}(X), \bm{b}(X)\right\> + \left\<\tilde{\bm{\lambda}}(X), \bm{b}(X)\right\>\\
  & =   \left\<\bm{\lambda}^\ast(X) - \tilde{\bm{\lambda}}(X), \hat{\bm{b}}(X)\right\> + \left\<\bm{\lambda}^\ast(X) - \tilde{\bm{\lambda}}(X), \bm{b}(X) - \hat{\bm{b}}(X)\right\> + \left\<\tilde{\bm{\lambda}}(X), \bm{b}(X)\right\>\\
  & \leq\left\<\bm{\lambda}^\ast(X) - \tilde{\bm{\lambda}}(X), \bm{b}(X) - \hat{\bm{b}}(X)\right\> + \left\<\tilde{\bm{\lambda}}(X), \bm{b}(X)\right\>
\end{align*}
where we have used the fact that $\tilde{\bm{\lambda}}(X)$ and $\bm{\lambda}^\ast(X)$ are both feasible for the pseudo-true dual problem, but $\tilde{\bm{\lambda}}(X)$ is optimal for it.
At the same time, the estimated primal variable $\hat{\bm{p}}(X)$ and the pseudo-true primal variable $\tilde{\bm{p}}(X)$ have the same constraint set, so $\left\<\bm{c}(X), \hat{\bm{p}}(X)\right\> \geq \left\<\bm{c}(X), \tilde{\bm{p}}(X)\right\> = \left\<\tilde{\bm{\lambda}}(X), \hat{\bm{b}}(X)\right\>$.
Combining these bounds gives that

\begin{align*}
  \left\<\bm{c}(X), ( A_{\widehat{B}(X)}^{-1} - A_{B^\ast(X)}^{-1})\bm{b}(X)\right\>  & \geq \left\<\tilde{\bm{\lambda}}(X), \hat{\bm{b}}(X)\right\>  - \left\<\bm{\lambda}^\ast(X), \bm{b}(X)\right\>  + \left\<\bm{c}(X), A_{\widehat{B}(X)}^{-1} \left(\bm{b}(X) - \hat{\bm{b}}(X) \right)\right\>\\
  & \geq  \left\<\bm{\lambda}^\ast(X) - A_{\widehat{B}(X)}^{-1 \prime}\bm{c}(X), \left(\hat{\bm{b}}(X) - \bm{b}(X)\right)\right\>\\
  & \geq -C_0 \|\hat{\bm{b}} - \bm{b}\|_\infty,
\end{align*}
where $C_0 \equiv 2 \sup_{x \in \mathcal{X}}\max_{B \in \mathcal{B}}\|A^{-1\prime}_B \bm{c}(x)\|_1 < \infty$ is bounded because $A_{B}$ is full rank for all $x$, there are a finite number of bases, and $\bm{c}(x)$ is bounded.

Moving to the upper bound, define pseudo-true primal and dual variables that use the estimated objective vector $\hat{\bm{c}}(x)$ and the true constraint vector $\bm{b}(x)$ as the solutions to
\begin{align*}
  \check{\bm{p}}(x) \in \argmin_{\bm{p} \geq 0 \mid A\bm{p} = \bm{b}(x)} \<\hat{\bm{c}}(x), \bm{p}\>, \quad
  \check{\bm{\lambda}}(x) \in \argmax_{\bm{\lambda} \mid A'\bm{\lambda} \leq \hat{\bm{c}}(x)} \<\bm{\lambda}, \bm{b}(x)\>.
\end{align*}
Now we can decompose
\begin{align*}
  \<\bm{c}(X), (A_{\widehat{B}(X)}^{-1} - A_{B^\ast(X)}^{-1})\bm{b}(X)\> & = \<\hat{\bm{c}}(X), A_{\widehat{B}(X)}^{-1} \bm{b}(X)\> - \<\bm{c}(X), \bm{p}^\ast(X)\>  + \<\bm{c}(X) - \hat{\bm{c}}(X), A_{\widehat{B}(X)}^{-1} \bm{b}(X)\>\\
  & = \<\hat{\bm{\lambda}}(X), \bm{b}(X)\> - \<\bm{c}(X), \bm{p}^\ast(X)\>  + \<\bm{c}(X) - \hat{\bm{c}}(X), A_{\widehat{B}(X)}^{-1} \bm{b}(X)\>
\end{align*}
Now, because the estimated dual variable $\hat{\bm{\lambda}}(x)$ and the pseudo-true dual variable $\check{\bm{\lambda}}(x)$ have the same constraint set, $ \<\hat{\bm{\lambda}}(X), \bm{b}(X)\> \leq  \<\check{\bm{\lambda}}(X), \bm{b}(X)\> = \<\hat{\bm{c}}(X), \check{\bm{p}}(X)\>$. Next, note that
\begin{align*}
  \<\bm{c}(X), \bm{p}^\ast(X)\> & = \<\bm{c}(X), \bm{p}^\ast(X) - \check{\bm{p}}(X)\> + \<\bm{c}(X), \check{\bm{p}}(X)\>\\
  & = \<\bm{c}(X) - \hat{\bm{c}}(X), \bm{p}^\ast(X) - \check{\bm{p}}(X)\> + \<\hat{\bm{c}}(X), \bm{p}^\ast(X) - \check{\bm{p}}(X)\> + \<\bm{c}(X), \check{\bm{p}}(X)\>\\
  & \geq  \<\bm{c}(X) - \hat{\bm{c}}(X), \bm{p}^\ast(X) - \check{\bm{p}}(X)\> + \<\bm{c}(X), \check{\bm{p}}(X)\>,
\end{align*}
where we have used the fact that $\check{\bm{p}}(X)$ and $\bm{p}^\ast(X)$ are both feasible for the pseudo-true primal problem, but $\check{\bm{p}}(X)$ is optimal for it.
Together, these two bounds imply that
\begin{align*}
  \<\bm{c}(X), (A_{\widehat{B}(X)}^{-1} - A_{B^\ast(X)}^{-1})\bm{b}(X)\> & \leq \<\hat{\bm{c}}(X), \check{\bm{p}}(X)\> - \<\bm{c}(X), \bm{p}^\ast(X)\>  + \<\bm{c}(X) - \hat{\bm{c}}(X), A_{\widehat{B}(X)}^{-1} \bm{b}(X)\>\\
  & \leq \<\bm{c}(X) - \hat{\bm{c}}(X), A_{\widehat{B}(X)}^{-1} \bm{b}(X) - \bm{p}^\ast(X)\>\\
  & \leq C_1 \|\hat{\bm{c}} - \bm{c}\|_\infty,
\end{align*}
where $C_1 = 2\sup_{x \in \mathcal{X}}\max_{B \in \mathcal{B}}\|A^{-1}_{B} \bm{b}(x)\|_1$ is also bounded because $A_{B}$ is full rank for all $x$, there are a finite number of bases, and $\bm{b}(x)$ is bounded.

\end{proof}

\begin{lemma}
  \label{lem:dual_feasible}
  Under Assumption~\ref{a:infeasible_gap}, if $\hat{\bm{b}}$ and $\hat{\bm{c}}$ are fit on a separate, independent sample, then the probability that an estimated optimal basis $\widehat{B}(X)$ is not dual feasible is bounded by
  \[
    P\left(\widehat{B}(X) \not \in \mathcal{D}(X)\right) \lesssim \|\hat{\bm{c}} - \bm{c}\|_\infty^\alpha.
  \]

\end{lemma}

\begin{proof}[Proof of Lemma~\ref{lem:dual_feasible}]
  Because $\widehat{B}(x)$ is feasible for the estimated dual problem, if $\widehat{B}$ is not feasible for the true dual problem, then for every feasibility violating index $k$ we have that
  \begin{align*}
    \bm{e}_k'(A'A_{\widehat{B}(X)}^{-1\prime} - I_K)\hat{\bm{c}}(X) &\leq 0\\
    \bm{e}_k'(A'A_{\widehat{B}(X)}^{-1\prime} - I_K)\bm{c}(X) &> 0,
  \end{align*}
  where $\bm{e}_k$ is the $k$\super{th} standard basis vector. This implies that the violation of the dual constraint is at most
  \[
    0 < \bm{e}_k'(A'A_{\widehat{B}(X)}^{-1\prime} - I_K)\bm{c}(X) \leq \bm{e}_k'(A'A_{\widehat{B}(X)}^{-1\prime} - I_K)(\bm{c}(X) - \hat{\bm{c}}(X)) \leq \|A'A_{\widehat{B}(X)}^{-1\prime} - I_K\|_\infty \|\bm{c}(X) - \hat{\bm{c}}(X)\|_\infty.
  \]
  Letting $C_0 \equiv \max_{B \in \mathcal{B}}\|A'A_{B}^{-1 \prime} - I_K\|_\infty < \infty$, this implies that the probability that $\widehat{B}(X)$ is not dual feasible is at most the probability that the smallest violation is smaller than $C_0 \|\hat{\bm{c}} - \bm{c}\|_\infty$, i.e.,
  \[
    P\left(\widehat{B}(X) \not \in \mathcal{D}(X)\right) \leq P\left(\mathcal{E}_\lambda(X) \leq  C_0\|\hat{\bm{c}} - \bm{c}\|_\infty\right) \lesssim \|\hat{\bm{c}} - \bm{c}\|_\infty^\alpha,
  \]
  where the second bound comes from Assumption~\ref{a:infeasible_gap}.

\end{proof}

\begin{lemma}
  \label{lem:dual_optimal}
  Under Assumptions~\ref{a:riesz}, \ref{a:rates}, and \ref{a:primal_infeasible_gap}, if $\hat{\bm{b}}$ and $\hat{\bm{c}}$ are fit on a separate, independent sample, then the probability that an estimated basis $\widehat{B}(X)$ is dual feasible but not dual optimal is bounded by
  \[
    P\left(\widehat{B}(X) \in D^{-}(X)\right) \lesssim \|\hat{\bm{b}} - \bm{b}\|_\infty^\alpha
  \]
\end{lemma}
\begin{proof}[Proof of Lemma~\ref{lem:dual_optimal}]
  Note that if $\widehat{B}(x)$ is dual-feasible but not dual-optimal then it is not feasible for the true primal problem, i.e. there exists some index $k \in \widehat{B}(x)$ such that $\bm{e}_k'A_{\widehat{B}(X)}^{-1} \bm{b}(X) < 0$.
  On the other hand, by construction $\widehat{B}(x)$ is feasible for the estimated primal problem, so $A_{\widehat{B}(x)}^{-1} \hat{\bm{b}}(X) \geq 0$.
  This implies that for every index $k$ where $\bm{e}_k'A_{\widehat{B}(X)}^{-1} \bm{b}(X) < 0$, we have
  \[
    0 > \bm{e}_k'A_{\widehat{B}(X)}^{-1} \bm{b}(X) \geq \bm{e}_k'A_{\widehat{B}(X)}^{-1} (\bm{b}(X) - \hat{\bm{b}}(X)) \geq -\|A_{\widehat{B}(X)}^{-1}\|_\infty \|\bm{b}(X) - \hat{\bm{b}}(X)\|_\infty.
  \]
  Again because $\|A^{-1}_B\|_\infty$ is bounded for each of the finite number of bases $B \in \mathcal{B}$, this implies that the probability that $\widehat{B}(X)$ is dual feasible but not dual optimal is at most the probability that the smallest primal feasibility violation is smaller than $C_0 \|\hat{\bm{b}} - \bm{b}\|_\infty$ for a constant $C_0$, i.e.,
  \[
    P\left(\widehat{B}(X) \in \mathcal{D}^{-}(X)\right) \leq P\left(\mathcal{E}_p(X) \leq C_0 \|\hat{\bm{b}} - \bm{b}\|_\infty\right) \lesssim \|\hat{\bm{b}} - \bm{b}\|_\infty^\alpha,
  \]
  where the second bound comes from Assumption~\ref{a:primal_infeasible_gap}.
\end{proof}

\begin{lemma}
  \label{lem:validity_and_sharpness}
  Under Assumptions~\ref{a:riesz}, \ref{a:rates}, and \ref{a:infeasible_gap}, if $\hat{\bm{b}}$ and $\hat{\bm{c}}$ are fit on a separate, independent sample, then the plugin basic feasible solution satisfies
  \[
     -\|\hat{\bm{b}} - \bm{b}\|_\infty (1 + \|\hat{\bm{c}} - \bm{c}\|_\infty^{\alpha}) \lesssim \E\left[\<\bm{c}(X), \left(A_{\widehat{B}(X)}^{-1} - A_{B^\ast(X)}^{-1}\right) \bm{b}(X)\>\right] \lesssim \|\hat{\bm{c}} - \bm{c}\|_\infty^{1 + \alpha}.
  \]
  Furthermore, if Assumption~\ref{a:primal_infeasible_gap} also holds,
    \[
    -\|\hat{\bm{b}} - \bm{b}\|_\infty^{1 + \alpha} - \|\hat{\bm{b}} - \bm{b}\|_\infty\|\hat{\bm{c}} - \bm{c}\|_\infty^{\alpha} \lesssim \E\left[\<\bm{c}(X), \left(A_{\widehat{B}(X)}^{-1} - A_{B^\ast(X)}^{-1}\right) \bm{b}(X)\>\right] \lesssim \|\hat{\bm{c}} - \bm{c}\|_\infty^{1 + \alpha}.
  \]
\end{lemma}

\begin{proof}[Proof of Lemma~\ref{lem:validity_and_sharpness}]
  For the upper bound, we decompose
\begin{align*}
     \E\left[\<\bm{c}(X), \left(A_{\widehat{B}(X)}^{-1} - A_{B^\ast(X)}^{-1}\right)\bm{b}(X)\>\right]&   =  \E\left[\bbone\{\widehat{B}(X) \in \mathcal{D}(X)\}\<\bm{c}(X), \left(A_{\widehat{B}(X)}^{-1}  -  A_{B^\ast(X)}^{-1}\right)\bm{b}(X)\>\right] \tag{$\ast$}\\
    & + \E\left[\bbone\{\widehat{B}(X) \not \in \mathcal{D}(X)\}\<\bm{c}(X), \left(A_{\widehat{B}(X)}^{-1}  -  A_{B^\ast(X)}^{-1}\right)\bm{b}(X)\>\right] \tag{$\ast\ast$}
\end{align*}

If $\widehat{B}(X) \in \mathcal{D}(X)$, then $A_{\widehat{B}(X)}^{-1 \prime} \bm{c}(X)$ is feasible for the true dual problem. Because $ A_{B^\ast(X)}^{-1 \prime} \bm{c}(X)$ is a maximizer of the dual problem, this implies that
\[
  \<\bm{c}(X), (A_{\widehat{B}(X)}^{-1} - A_{B^\ast(X)}^{-1})\bm{b}(X)\> \leq 0
\]

For term ($\ast\ast$), by Lemmas~\ref{lem:plugin_basis_generic} and \ref{lem:dual_feasible}, we have that
\begin{align*}
  \E\left[\bbone\{\widehat{B}(X) \not \in \mathcal{D}(X)\}\<\bm{c}(X), \left(A_{\widehat{B}(X)}^{-1}  -  A_{B^\ast(X)}^{-1}\right)\bm{b}(X)\>\right] & \leq C_1 \|\hat{\bm{c}} - \bm{c}\|_\infty P\left(\widehat{B}(X) \not \in \mathcal{D}(X)\right)\\
  & \lesssim \|\hat{\bm{c}} - \bm{c}\|_\infty^{1 + \alpha}.
\end{align*}

For the lower bound, note that if the plugin basic feasible solution is dual-feasible and dual-optimal, then $\<\bm{c}(X), (A_{\widehat{B}(X)}^{-1} - A_{B^\ast(X)}^{-1})\bm{b}(X)\> = 0$.
So, we can further decompose the expectation as
\begin{align*}
     \E\left[\<\bm{c}(X), \left(A_{\widehat{B}(X)}^{-1} - A_{B^\ast(X)}^{-1}\right)\bm{b}(X)\>\right]&   =  \E\left[\bbone\{\widehat{B}(X) \in \mathcal{D}^{-}(X)\}\left\<\bm{c}(X), \left(A_{\widehat{B}(X)}^{-1}  -  A_{B^\ast(X)}^{-1}\right)\bm{b}(X)\right\>\right]\\
    & \quad + \E\left[\bbone\{\widehat{B}(X) \not \in \mathcal{D}(X)\}\left\<\bm{c}(X), \left(A_{\widehat{B}(X)}^{-1}  -  A_{B^\ast(X)}^{-1}\right)\bm{b}(X)\right\>\right]\\
    & \geq -C_0 \|\hat{\bm{b}} - \bm{b}\|_\infty\left( P\left(\widehat{B}(X) \in \mathcal{D}^{-}(X)\right)  +  P\left(\widehat{B}(X) \not \in \mathcal{D}(X)\right) \right),
\end{align*}
where the lower bound uses Lemma~\ref{lem:plugin_basis_generic}. Without any further assumptions, we can bound $ P\left(\widehat{B}(X) \in \mathcal{D}^{-}(X)\right) \leq 1$ (i.e. almost all dual feasible solutions are sub-optimal in the worst case) and Lemma~\ref{lem:dual_feasible} to bound $P\left(\widehat{B}(X) \not \in \mathcal{D}(X)\right) \lesssim \|\hat{\bm{c}} - \bm{c}\|_\infty^{\alpha}$ to get the lower bound:
\[
  \E\left[\<\bm{c}(X), \left(A_{\widehat{B}(X)}^{-1} - A_{B^\ast(X)}^{-1}\right)\bm{b}(X)\>\right] \gtrsim - \|\hat{\bm{b}} - \bm{b}\|_\infty (1 + \|\hat{\bm{c}} - \bm{c}\|_\infty^{\alpha}).
\]
If in addition Assumption~\ref{a:primal_infeasible_gap} holds, then we can use Lemma~\ref{lem:dual_optimal} to get the lower bound
\[
  \E\left[\<\bm{c}(X), \left(A_{\widehat{B}(X)}^{-1} - A_{B^\ast(X)}^{-1}\right)\bm{b}(X)\>\right] \gtrsim -\|\hat{\bm{b}} - \bm{b}\|_\infty^{1 + \alpha} - \|\hat{\bm{b}} - \bm{b}\|_\infty\|\hat{\bm{c}} - \bm{c}\|_\infty^{\alpha}.
\]

\end{proof}

\begin{lemma}
  \label{lem:direct_est_bias}
    Under Assumptions~\ref{a:riesz}, \ref{a:rates}, and \ref{a:infeasible_gap}, if $\hat{\bm{b}}$ and $\hat{\bm{c}}$ are fit on a separate, independent sample, then
    \begin{align*}
        -\|\hat{\bm{b}} - \bm{b}\|_\infty(1 + \|\hat{\bm{c}} - \bm{c}\|_\infty^{\alpha}) - r_n - \|\hat{\bm{b}} - \bm{b}\|_2 \|\hat{\bm{c}} - \bm{c}\|_2 \lesssim \E[\hat{\theta}_L] - \theta_L & \lesssim \|\hat{\bm{c}} - \bm{c}\|_\infty^{1 + \alpha} +  r_n + \|\hat{\bm{b}} - \bm{b}\|_2 \|\hat{\bm{c}} - \bm{c}\|_2.
    \end{align*}
    Furthermore, if Assumption~\ref{a:primal_infeasible_gap} also holds, then
    \begin{align*}
        -\|\hat{\bm{b}} - \bm{b}\|_\infty(\|\hat{\bm{b}} - \bm{b}\|_\infty^{\alpha} + \|\hat{\bm{c}} - \bm{c}\|_\infty^{\alpha}) - r_n - \|\hat{\bm{b}} - \bm{b}\|_2 \|\hat{\bm{c}} - \bm{c}\|_2 & \\
         \lesssim \E[\hat{\theta}_L] - \theta_L & \lesssim \|\hat{\bm{c}} - \bm{c}\|_\infty^{1 + \alpha} +  r_n + \|\hat{\bm{b}} - \bm{b}\|_2 \|\hat{\bm{c}} - \bm{c}\|_2.
    \end{align*}
\end{lemma}

\begin{proof}[Proof of Lemma~\ref{lem:direct_est_bias}]

  This follows directly from Lemma~\ref{lem:direct_est_bias_basis} and Lemma~\ref{lem:validity_and_sharpness}. The upper bound follows from the upper bound in Lemma~\ref{lem:validity_and_sharpness} and the $O(r_n + \|\hat{\bm{b}} - \bm{b}\|_2 \|\hat{\bm{c}} - \bm{c}\|_2)$ term in Lemma~\ref{lem:direct_est_bias_basis}. The lower bound follows from the lower bound in Lemma~\ref{lem:validity_and_sharpness} and the $O(r_n + \|\hat{\bm{b}} - \bm{b}\|_2 \|\hat{\bm{c}} - \bm{c}\|_2)$ term in Lemma~\ref{lem:direct_est_bias_basis}.
\end{proof}

\begin{lemma}
  \label{lem:prob_subopt_converges}
  Under Assumption~\ref{a:rates} the probability that the estimated basis $\widehat{B}(x)$ is not optimal converges in probability to zero, i.e.
  \[
    P\left(\widehat{B}(X) \not \in \mathcal{B}^\ast(X)\right) = o_p(1).
  \]
\end{lemma}
\begin{proof}[Proof of Lemma~\ref{lem:prob_subopt_converges}]
  Note that if the estimated basis $\widehat{B}(x)$ is not an optimal, feasible basis, then it is either not primal feasible or not dual feasible. Therefore,
  \begin{align*}
    P\left(\widehat{B}(X) \not \in \mathcal{B}^\ast(X)\right) & = P(\widehat{B}(X) \not \in \mathcal{B}(X) \cup \widehat{B}(X) \not \in \mathcal{D}(X))\\
    & \leq P(\widehat{B}(X) \not \in \mathcal{B}(X)) + P(\widehat{B}(X) \not \in \mathcal{D}(X))
  \end{align*}
  From the proof of Lemma~\ref{lem:dual_feasible}, we have that $P(\widehat{B}(X) \not \in \mathcal{D}(X)) \leq P(\mathcal{E}_\lambda(X) \leq C_0 \|\hat{\bm{c}} - \bm{c}\|_\infty)$. Because by definition $\mathcal{E}_\lambda(X) > 0$, the function $P(\mathcal{E}_\lambda(X) \leq t)$ is deterministic and decreases to 0 as $t \to 0$. Since $\hat{\bm{c}}$ is fit on a separate, independent sample and $\|\hat{\bm{c}} - \bm{c}\|_\infty = o_p(1)$ by Assumption~\ref{a:rates}, this implies that $P(\widehat{B}(X) \not \in \mathcal{D}(X)) = o_p(1)$.

  Similarly, following the same argument as in Lemma~\ref{lem:dual_optimal}, if $\widehat{B}(x)$ is not primal feasible but is feasible for the estimated primal problem, then for every feasibility violating index $k$ we have that  $\bm{e}_k'A_{\widehat{B}(X)}^{-1} \bm{b}(X) < 0$ and $A_{\widehat{B}(x)}^{-1} \hat{\bm{b}}(x) \geq 0$. Defining $\mathcal{E}_p^\text{all}(x) = \min_{B \in \mathcal{B}\setminus\mathcal{B}(X)}  \|[A_B^{-1} \bm{b}(x)]_-\|_\infty$ as the primal infeasibility gap for all primal infeasible bases, with $\mathcal{E}_p^\text{all}(x) = \infty$ when all bases are feasible, we can follow the proof of Lemma~\ref{lem:dual_optimal} to see that implies that
  \[
    P(\widehat{B}(X) \not \in \mathcal{B}(X)) \leq P\left(  \mathcal{E}_p^\text{all}(X) \leq C_0 \|\hat{\bm{b}} - \bm{b}\|_\infty\right),
  \]
  for some constant $C_0$. As above, since $\mathcal{E}_p^\text{all}(X) > 0$ by definition and $\|\hat{\bm{b}} - \bm{b}\|_\infty = o_p(1)$ by Assumption~\ref{a:rates}, this implies that $P(\widehat{B}(X) \not \in \mathcal{B}(X)) = o_p(1)$.

\end{proof}

\begin{lemma}
  \label{lem:jacobian}
  Define $Q(\bm{p}) =  \left(A \;\text{diag}(\bm{p})A'\right)^{-1}$, where $\text{diag}(\bm{p})$ denotes the diagonal matrix with the vector $\bm{p}$ on the diagonal. The Jacobians of $\bm{p}^\eta(\bm{b}, \bm{c})$ with respect to a strictly feasible $\bm{b} \in \R^J$ and $\bm{c} \in \R^K$ are:
    \begin{align*}
    \nabla_{\bm{b}} \bm{p}^\eta(\bm{b}, \bm{c}) & = \text{diag}(\bm{p}^\eta(\bm{b}, \bm{c})) A'Q(\bm{p}^\eta(\bm{b}, \bm{c}))\\
     \nabla_{\bm{c}} \bm{p}^\eta(\bm{b}, \bm{c}) & = \eta \ \text{diag}(\bm{p}^\eta(\bm{b}, \bm{c})) A'Q(\bm{p}^\eta(\bm{b},\bm{c}))A \text{diag}(\bm{p}^\eta(\bm{b}, \bm{c})) - \text{diag}(\bm{p}^\eta(\bm{b}, \bm{c})),
  \end{align*}
  The Jacobians of $\theta^\eta_L(\bm{b}, \bm{c})$ with respect to a strictly feasible $\bm{b} \in \R^J$ and $\bm{c} \in \R^K$ are
  \begin{align*}
    \nabla_{\bm{b}} \theta^\eta_L(\bm{b}, \bm{c}) & = Q(\bm{p}^\eta(\bm{b}, \bm{c}))A\text{diag}(\bm{p}^\eta(\bm{b}, \bm{c}))\bm{c}\\
    \nabla_{\bm{c}} \theta^\eta_L(\bm{b}, \bm{c}) & = \bm{p}^\eta(\bm{b}, \bm{c}) +  \left(\eta \ \text{diag}(\bm{p}^\eta(\bm{b}, \bm{c})) A'Q(\bm{p}^\eta(\bm{b},\bm{c}))A \text{diag}(\bm{p}^\eta(\bm{b}, \bm{c})) - \text{diag}(\bm{p}^\eta(\bm{b}, \bm{c}))\right) \bm{c}.
  \end{align*}
\end{lemma}

\begin{proof}[Proof of Lemma~\ref{lem:jacobian}]

  Letting $\bm{\lambda}$ denote Lagrangian dual variables/multipliers, the Lagrangian of the entropic regularized program \eqref{eq:condl_entropic_primal} is
  \[
    \mathcal{L}(\bm{p}, \bm{\lambda}) = \<\bm{c}, \bm{p}\> + \frac{1}{\eta} \sum_{k=1}^K p_k (\log(p_k) - 1) + \<\bm{\lambda}, A\bm{p} - \bm{b}\>.
  \]
  The optimal values $\bm{p}^\eta(\bm{b}, \bm{c})$ and $\bm{\lambda}^\eta(\bm{b}, \bm{c})$ satisfy the first-order conditions
  \begin{align*}
    \nabla_{\bm{p}} \mathcal{L}(\bm{p}, \bm{\lambda}) & = \bm{c} + \frac{1}{\eta} \log(\bm{p}) + A'\bm{\lambda} = 0\\
    \nabla_{\bm{\lambda}} \mathcal{L}(\bm{p}, \bm{\lambda}) & = A\bm{p} - \bm{b} = 0.
  \end{align*}
  Taking the Jacobians of the first-order conditions with respect to $\bm{b}$, we have that
  \begin{align*}
    \frac{1}{\eta} \diag(\bm{p}^\eta(\bm{b}, \bm{c}))^{-1}\nabla_{\bm{b}} \bm{p}^\eta(\bm{b}, \bm{c}) + A' \nabla_{\bm{b}} \bm{\lambda}^\eta(\bm{b}, \bm{c}) & = 0\\
    A \nabla_{\bm{b}} \bm{p}^\eta(\bm{b}, \bm{c}) & = I_J.
  \end{align*}
  Solving for $\nabla_{\bm{b}} \bm{p}^\eta(\bm{b}, \bm{c})$ in terms of $\nabla_{\bm{b}} \bm{\lambda}^\eta(\bm{b}, \bm{c})$ gives
  \[
    \nabla_{\bm{b}} \bm{p}^\eta(\bm{b}, \bm{c}) = -\eta \diag(\bm{p}^\eta(\bm{b}, \bm{c})) A' \nabla_{\bm{b}} \bm{\lambda}^\eta(\bm{b}, \bm{c}).
  \]
  Plugging this in to the second equation and solving for $\nabla_{\bm{b}} \bm{\lambda}^\eta(\bm{b}, \bm{c})$ gives
  \[
    \nabla_{\bm{b}} \bm{\lambda}^\eta(\bm{b}, \bm{c}) = -\frac{1}{\eta} \left(A \diag(\bm{p}^\eta(\bm{b}, \bm{c})) A'\right)^{-1} = -\frac{1}{\eta}Q(\bm{p}^\eta(\bm{b}, \bm{c})).
  \]
  Finally, this gives that the Jacobian of $\bm{p}^\eta(\bm{b}, \bm{c})$ with respect to $\bm{b}$ is
  \[
    \nabla_{\bm{b}} \bm{p}^\eta(\bm{b}, \bm{c}) = \text{diag}(\bm{p}^\eta(\bm{b}, \bm{c})) A'Q(\bm{p}^\eta(\bm{b}, \bm{c})).
  \]

  Now, taking the Jacobians of the first-order conditions with respect to $\bm{c}$, we have that
  \begin{align*}
    \frac{1}{\eta} \diag(\bm{p}^\eta(\bm{b}, \bm{c}))^{-1}\nabla_{\bm{c}} \bm{p}^\eta(\bm{b}, \bm{c}) + A' \nabla_{\bm{c}} \bm{\lambda}^\eta(\bm{b}, \bm{c}) & = -I_K\\
    A \nabla_{\bm{c}} \bm{p}^\eta(\bm{b}, \bm{c}) & = 0.
  \end{align*}
  Solving for $\nabla_{\bm{c}} \bm{p}^\eta(\bm{b}, \bm{c})$ in terms of $\nabla_{\bm{c}} \bm{\lambda}^\eta(\bm{b}, \bm{c})$ gives
  \[
    \nabla_{\bm{c}} \bm{p}^\eta(\bm{b}, \bm{c}) = -\eta \diag(\bm{p}^\eta(\bm{b}, \bm{c}))\left( A' \nabla_{\bm{c}} \bm{\lambda}^\eta(\bm{b}, \bm{c}) + I_K\right).
  \]
  Plugging this in to the second equation and solving for $\nabla_{\bm{c}} \bm{\lambda}^\eta(\bm{b}, \bm{c})$ gives
  \[
    \nabla_{\bm{c}} \bm{\lambda}^\eta(\bm{b}, \bm{c}) = - \left(A \diag(\bm{p}^\eta(\bm{b}, \bm{c})) A'\right)^{-1} A \diag(\bm{p}^\eta(\bm{b}, \bm{c})).
  \]
  Finally, this gives that the Jacobian of $\bm{p}^\eta(\bm{b}, \bm{c})$ with respect to $\bm{c}$ is
  \begin{align*}
     \nabla_{\bm{c}} \bm{p}^\eta(\bm{b}, \bm{c}) &= -\eta \diag(\bm{p}^\eta(\bm{b}, \bm{c}))\left(- A' \left(A \diag(\bm{p}^\eta(\bm{b}, \bm{c})) A'\right)^{-1} A \diag(\bm{p}^\eta(\bm{b}, \bm{c}))  + I_K\right).
  \end{align*}
  Computing the Jacobians of $\theta^\eta_L(\bm{b}, \bm{c})$ follows from the chain rule.

\end{proof}

\begin{lemma}
  \label{lem:jac_hes_bounds}
  Define $\theta_L^\eta(\bm{b}, \bm{c}) = \<\bm{c}, \bm{p}^\eta(\bm{b}, \bm{c})\>$.
  For any vector $\bm{b} \in \{A\bm{p} \mid \bm{p} \in \R_{++}^K\}$ (where $\R_{++}^K$ is the strictly positive orthant of $\R^K$) and any vectors $\bm{c} \in \R^K$ and $\bm{v} \in \R^K$:
  \begin{align*}
    \|\nabla_{\bm{b}}\bm{p}^\eta(\bm{b},\bm{c})'\bm{v} \|_2 & \leq \max_{B \in \mathcal{B}} \|A_B^{-1\prime}\|_2\|\bm{v}\|_2\\
    \|\nabla_{\bm{c}}\bm{p}^\eta(\bm{b}, \bm{c})'\bm{c}\|_2 & \leq \sqrt{\|\bm{p}^\eta(\bm{b}, \bm{c})\|_\infty} \sqrt{\sum_{k=1}^K p^\eta_k(\bm{b}, \bm{c})(\log p^\eta_k(\bm{b}, \bm{c}))^2}\\
    \left\|\frac{\partial}{\partial b_j} \nabla_{\bm{b}}\theta^\eta_L(\bm{b}, \bm{c}) \right\|_2 & \leq 2 J \frac{\|A\|_2 \max_{B \in \mathcal{B}}\|A_B^{-1}\|_2 \|\bm{c}\|_2}{\sigma_\text{min}(A \text{diag}(\bm{p}^\eta(\bm{b}, \bm{c}))A')} \\
    \left\|\frac{\partial}{\partial c_k} \nabla_{\bm{b}}\theta^\eta_L(\bm{b}, \bm{c})\right\|_2 & \leq 2J\eta\left( \left(\|A\|_2\| \max_{B \in \mathcal{B}}\|A_B^{-1\prime}\|_2 + 1\right)  \|\bm{c}\|_2 + 1\right) \max_{B \in \mathcal{B}}\|A_B^{-1}\|_2\\
    \left\|\pderiv{}{c_k}\nabla_{\bm{c}}\theta^\eta(\bm{b}, \bm{c})\right\| & \leq \eta \sqrt{K}\|\bm{p}^\eta(\bm{b}, \bm{c})\|_\infty  \left( 2 + \sqrt{\frac{\sum_{k=1}^K p^\eta_k(\bm{b}, \bm{c})(\log p^\eta_k(\bm{b}, \bm{c}))^2 }{\|\bm{p}^\eta(\bm{b}, \bm{c})\|_\infty} }\right) \times \left(\|A\|_2\| \max_{B \in \mathcal{B}}\|A_B^{-1\prime}\|_2 + 1\right)
  \end{align*}
\end{lemma}

\begin{proof}[Proof of Lemma~\ref{lem:jac_hes_bounds}]

  Note that  from Lemma~\ref{lem:jacobian}
  \[
     \nabla_{\bm{b}}\bm{p}^\eta(\bm{b},\bm{c})'\bm{v} = (A \text{diag}(\bm{p}^\eta(\bm{b}, \bm{c})) A')^{-1}A\text{diag}(\bm{p}^\eta(\bm{b}, \bm{c}))\bm{v},
  \]
  is the weighted least squares solution to
  \[
    \min_{\bm{u} \in \R^J} \|\text{diag}(\bm{p}^\eta(\bm{b}, \bm{c}))^{1/2}(\bm{v} - A'\bm{u})\|_2^2.
  \]
  From \citet{ben-tal_geometric_1990}, Corollary 2.1, the solution to a weighted least squares problem can be written as a convex combination of the basic solutions $A_B^{-1\prime}\bm{c}_B$ for bases $B \in \mathcal{B}$:

  \[
    \nabla_{\bm{b}}\bm{p}^\eta(\bm{b}, \bm{c})'\bm{v} = \frac{\sum_{B \in \mathcal{B}} w(\bm{b}, \bm{c}; B) A_B^{-1\prime}\bm{v}_B}{\sum_{B' \in \mathcal{B}} w(\bm{b}, \bm{c}; B)}
  \]
  where the weights
  \[
    w(\bm{b}, \bm{c}; B) = \det(A_B)^2 \prod_{i \in B} p^\eta_i(\bm{b}, \bm{c}),
  \]
  are non-negative and scale with the determinant of $A_B$ (potential bases $B$ where $A_B$ is not full rank get zero weight and so we have excluded them here and only ranged over the set of bases where $A_B$ is full rank, $\mathcal{B}$).
  Therefore, we have the bound
  \[
    \|\nabla_{\bm{b}}\bm{p}^\eta(\bm{b},\bm{c})'\bm{v} \|_2 \leq \max_{B \in \mathcal{B}} \|A_B^{-1\prime}\bm{v}_B\|_2 \leq  \max_{B \in \mathcal{B}} \|A_B^{-1\prime}\|_2\|\bm{v}\|_2.
  \]

  Next, note that
  \[
    \nabla_{\bm{c}}\theta^\eta_L(\bm{b}, \bm{c}) = \bm{p}^\eta(\bm{b}, \bm{c}) + \nabla_{\bm{c}}\bm{p}^\eta(\bm{b}, \bm{c})'\bm{c}.
  \]
  Defining $P = \text{diag}(\bm{p}^\eta(\bm{b}, \bm{c}))^{1/2} A' (A \text{diag}(\bm{p}^\eta(\bm{b}, \bm{c})) A')^{-1}A\text{diag}(\bm{p}^\eta(\bm{b}, \bm{c}))^{1/2}$ as the projection matrix for the weighted least squares problem above, from Lemma~\ref{lem:jacobian}
  \[
\nabla_{\bm{c}}\bm{p}^\eta(\bm{b}, \bm{c}) = \eta \times  \text{diag}(\bm{p}^\eta(\bm{b}, \bm{c}))^{1/2} (P - I_K) \text{diag}(\bm{p}^\eta(\bm{b}, \bm{c}))^{1/2}.
\]
Now, from the zero gradient condition for the Lagrangian of the entropic problem computed in Lemma~\ref{lem:jacobian}, we have that
\[
  \frac{1}{\eta}\log(\bm{p}^\eta(\bm{b}, \bm{c})) + A'\bm{\lambda} = -\bm{c},
\]
where $\bm{\lambda}$ are the Lagrange multipliers. Since $P$ is a weighted projection matrix on to row space of $A$,
\[
    \text{diag}(\bm{p}^\eta(\bm{b}, \bm{c}))^{1/2} (P - I_K) \text{diag}(\bm{p}^\eta(\bm{b}, \bm{c}))^{1/2} A'\bm{\lambda} = 0,
\]
giving that
\[
  \nabla_{\bm{c}}\bm{p}^\eta(\bm{b}, \bm{c})'\bm{c} = \text{diag}(\bm{p}^\eta(\bm{b}, \bm{c}))^{1/2} (I_K - P) \text{diag}(\bm{p}^\eta(\bm{b}, \bm{c}))^{1/2} \log \bm{p}^\eta(\bm{b}, \bm{c}).
\]
This gives the bound
\begin{align*}
  \|\nabla_{\bm{c}}\bm{p}^\eta(\bm{b}, \bm{c})'\bm{c}\|_2 & \leq \|\text{diag}(\bm{p}^\eta(\bm{b}, \bm{c}))^{1/2} \|_2 \|P - I\|_2 \|\text{diag}(\bm{p}^\eta(\bm{b}, \bm{c}))^{1/2} \log \bm{p}^\eta(\bm{b}, \bm{c})\|_2\\
  & \leq \sqrt{\|\bm{p}^\eta(\bm{b}, \bm{c})\|_\infty} \sqrt{\sum_{k=1}^K p^\eta_k(\bm{b}, \bm{c})(\log p^\eta_k(\bm{b}, \bm{c}))^2 }.
\end{align*}

Next, we consider the second derivatives. First, differentiating the weighted form of  $\nabla_{\bm{b}}\theta^\eta_L(\bm{b}, \bm{c})$ with respect to $b_j$ gives
\begin{align*}
  \frac{\partial}{\partial b_j} \nabla_{\bm{b}}\theta^\eta_L(\bm{b}, \bm{c}) & = \sum_{B \in \mathcal{B}} \frac{\partial}{\partial b_j}\left[\frac{w(\bm{b}, \bm{c}; B)}{\sum_{B' \in \mathcal{B}} w(\bm{b}, \bm{c}; B')}\right] A_B^{-1\prime}\bm{c}_B.
\end{align*}
Letting $Z(\bm{b}, \bm{c}) = \sum_{B' \in \mathcal{B}} w(\bm{b}, \bm{c}; B')$ be the sum of the weights, note that
\begin{align*}
  \frac{\partial}{\partial b_j}\left[\frac{w(\bm{b}, \bm{c}; B)}{Z(\bm{b}, \bm{c})}\right] & = \frac{w(\bm{b}, \bm{c}; B)}{Z(\bm{b}, \bm{c})}\frac{\partial}{\partial b_j}\left[\log w(\bm{b}, \bm{c}; B) - \log Z(\bm{b}, \bm{c})\right]\\
  & =   \frac{w(\bm{b}, \bm{c}; B)}{Z(\bm{b}, \bm{c})} \left(\sum_{k \in B}\pderiv{}{b_j}\log p^\eta_k(\bm{b}, \bm{c}) - \frac{1}{Z}\sum_{B'}w(\bm{b}, \bm{c}; B')\sum_{k' \in B'}\pderiv{}{b_j}\log p^\eta_{k'}(\bm{b}, \bm{c}) \right).
\end{align*}
Because the weights are non-negative, this gives that
\[
  \left|\frac{\partial}{\partial b_j}\left[\frac{w(\bm{b}, \bm{c}; B)}{Z(\bm{b}, \bm{c})}\right] \right| \leq 2 J  \frac{w(\bm{b}, \bm{c}; B)}{Z(\bm{b}, \bm{c})} \left\|\pderiv{}{b_j}\log \bm{p}^\eta(\bm{b}, \bm{c})\right\|_\infty.
\]
Putting together the pieces, this implies that a bound on the second derivative is
\[
  \left\|\frac{\partial}{\partial b_j} \nabla_{\bm{b}}\theta^\eta_L(\bm{b}, \bm{c}) \right\|_2 \leq 2 J \left\|\pderiv{}{b_j}\log \bm{p}^\eta(\bm{b}, \bm{c})\right\|_\infty \max_{B \in \mathcal{B}}\|A_B^{-1}\|_2 \|\bm{c}\|_2.
\]

Next, differentiating $\nabla_{\bm{b}}\theta^\eta_L(\bm{b}, \bm{c})$ with respect to $c_k$:
\begin{align*}
  \frac{\partial}{\partial c_k} \nabla_{\bm{b}}\theta^\eta_L(\bm{b}, \bm{c}) & = \sum_{B \in \mathcal{B}} \frac{\partial}{\partial c_k}\left[\frac{w(\bm{b}, \bm{c}; B)}{\sum_{B' \in \mathcal{B}} w(\bm{b}, \bm{c}; B')}\right] A_B^{-1\prime}\bm{c}_B\\
  & \quad + \sum_{B \in \mathcal{B}} \left[\frac{w(\bm{b}, \bm{c}; B)}{\sum_{B' \in \mathcal{B}} w(\bm{b}, \bm{c}; B')}\right] A_B^{-1\prime}[\bm{e}_k]_B.
\end{align*}

Applying the same argument for the bounds on $\nabla_{\bm{b}}\theta^\eta_L(\bm{b}, \bm{c})$ and $ \frac{\partial}{\partial b_j}\left[\frac{w(\bm{b}, \bm{c}; B)}{Z(\bm{b}, \bm{c})}\right] $ above, this gives that
\[
  \left\|\frac{\partial}{\partial c_k} \nabla_{\bm{b}}\theta^\eta_L(\bm{b}, \bm{c})\right\|_2 \leq 2J\left(\left\|\pderiv{}{c_k}\log \bm{p}^\eta(\bm{b}, \bm{c})\right\|_\infty  \|\bm{c}\|_2 + 1\right) \max_{B \in \mathcal{B}}\|A_B^{-1}\|_2
\]

To control the Hessian with respect to the vector $\bm{c}$, note that
\begin{align*}
  \pderiv{}{c_k}\nabla_{\bm{c}}\theta_L^\eta(\bm{b}, \bm{c}) & = \pderiv{}{c_k} \bm{p}^\eta(\bm{b}, \bm{c}) \tag{(i)}\\
  & + \pderiv{}{c_k}\left[\text{diag}(\bm{p}^\eta(\bm{b}, \bm{c}))^{1/2} (I_K - P ) \text{diag}(\bm{p}^\eta(\bm{b}, \bm{c}))^{1/2}\right] \log \bm{p}^\eta(\bm{b}, \bm{c}) \tag{ii}\\
  & \quad + \text{diag}(\bm{p}^\eta(\bm{b}, \bm{c}))^{1/2} (I_K - P) \text{diag}(\bm{p}^\eta(\bm{b}, \bm{c}))^{1/2}\pderiv{}{c_k}\left[\log \bm{p}^\eta(\bm{b}, \bm{c})\right]. \tag{iii}
\end{align*}

For term (ii), the derivative is
\begin{align*}
  & \pderiv{}{c_k}\left[\text{diag}(\bm{p}^\eta(\bm{b}, \bm{c}))^{1/2} (I_K - P) \text{diag}(\bm{p}^\eta(\bm{b}, \bm{c}))^{1/2}\right]\\
   = & \text{diag}(\bm{p}^\eta(\bm{b}, \bm{c}))^{1/2} (P - I_k)\text{diag}\left(\pderiv{}{c_k} \log \bm{p}^\eta(\bm{b}, \bm{c})\right) (P - I_k)  \text{diag}(\bm{p}^\eta(\bm{b}, \bm{c}))^{1/2}
\end{align*}

Since $\|P - I_K\|_2 \leq 1$, bounding each term in the derivative gives
\begin{align*}
   (i) & \leq \sqrt{K}\|\bm{p}^\eta(\bm{b}, \bm{c})\|_\infty\left\|\pderiv{}{c_k}\log \bm{p}^\eta(\bm{b}, \bm{c})\right\|_\infty\\
   (ii) & \leq  \sqrt{K} \sqrt{\|\bm{p}^\eta(\bm{b}, \bm{c})\|_\infty} \sqrt{\sum_{k=1}^K p^\eta_k(\bm{b}, \bm{c})(\log p^\eta_k(\bm{b}, \bm{c}))^2 }  \left\|\pderiv{}{c_k}\log \bm{p}^\eta(\bm{b}, \bm{c})\right\|_\infty\\
   (iii) & \leq \sqrt{K}\|\bm{p}^\eta(\bm{b}, \bm{c})\|_\infty \left\|\pderiv{}{c_k}\log \bm{p}^\eta(\bm{b}, \bm{c})\right\|_\infty.
\end{align*}

Now it remains to bound $\pderiv{}{b_j}\log \bm{p}^\eta(\bm{b}, \bm{c})$ and $\pderiv{}{c_k}\log \bm{p}^\eta(\bm{b}, \bm{c})$. Using the Jacobians computed in Lemma~\ref{lem:jacobian}, we have that
\begin{align*}
  \left\|\pderiv{}{b_j}\log \bm{p}^\eta(\bm{b}, \bm{c})\right\|_\infty & \leq \left\|A'Q(\bm{p}^\eta(\bm{b}, \bm{c}))\right\|_2 \leq \frac{\|A\|_2}{\sigma_\text{min}(A \text{diag}(\bm{p}^\eta(\bm{b}, \bm{c}))A')}\\
  \left\|\pderiv{}{c_k}\log \bm{p}^\eta(\bm{b}, \bm{c})\right\|_\infty & \leq \eta \left\|A' \nabla_{\bm{b}}\bm{p}^\eta(\bm{b}, \bm{c})' - I\right\|_\infty \leq \eta \left(\|A\|_2\| \nabla_{\bm{b}}\bm{p}^\eta(\bm{b}, \bm{c}) \|_2 + 1\right)\\
  & \leq  \eta \left(\|A\|_2\| \max_{B \in \mathcal{B}}\|A_B^{-1\prime}\|_2 + 1\right)
\end{align*}
Putting together the pieces gives the stated bounds.

\end{proof}

\begin{lemma}
  \label{lem:entropic_est_bias_fixed}
  For a fixed $\eta > 0$, under Assumptions~\ref{a:riesz}, \ref{a:rates} and \ref{a:strict_feasibility}, the bias for $\hat{\theta}^\eta_L$ relative to $\theta^\eta_L$ is bounded by
    \begin{align*}
      \left|\E[\hat{\theta}_L^\eta] - \theta_L^\eta\right| & \lesssim r_n + \eta \|\hat{\bm{b}} - \bm{b}\|_2\|\hat{\bm{c}} - \bm{c}\|_2 + \eta \|\hat{\bm{c}} - \bm{c}\|_2^2 + \|\hat{\bm{b}} - \bm{b}\|_2^{2}
    \end{align*}
\end{lemma}

\begin{proof}[Proof of Lemma~\ref{lem:entropic_est_bias_fixed}]

  Defining $\theta_L^\eta(\bm{b}, \bm{c}) = \<\bm{c}, \bm{p}^\eta(\bm{b}, \bm{c})\>$, we can write
  \begin{align*}
    \E[\hat{\theta}^\eta_L] - \theta^\eta_L %
    & = \E\left[\left\<\nabla_{\bm{b}}\theta_L^\eta(\hat{\bm{b}}(X), \hat{\bm{c}}(X)), \hat{\bm{b}}(X) + \hat{\bm{\varphi}}^{(b)}(O) - \bm{b}(X)\right\>\right]\\
    & + \quad \E\left[\left\<\nabla_{\bm{c}}\theta_L^\eta(\hat{\bm{b}}(X), \hat{\bm{c}}(X)), \hat{\bm{c}}(X) + \hat{\bm{\varphi}}^{(c)}(O) - \bm{c}(X)\right\>\right]\\
    & + \quad \E\left[\theta_L^\eta(\hat{\bm{b}}(X), \hat{\bm{c}}(X)) - \theta_L^\eta(\bm{b}(X), \bm{c}(X)) + \left\<\nabla_{\bm{b}}\theta_L^\eta(\hat{\bm{b}}(X), \hat{\bm{c}}(X)), \bm{b}(X) -\hat{\bm{b}}(X) \right\>\right]\\
    & + \quad  \E\left[\left\<\nabla_{\bm{c}}\theta_L^\eta(\hat{\bm{b}}(X), \hat{\bm{c}}(X)), \bm{c}(X) -\hat{\bm{c}}(X) \right\>\right].
  \end{align*}

  From Lemma~\ref{lem:jac_hes_bounds} $\nabla_{\bm{b}}\theta_L^\eta(\hat{\bm{b}}(X), \hat{\bm{c}}(X))$ and $\nabla_{\bm{c}}\theta_L^\eta(\hat{\bm{b}}(X), \hat{\bm{c}}(X))$ are bounded and so from Assumption~\ref{a:rates}, the magnitudes of the first two terms are $\lesssim r_n$.

  Turning to the remainder term,
  \begin{align*}
    R & \equiv \E[\hat{\theta}^\eta_L] - \theta^\eta_L\\
    & - \E\left[\left\<\nabla_{\bm{b}}\theta_L^\eta(\hat{\bm{b}}(X), \hat{\bm{c}}(X)), \hat{\bm{b}}(X) + \hat{\bm{\varphi}}^{(b)}(O) - \bm{b}(X)\right\>\right]+ \E\left[\left\<\nabla_{\bm{c}}\theta_L^\eta(\hat{\bm{b}}(X), \hat{\bm{c}}(X)), \hat{\bm{c}}(X) + \hat{\bm{\varphi}}^{(c)}(O) - \bm{c}(X)\right\>\right],\\
  \end{align*}
  we can decompose it as
  \begin{align*}
    R & = \E\left[\theta_L^\eta(\hat{\bm{b}}(X), \hat{\bm{c}}(X)) - \theta_L^\eta(\hat{\bm{b}}(X), \bm{c}(X)) - \left\<\nabla_{\bm{c}}\theta_L^\eta(\hat{\bm{b}}(X), \hat{\bm{c}}(X)), \hat{\bm{c}}(X) - \bm{c}(X)\right\>\right] \tag{i}\\
    & +  \E\left[\left\<\nabla_{\bm{b}}\theta_L^\eta(\hat{\bm{b}}(X), \bm{c}(X)) - \nabla_{\bm{b}}\theta_L^\eta(\hat{\bm{b}}(X), \hat{\bm{c}}(X)), \hat{\bm{b}}(X) - \bm{b}(X)\right\>\right] \tag{ii}\\
    & + \E\left[\theta_L^\eta(\hat{\bm{b}}(X), \bm{c}(X)) - \theta_L^\eta(\bm{b}(X), \bm{c}(X)) - \left\<\nabla_{\bm{b}}\theta_L^\eta(\hat{\bm{b}}(X), \bm{c}(X)), \hat{\bm{b}}(X) - \bm{b}(X)\right\>\right]. \tag{iii}
  \end{align*}

  From a second-order Taylor expansion around $\hat{\bm{c}}(X)$, term (i) is
  \[
    (i) = \E\left[\frac{1}{2}\left\<\hat{\bm{c}}(X) - \bm{c}(X), \nabla^2_{\bm{c}}\theta_L^\eta(\hat{\bm{b}}(X), \tilde{\bm{c}}(X))\hat{\bm{c}}(X) - \bm{c}(X)\right\>\right],
  \]

  where $\tilde{c}(x) = t \hat{\bm{c}}(x) + (1 - t) \bm{c}(x)$ for some $0\leq t \leq 1$. From Lemma~\ref{lem:jac_hes_bounds}, since the function $x(\log x)^2$ is bounded for bounded $x$, term (i) is bounded by:
  \[
    (i) \lesssim \eta \E\left[\|\hat{\bm{c}}(X) - \bm{c}(X)\|_2^2\right].
  \]
  For term (ii), by the mean value theorem, for some potentially different $\tilde{c}(x) = t \hat{\bm{c}}(x) + (1 - t) \bm{c}(x)$, we have that
  \[
    (ii) = \E\left[\sum_{k=1}^K (\hat{c}_k(X) - c_k(X)) \left\<\pderiv{}{c_k}\nabla_{\bm{b}}\theta_L^\eta(\hat{\bm{b}}(X), \tilde{\bm{c}}(X)), \hat{\bm{b}}(X) - \bm{b}(X)\right\>\right].
  \]
  So from Lemma~\ref{lem:jac_hes_bounds}, term (ii) is bounded by
  \[
    (ii) \lesssim \eta \E\left[\|\hat{\bm{b}}(X) - \bm{b}(X)\|_2 \|\hat{\bm{c}}(X) - \bm{c}(X)\|_2\right].
  \]

  Finally, for term (iii), from a second-order Taylor expansion around $\bm{b}(X)$, we have that the term inside of the expectation is
  \[
    \frac{1}{2}\left\<\hat{\bm{b}}(X) - \bm{b}(X), \nabla^2_{\bm{b}}\theta_L^\eta(\tilde{\bm{b}}(X), \bm{c}(X))\hat{\bm{b}}(X) - \bm{b}(X)\right\>,
  \]
  for some $\tilde{\bm{b}}(X) = t \hat{\bm{b}}(X) + (1 - t)\bm{b}(X)$ for some $0 \leq t \leq 1$. Note that the Hessian exists because both $\hat{\bm{b}}(X)$ and $\bm{b}(X)$ lead to constraints sets with at least one strictly feasible solution, so a convex combination of them does as well.

  Define $\psi^\eta(\bm{b}, \bm{c}) = \sigma_\text{min}(A\diag(\bm{p}^\eta(\bm{b}, \bm{c}))A')$. From the Hessian bound in Lemma~\ref{lem:jac_hes_bounds}, this gives a bound
  \[
    (iii) \leq \E\left[\left\|\hat{\bm{b}}(X) - \bm{b}(X)\right\|_2^2\frac{2J^2 \|A\|_2\max_{B \in \mathcal{B}}\|A_B^{-1}\|_2\|\bm{c}\|_2}{\psi^\eta(\tilde{\bm{b}}(X), \bm{c}(X))}\right].
  \]
  Note that for $\tilde{\bm{p}}(X) = t\hat{\bm{p}}(X) + (1-t)\bm{p}(X)$, $A\tilde{\bm{p}}(X) = \tilde{\bm{b}}(X)$ and $\tilde{\bm{p}}_k(X) \geq \delta$ for all $k=1,\ldots,K$, so $\tilde{\bm{b}}(X)$ leads to a strictly feasible set.
  Next note that from the first order condition of the entropic regularized objective,
  \[
    \frac{1}{\eta}\log \bm{p}^\eta(\tilde{\bm{b}}, \bm{c}) +\bm{c} +  A'\bm{\lambda}(\tilde{\bm{b}}, \bm{c})= 0 \Leftrightarrow \frac{1}{\eta}\log\bm{p}^\eta(\tilde{\bm{b}}, \bm{c}) = -\bm{c} - A'\bm{\lambda}(\tilde{\bm{b}}, \bm{c}).
  \]
  Next, notice that since $\min_k \tilde{p}_k \geq \delta$ from the first order condition above and the second first order condition $A\tilde{\bm{p}} = \tilde{\bm{b}} = A\bm{p}(\tilde{\bm{b}}, \bm{c})$, we have that
  \begin{align*}
    -\left\< \tilde{\bm{p}},  \frac{1}{\eta}\log\bm{p}^\eta(\tilde{\bm{b}}, \bm{c})\right\> & =  \left\< \tilde{\bm{p}}, \bm{c} + A'\bm{\lambda}(\tilde{\bm{b}}, \bm{c})\right\>\\
     & = \left\< \tilde{\bm{p}}, \bm{c}\right\> + \left\<\bm{p}^\eta(\tilde{\bm{b}}, \bm{c}),  A'\bm{\lambda}(\tilde{\bm{b}}, \bm{c})\right\>\\
    &  = \left\< \tilde{\bm{p}} - \bm{p}(\tilde{\bm{b}}, \bm{c}), \bm{c}\right\> - \left\<\bm{p}^\eta(\tilde{\bm{b}}, \bm{c}), \frac{1}{\eta}\log \bm{p}^\eta(\tilde{\bm{b}}, \bm{c}) \right\>\\
    & \leq 2\|\bm{c}\|_\infty R_1 + \frac{K}{\eta e},
  \end{align*}
  where $R_1 = \sup_{x \in \mathcal{X}} \max_{\bm{p} \in \mathcal{P}(x)} \|\bm{p}\|_1$ is the maximum $L^1$ norm of all feasible primal variables and we have used $-x\log x \leq e^{-1}$.
  This implies that for any particular index $k$, if $p_k^\eta(\tilde{\bm{b}}, \bm{c}) < 1$ and so $\log p_k^\eta(\tilde{\bm{b}}, \bm{c}) < 0$,  we have that
\begin{align*}
  -\delta \frac{1}{\eta}\log p_k^\eta(\tilde{\bm{b}}, \bm{c})  & \leq -\frac{1}{\eta}\tilde{p}_k\log p_k^\eta(\tilde{\bm{b}}, \bm{c})\\
  &  \leq  2\|\bm{c}\|_\infty R_1 + \frac{K}{\eta e} + \frac{1}{\eta}\sum_{k'\neq k} \tilde{p}_{k'}\log p_{k'}^\eta(\tilde{\bm{b}}, \bm{c})\\
  & \leq 2\|\bm{c}\|_\infty R_1 + \frac{K}{\eta e} + \frac{R_1 \log R_\infty}{\eta},
\end{align*}
  where $R_\infty = \max\{1, \sup_{x \in \mathcal{X}} \max_{\bm{p} \in \mathcal{P}(x)} \|\bm{p}\|_\infty\}$. Since this lower bound is less than or equal to 1, it holds for indices where $p^\eta_k(\tilde{\bm{b}}, \bm{c}) > 1$ as well. This implies that
  \[
    p_k^\eta(\tilde{\bm{b}}, \bm{c}) \geq \exp\left(-\frac{1}{\delta}\left(2\eta \|\bm{c}\|_\infty R_1 + \frac{K}{e} + R_1 \log R_\infty\right)\right) > 0
  \]

  This implies that $\psi^\eta(\tilde{\bm{b}}(X), \bm{c}(X)) \geq \sigma_\text{min}(A)^2 \exp\left(-\frac{1}{\delta}\left(2\eta \|\bm{c}\|_\infty R_1 + \frac{K}{e} + R_1 \log R_\infty\right)\right) > 0$ and so
  \[
    | (iii)| \lesssim \E\left[\left\|\hat{\bm{b}}(X) - \bm{b}(X)\right\|_2^2\right].
  \]
  Putting together the three parts completes the proof.
\end{proof}

\begin{lemma}
  \label{lem:entropic_est_bias_growing}
  Under the conditions of Theorem~\ref{thm:entropic_est_approx}, the bias for $\hat{\theta}^\eta_L$ relative to $\theta^\eta_L$ is bounded by
    \begin{align*}
      \left|\E[\hat{\theta}_L^\eta] - \theta_L^\eta\right| & \lesssim r_n + \eta \|\hat{\bm{b}} - \bm{b}\|_2\|\hat{\bm{c}} - \bm{c}\|_2 + \eta \|\hat{\bm{c}} - \bm{c}\|_2^2 + \|\hat{\bm{b}} - \bm{b}\|_\infty^{1+q}.
    \end{align*}
  In addition,
  \[
    \E\left[\|\nabla_{\bm{b}}\theta_L^\eta(\hat{\bm{b}}(X), \bm{c}(X)) - \nabla_{\bm{b}}\theta_L^\eta(\bm{b}(X), \bm{c}(X))\|_2^2\right] \lesssim \|\hat{\bm{b}} - \bm{b}\|_\infty^{q}.
  \]
\end{lemma}

\begin{proof}[Proof of Lemma~\ref{lem:entropic_est_bias_growing}]

  We can follow the proof of Lemma~\ref{lem:entropic_est_bias_fixed} exactly, except for the term (iii)
  \begin{align*}
    (iii) = \E\left[\theta_L^\eta(\hat{\bm{b}}(X), \bm{c}(X)) - \theta_L^\eta(\bm{b}(X), \bm{c}(X)) - \left\<\nabla_{\bm{b}}\theta_L^\eta(\hat{\bm{b}}(X), \bm{c}(X)), \hat{\bm{b}}(X) - \bm{b}(X)\right\>\right].
  \end{align*}

  We can again use a second-order Taylor expansion around $\bm{b}(X)$ and get the bound
  \[
    (iii) \leq \E\left[\left\|\hat{\bm{b}}(X) - \bm{b}(X)\right\|_2^2\frac{2J^2 \|A\|_2\max_{B \in \mathcal{B}}\|A_B^{-1}\|_2\|\bm{c}\|_2}{\psi^\eta(\tilde{\bm{b}}(X), \bm{c}(X))}\right].
  \]
  Next note that minimum singular value is
  \begin{align*}
    \psi^\eta(\tilde{\bm{b}}(X), \bm{c}(X)) & = \psi^\eta(\bm{b}(X) + \tilde{\bm{b}}(X) - \bm{b}(X), \bm{c}(X))\\
    &  \geq \psi^\eta(\bm{b}(X), \bm{c}(X)) +  \lambda_\text{min}(A\diag(\bm{p}^\eta)(\tilde{\bm{b}}(X), \bm{c}(X) - \bm{p}^\eta)(\bm{b}(X), \bm{c}(X))A'),
  \end{align*}
  where $\lambda_\text{min}(\cdot)$ is the smallest eigenvalue of a square matrix.
  We can further lower bound this:
  \begin{align*}
    \lambda_\text{min}(A\diag(\bm{p}^\eta)(\tilde{\bm{b}}(X), \bm{c}(X) - \bm{p}^\eta)(\bm{b}(X), \bm{c}(X))A')& \geq - \left\|A \text{diag}(\bm{p}^\eta(\tilde{\bm{b}}(X), \bm{c}(X)) - \bm{p}^\eta(\bm{b}(X), \bm{c}(X)))A'\right\|_2\\
    & \geq - \|A\|_2^2\|\bm{p}^\eta(\tilde{\bm{b}}(X), \bm{c}(X)) - \bm{p}^\eta(\bm{b}(X), \bm{c}(X))\|_2\\
    & \geq - \|A\|_2^2 \max_{B \in \mathcal{B}}\|A_B^{-1}\|_2\|\tilde{\bm{b}}(X) - \bm{b}(X)\|_2\\
    & \geq - \|A\|_2^2 \max_{B \in \mathcal{B}}\|A_B^{-1}\|_2\|\hat{\bm{b}}(X) - \bm{b}(X)\|_2,
  \end{align*}
  where by  Lemma~\ref{lem:jac_hes_bounds} $\bm{p}^\eta(\bm{b}, \bm{c})$ is Lipschitz in $\bm{b}$ with constant $\max_{B \in \mathcal{B}}\|A_B^{-1}\|_2$. This implies that
  \[
    |(iii)| \leq \E\left[ \left\|\hat{\bm{b}}(X) - \bm{b}(X)\right\|_2^2\frac{2J^2 \|A\|_2\max_{B \in \mathcal{B}}\|A_B^{-1}\|_2\|\bm{c}\|_2}{\psi^\eta(\bm{b}(X), \bm{c}(X)) - \|A\|_2^2 \max_{B \in \mathcal{B}}\|A_B^{-1}\|_2\|\hat{\bm{b}}(X) - \bm{b}(X)\|_2}\right].
  \]
  An alternative first-order bound on term (iii) uses Lemma~\ref{lem:jac_hes_bounds} to note that $\theta^\eta_L(\bm{b}, \bm{c})$ is Lipschitz in $\bm{b}$ with constant $\max_{B \in \mathcal{B}} \|A_B^{-1\prime}\|_2\|\bm{c}\|_2$, giving
  \[
    |(iii)| \lesssim \E\left[\|\hat{\bm{b}}(X) - \bm{b}(X)\|_2 \max_{B \in \mathcal{B}} \|A_B^{-1\prime}\|_2\|\bm{c}\|_2\right].
  \]
  Combining these two bounds for each covariate value $X$ gives
  \[
    |(iii)| \leq 2\E\left[\|\hat{\bm{b}}(X) - \bm{b}(X)\|_2 \min\left\{1, \frac{\|\hat{\bm{b}}(X) - \bm{b}(X)\|_2}{\psi^\eta(\bm{b}(X), \bm{c}(X)) - \|A\|_2^2 \max_{B \in \mathcal{B}}\|A_B^{-1}\|_2\|\hat{\bm{b}}(X) - \bm{b}(X)\|_2}\right\}\right].
  \]
  We now consider two cases. In the first case, $\|\hat{\bm{b}}(X) - \bm{b}(X)\|_2 \leq \frac{1}{2\|A\|_2^2 \max_{B \in \mathcal{B}}\|A_B^{-1}\|_2}\psi^\eta(\bm{b}(X), \bm{c}(X))$. In this case,
  \[
    \psi^\eta(\bm{b}(X), \bm{c}(X)) - \|A\|_2^2 \max_{B \in \mathcal{B}}\|A_B^{-1}\|_2\|\hat{\bm{b}}(X) - \bm{b}(X)\|_2 \geq \frac{1}{2}  \psi^\eta(\bm{b}(X), \bm{c}(X)).
  \]
  Now, note that $\|\hat{\bm{b}}(X) - \bm{b}(X)\|_2 \leq \sqrt{J}\|\hat{\bm{b}} - \bm{b}\|_\infty$, so denoting $S(X) = \bbone\{\|\hat{\bm{b}} - \bm{b}\|_\infty \leq \frac{1}{2\sqrt{J}\|A\|_2^2 \max_{B \in \mathcal{B}}\|A_B^{-1}\|_2}\psi^\eta(\bm{b}(X), \bm{c}(X))\}$ we can divide into the two cases:
  \begin{align*}
    (iii) & \leq \E\left[S(X)\frac{\|\hat{\bm{b}}(X) - \bm{b}(X)\|_2^2}{\psi^\eta(\bm{b}(X), \bm{c}(X))} + \E\left[\|\hat{\bm{b}}(X) - \bm{b}(X)\|_2 (1-S(X))\right]\right]\\
    & \leq \|\hat{\bm{b}} - \bm{b}\|_\infty^2 \E\left[\psi^\eta(\bm{b}(X), \bm{c}(X))^{-q}\psi^\eta(\bm{b}(X), \bm{c}(X))^{-(1-q)}S(x)\right] + \|\hat{\bm{b}} - \bm{b}\|_\infty P(S(X) = 0)\\
    & \leq  \frac{\|\hat{\bm{b}} - \bm{b}\|_\infty^{1 + q}}{\left(2\sqrt{J}\|A\|_2^2 \max_{B \in \mathcal{B}}\|A_B^{-1}\|_2\right)^{1-q}} M_q\\
    & \quad+  \|\hat{\bm{b}} - \bm{b}\|_\infty  P\left(\psi^\eta(\bm{b}(X), \bm{c}(X)) < \|\hat{\bm{b}} - \bm{b}\|_\infty\frac{2\|A\|_2^2 \max_{B \in \mathcal{B}}\|A_B^{-1}\|_2}{\sqrt{J}}\right).
  \end{align*}

  Note that we can re-write the probability term as
  \begin{align*}
    P\left(\psi^\eta(\bm{b}(X), \bm{c}(X))^{-q} > \|\hat{\bm{b}} - \bm{b}\|_\infty^{-q}\left(\frac{2\|A\|_2^2 \max_{B \in \mathcal{B}}\|A_B^{-1}\|_2}{\sqrt{J}}\right)^{-q}\right)\\ \leq \E[\psi^\eta(\bm{b}(X), \bm{c}(X))^{-q}]\|\hat{\bm{b}} - \bm{b}\|_\infty^{q}\left(\frac{2\|A\|_2^2 \max_{B \in \mathcal{B}}\|A_B^{-1}\|_2}{\sqrt{J}}\right)^{q},
  \end{align*}
  where the right hand side comes from Markov's inequality. Putting together the pieces, we get that
  \[
    |(iii)| \lesssim \|\hat{\bm{b}} - \bm{b}\|_\infty^{1 + q}.
  \]

  For the second bound in the statement of the lemma, we can note that in the case where $S(X) = 1$, we can use the Hessian bound in Lemma~\ref{lem:jac_hes_bounds} to get the bound
  \[
     \|\nabla_{\bm{b}}\theta_L^\eta(\hat{\bm{b}}(X), \bm{c}(X)) - \nabla_{\bm{b}}\theta_L^\eta(\bm{b}(X), \bm{c}(X))\|_2 \leq \frac{C\|\bm{c}(X)\|_2\|\hat{\bm{b}}(X) - \bm{b}(X)\|_2}{\psi(\bm{b}(X), \bm{c}(X))}.
  \]
  Bounding the gradient directly with Lemma~\ref{lem:jac_hes_bounds} also gives
  \[
     \E\left[\|\nabla_{\bm{b}}\theta_L^\eta(\hat{\bm{b}}(X), \bm{c}(X)) - \nabla_{\bm{b}}\theta_L^\eta(\bm{b}(X), \bm{c}(X))\|_2\right] \leq 2 \max_{B \in \mathcal{B}}\|A_B^{-1}\|_2.
  \]
  So we can again separate into $S(X) = 1$ and $S(X) = 0$ to get the bound
  \begin{align*}
    \E\left[\|\nabla_{\bm{b}}\theta_L^\eta(\hat{\bm{b}}(X), \bm{c}(X)) - \nabla_{\bm{b}}\theta_L^\eta(\bm{b}(X), \bm{c}(X))\|_2^2\right] &  \lesssim \|\hat{\bm{b}} - \bm{b}\|_\infty^2 \E\left[\psi(\bm{b}(X), \bm{c}(X))^{-2}S(X)\right] + P(S(X) = 0)\\
    & \lesssim \|\hat{\bm{b}} - \bm{b}\|_\infty^q \E\left[\psi(\bm{b}(X), \bm{c}(X))^{-q}S(X)\right] + \|\hat{\bm{b}} - \bm{b}\|_\infty^q\\
    & \lesssim \|\hat{\bm{b}} - \bm{b}\|_\infty^q .
  \end{align*}
\end{proof}

\subsection{Proofs of results in main text}

\begin{proof}[Proof of Lemma~\ref{lem:direct_est_bias_basis}]

  First, note that we can decompose the bias into

  \begin{align*}
    \E[\hat{\theta}_L] - \theta_L & = \E\left[\left\<\hat{\bm{\lambda}}(X), \hat{\bm{b}}(X) + \hat{\bm{\varphi}}^{(b)}(O)- \bm{b}(X)\right\>\right] \tag{i}\\
    & \quad + \E\left[\left\<\hat{\bm{c}}(X) + \hat{\bm{\varphi}}^{(c)}(O) - \bm{c}(X), \hat{\bm{p}}(X)\right\>\right] \tag{ii}\\
    & \quad + \E\left[\left\<\hat{\bm{c}}(X) - \bm{c}(X), A^{-1}_{\widehat{B}(X)}(\bm{b}(X) - \hat{\bm{b}}(X))\right\>\right] \tag{iii}\\
    & \quad + \E\left[\left\<\bm{c}(X), \left(A^{-1}_{\widehat{B}(X)} - A^{-1}_{B^\ast(X)}\right) \bm{b}(X)\right\>\right].
  \end{align*}
  For term (i), by Assumption~\ref{a:rates}, we have that
    \begin{align*}
      |(i)| & \lesssim \sup_{x \in \mathcal{X}} \|\hat{\bm{\lambda}}(x)\|_2 r_n \lesssim r_n,
    \end{align*}
    because the estimated dual variables $\hat{\bm{\lambda}}(x) = A^{-1 \prime}_{\widehat{B}(x)}\hat{\bm{c}}(x)$ are bounded for all $x$, $A_{B}$ is full rank for any basis $B \in \mathcal{B}$ and there are finitely many bases.
  Similarly, for term (ii), we have that
    \begin{align*}
      |(ii)| & \lesssim \sup_{x \in \mathcal{X}} \|\hat{\bm{p}}(x)\|_2 r_n \lesssim r_n,
    \end{align*}
    because similar to the estimated dual variables, the estimated primal variables $\hat{\bm{p}}(x)$ are bounded for all $x$.
    For term (iii), by Cauchy-Schwarz, we have that
    \[
      |(iii)| \leq  \max_{B \in \mathcal{B}} \|A_B^{-1}\|_2 \E\left[\|\hat{\bm{b}}(X) - \bm{b}(X)\|_2 \|\hat{\bm{c}}(X) - \bm{c}(X)\|_2\right] \lesssim \|\hat{\bm{b}} - \bm{b}\|_2 \|\hat{\bm{c}} - \bm{c}\|_2,
    \]
    because $A_B$ is full rank and there are a finite number of bases, so $\max_{B \in \mathcal{B}} \|A_B^{-1}\|_2$ is bounded

\end{proof}

\begin{proof}[Proof of Theorem~\ref{thm:direct_est_rate}]
  We can decompose the error as
  \[
    \hat{\theta}_L - \theta_L  = \tilde{\theta}_L - \E[\tilde{\theta}_L] + \hat{\theta}_L - \E[\hat{\theta}_L] - (\tilde{\theta}_L - \E[\tilde{\theta}_L]) + \E[\hat{\theta}_L] - \theta_L.
  \]
  We can use Lemma~\ref{lem:direct_est_bias} to bound the bias term $\E[\hat{\theta}_L] - \theta_L$. Moving on to $\hat{\theta}_L - \E[\hat{\theta}_L] - (\tilde{\theta}_L - \E[\tilde{\theta}_L])$,
  define
  \begin{align*}
    \tilde{h}(O) & = \left\<\bm{c}(X) + \bm{\varphi}^{(c)}(O), \hat{\bm{p}}(X)\right\> + \left\<\hat{\bm{\lambda}}(X), \bm{\varphi}^{(b)}(O)\right\>\\
    \hat{h}(O) & = \left\<\hat{\bm{c}}(X) + \hat{\bm{\varphi}}^{(c)}(O), \hat{\bm{p}}(X)\right\> + \left\<\hat{\bm{\lambda}}(X), \hat{\bm{\varphi}}^{(b)}(O)\right\>.
  \end{align*}
  Then by \citet{kennedy_semiparametric_2024}, Lemma 1, we have that
  \[
    (\hat{\theta}_L - \E[\hat{\theta}_L]) - (\tilde{\theta}_L - \E[\tilde{\theta}_L]) = O_p\left(\frac{1}{\sqrt{n}}\|\hat{h} - \tilde{h}\|_2\right),
  \]
  so it suffices to show that $\tilde{h} - \hat{h} = o_p(1)$ in order to show that $\hat{\theta}_L - \E[\hat{\theta}_L] - (\tilde{\theta}_L - \E[\tilde{\theta}_L]) = o_p(n^{-1/2})$. We can decompose the difference as follows:
  \begin{align*}
    \hat{h}(O) - \tilde{h}(O) & = \left\<\hat{\bm{c}}(X) - \bm{c}(X), \hat{\bm{p}}(X)\right\> + \left\<\hat{\bm{\varphi}}^{(c)}(O) - \bm{\varphi}^{(c)}(O), \hat{\bm{p}}(X)\right\> + \left\<\hat{\bm{\lambda}}(X), \hat{\bm{\varphi}}^{(b)}(O) - \bm{\varphi}^{(b)}(O)\right\>.
  \end{align*}
  By Assumption~\ref{a:rates}, we have that $\|\hat{\bm{c}} - \bm{c}\|_\infty = o_p(1)$, $\|\hat{\bm{\varphi}}^{(c)} - \bm{\varphi}^{(c)}\|_\infty = o_p(1)$, and $\|\hat{\bm{\varphi}}^{(b)} - \bm{\varphi}^{(b)}\|_\infty = o_p(1)$. Because the estimated primal and dual variables are bounded, this implies that $\tilde{h} - \hat{h} = o_p(1)$. Combining this with the upper and lower bounds in Lemma~\ref{lem:direct_est_bias} completes the proof.
\end{proof}

\begin{proof}[Proof of Corollary~\ref{cor:direct_est_normal}]
  By Theorem~\ref{thm:direct_est_rate} and Assumption~\ref{a:rate_bfs}\ref{a:rate_bfs_i} , we have that
  \[
    \hat{\theta}_L - \theta_L = \tilde{\theta}_L - \E[\tilde{\theta}_L] + o_p(n^{-1/2}).
  \]

  Now define $\check{\theta}_L  = \frac{1}{n} \sum_{i=1}^n \left\<\bm{c}(X_i) + \bm{\varphi}^{(c)}(O_i), \bm{p}^\ast(X_i)\right\> + \left\<\bm{\lambda}^\ast(X_i), \bm{\varphi}^{(b)}(O_i)\right\>$.
  By Assumption~\ref{a:rate_bfs}\ref{a:rate_bfs_iii} , we can apply the Lindeberg-Feller central limit theorem to $\check{\theta}_L$ to get that
  \[
    \frac{\sqrt{n}\left(\check{\theta}_L - \theta_L\right)}{\sqrt{V_L}} \Rightarrow N(0, 1).
  \]
  It remains to show that $\tilde{\theta}_L - \E[\tilde{\theta}_L] - (\check{\theta}_L - \E[\check{\theta}_L]) = o_p(n^{-1/2})$.
  Define
  \begin{align*}
    \tilde{h}(O) & = \left\<\bm{c}(X) + \bm{\varphi}^{(c)}(O), \hat{\bm{p}}(X)\right\> + \left\<\hat{\bm{\lambda}}(X), \bm{\varphi}^{(b)}(O)\right\>\\
    h(O) & = \left\<\bm{c}(X) + \bm{\varphi}^{(c)}(O), \bm{p}^\ast(X)\right\> + \left\<\bm{\lambda}^\ast(X), \bm{\varphi}^{(b)}(O)\right\>.
  \end{align*}
  As in the proof of Theorem~\ref{thm:direct_est_rate}, it suffices to control
  \begin{align*}
    & \E[(h(O) - \tilde{h}(O))^2]\\
     \leq & \E\left[\left\<\bm{c}(X), \hat{\bm{p}}(X) - \bm{p}^\ast(X)\right\>^2 + 2\left\<\bm{\varphi}^{(c)}(O), \hat{\bm{p}}(X) - \bm{p}^\ast(X)\right\>^2 + 2\left\<\hat{\bm{\lambda}}(X) - \bm{\lambda}^\ast(X), \bm{\varphi}^{(b)}(O)\right\>^2\right].
  \end{align*}

  First, note that
  \begin{align*}
    \left\<\bm{c}(X), \hat{\bm{p}}(X) - \bm{p}^\ast(X)\right\> & = \left\<\bm{c}(X), \left(A^{-1}_{\widehat{B}(X)} - A^{-1}_{B^\ast(X)}\right) \bm{b}(X)\right\> + \left\<\bm{c}(X), A^{-1}_{\widehat{B}(X)} \left(\hat{\bm{b}}(X) - \bm{b}(X)\right)\right\>
  \end{align*}
  From Lemma~\ref{lem:plugin_basis_generic} and Assumption~\ref{a:rate_bfs}\ref{a:rate_bfs_i} , this implies that the this first term is $o_p(1)$.

  For the next two terms, we will first consider the case where  there exists a unique primal solution for almost all $x \in \mathcal{X}$, and $\bm{b}(x)$ is known up to a scalar function.
  Note that,
  \[
    \left\<\bm{\varphi}^{(c)}(O), \hat{\bm{p}}(X) - \bm{p}^\ast(X)\right\> = \left\<\bm{\varphi}^{(c)}(O), \left(A^{-1}_{\widehat{B}(X)} - A^{-1}_{B^\ast(X)}\right)\bm{b}(X)\right\> + \left\<\bm{\varphi}^{(c)}(O), A^{-1}_{\widehat{B}(X)} \left(\hat{\bm{b}}(X) - \bm{b}(X)\right)\right\>.
  \]
  If the estimated basis $\widehat{B}(X)$ is optimal for the true primal problem, and the primal solution is unique, then $A^{-1}_{\widehat{B}(X)}\bm{b}(x) = A^{-1}_{B^\ast(X)}\bm{b}(x)$, and so the first term is zero. Therefore,
\begin{align*}
    \E\left[\left\<\bm{\varphi}^{(c)}(O), \left(A^{-1}_{\widehat{B}(X)} - A^{-1}_{B^\ast(X)}\right)\bm{b}(X)\right\>^2\right] & = \E\left[\bbone\{\widehat{B}(X) \not \in \mathcal{B}^\ast(X)\}\left\<\bm{\varphi}^{(c)}(O), \left(A^{-1}_{\widehat{B}(X)} - A^{-1}_{B^\ast(X)}\right)\bm{b}(X)\right\>^2\right]\\
    & \lesssim P\left(\widehat{B}(X) \not \in \mathcal{B}^\ast(X)\right).
  \end{align*}
  By Lemma~\ref{lem:prob_subopt_converges}, we have that $P\left(\widehat{B}(X) \not \in \mathcal{B}^\ast(X)\right) = o_p(1)$. The second component of the first term is also $o_p(1)$ because $\|\hat{\bm{b}} - \bm{b}\|_\infty = o_p(1)$. If $\bm{b}(x)$ is known up to a scalar function, then $\bm{\varphi}^{(b)}(O) = f(O)\bm{b}(X)$ for some function $f(O)$, and so
  \begin{align*}
    \left\<\hat{\bm{\lambda}}(X) - \bm{\lambda}^\ast(X), \bm{\varphi}^{(b)}(O)\right\> & = f(O)\left\<\hat{\bm{\lambda}}(X) - \bm{\lambda}^\ast(X), \bm{b}(X)\right\>\\
    & = f(O) \left\<\bm{c}(X), \left(A^{-1}_{\widehat{B}(X)} - A^{-1}_{B^\ast(X)}\right) \bm{b}(X)\right\> + f(O)\left\<\hat{\bm{c}}(X) - \bm{c}(X), A^{-1}_{\widehat{B}(X)} \bm{b}(X)\right\>.
  \end{align*}

  Again by Lemma~\ref{lem:plugin_basis_generic} and Assumption~\ref{a:rate_bfs}\ref{a:rate_bfs_i}, this implies that $\left\<\hat{\bm{\lambda}}(X) - \bm{\lambda}^\ast(X), \bm{\varphi}^{(b)}(O)\right\> = o_p(1)$.
  Therefore, $\tilde{h} - h = o_p(1)$.

  Next we consider the case where there exists a non-degenerate primal solution for almost all $x \in \mathcal{X}$, and $\bm{c}(x)$ is known up to a scalar function.
  We can use similar arguments as in the first case.
  First, if $\bm{c}(x)$ is known up to a scalar function, then $\bm{\varphi}^{(c)}(O) = f(O)\bm{c}(X)$ for some function $f(O)$, and we have the decomposition
  \begin{align*}
     \left\<\bm{\varphi}^{(c)}(O), \hat{\bm{p}}(X) - \bm{p}^\ast(X)\right\> & = f(O)\left\<\bm{c}(X), \hat{\bm{p}}(X) - \bm{p}^\ast(X)\right\>\\
    & = f(O) \left\<\bm{c}(X), \left(A^{-1}_{\widehat{B}(X)} - A^{-1}_{B^\ast(X)}\right) \bm{b}(X)\right\> + f(O)\left\<\bm{c}(X), A^{-1}_{\widehat{B}(X)} \left(\hat{\bm{b}}(X) - \bm{b}(X)\right)\right\>.
  \end{align*}
  Again by Lemma~\ref{lem:plugin_basis_generic} and Assumption~\ref{a:rate_bfs}\ref{a:rate_bfs_i}, this implies that $\left\<\bm{\varphi}^{(c)}(O), \hat{\bm{p}}(X) - \bm{p}^\ast(X)\right\> = o_p(1)$.

  Now note that if there exists a non-degenerate primal solution, then the dual solution is unique. Furthermore,
  \[
    \left\<\hat{\bm{\lambda}}(X) - \bm{\lambda}^\ast(X), \bm{\varphi}^{(b)}(O)\right\> = \left\<\bm{c}(X), \left(A^{-1}_{\widehat{B}(X)} - A^{-1}_{B^\ast(X)}\right) \bm{\varphi}^{(b)}(O)\right\> + \left\<\hat{\bm{c}}(X) - \bm{c}(X), A^{-1}_{\widehat{B}(X)} \bm{\varphi}^{(b)}(O)\right\>.
  \]
  Again, if the estimated basis $\widehat{B}(X)$ is optimal, and the dual solution is unique, then $\bm{c}(X)'A^{-1}_{\widehat{B}(X)} = \bm{c}(X)'A^{-1}_{B^\ast(X)}$, and so the first term is zero. Otherwise, we can use the same argument as in the case of a unique primal solution to show that,
  \[
    \E\left[\left\<\bm{c}(X), \left(A^{-1}_{\widehat{B}(X)} - A^{-1}_{B^\ast(X)}\right) \bm{\varphi}^{(b)}(O)\right\>^2\right] \lesssim P\left(\widehat{B}(X) \not \in \mathcal{B}^\ast(X)\right) = o_p(1),
  \]
  by Lemma~\ref{lem:prob_subopt_converges} and use that $\|\hat{\bm{c}} - \bm{c}\|_\infty = o_p(1)$.
  Putting together the pieces gives that $\tilde{h} - h = o_p(1)$ when there exists a non-degenerate primal solution for almost all $x \in \mathcal{X}$, and $\bm{c}(x)$ is known up to a scalar function.

  Finally, if there exists a unique, non-degenerate primal solution for almost all $x \in \mathcal{X}$, then the primal and dual solutions are unique and there is only one unique optimal basis, in which case $A^{-1}_{\widehat{B}(X)} = A^{-1}_{B^\ast(X)}$ if $\widehat{B}(X)$ is optimal. In this case we can use the same arguments as in the previous two cases to show that both
  \[
     \E\left[\left\<\bm{\varphi}^{(c)}(O), \left(A^{-1}_{\widehat{B}(X)} - A^{-1}_{B^\ast(X)}\right)\bm{b}(X)\right\>^2\right] \lesssim P\left(\widehat{B}(X) \not \in \mathcal{B}^\ast(X)\right) = o_p(1),
  \]
  and
  \[
     \E\left[\left\<\bm{c}(X), \left(A^{-1}_{\widehat{B}(X)} - A^{-1}_{B^\ast(X)}\right) \bm{\varphi}^{(b)}(O)\right\>^2\right] \lesssim  P\left(\widehat{B}(X) \not \in \mathcal{B}^\ast(X)\right) = o_p(1),
  \]
  and so $\tilde{h} - h = o_p(1)$.

\end{proof}

\begin{proof}[Proof of Theorem~\ref{thm:entropic_est_rate}]
  First, note that we can decompose the estimation error $\hat{\theta}^\eta_L - \theta^\eta_L$ into
  \begin{align*}
    \hat{\theta}^\eta_L - \theta^\eta_L & = \tilde{\theta}^\eta_L - \E[\tilde{\theta}^\eta_L]\\
    & \quad  + \E[\hat{\theta}^\eta_L] - \theta^\eta_L \tag{$\ast$}\\
    & \quad + (\hat{\theta}^\eta_L - \E[\hat{\theta}^\eta_L]) - (\tilde{\theta}^\eta_L - \E[\tilde{\theta}^\eta_L]) \tag{$\ast\ast$}
  \end{align*}
We have bounded the bias term ($\ast$) in Lemma~\ref{lem:entropic_est_bias_fixed}.   Next, we can control the second part of the decomposition, term ($\ast\ast$) almost exactly as in the proof of Theorem~\ref{thm:entropic_est_approx}, except noting that with $\eta > 0$ fixed, the lower bound $\psi^\eta(\tilde{\bm{b}}(X), \bm{c}(X)) \geq \sigma_\text{min}(A)^2 \exp\left(-\frac{1}{\delta}\left(2\eta \|\bm{c}\|_\infty R_1 + \frac{K}{e} + R_1 \log R_\infty\right)\right) > 0$ gives that term (ii) in the decomposition of $\hat{h}(O) - h(O)$ is
\[
  \E[(ii)^2] \lesssim \eta^2\|\hat{\bm{c}} - \bm{c}\|_2^2 + \|\hat{\bm{b}} - \bm{b}\|_2^2,
\]
due to the gradient with respect to $\bm{b}$ being Lipschitz.

\end{proof}

\begin{proof}[Proof of Theorem~\ref{thm:entropic_est_approx}]

  We compare $\hat{\theta}^\eta_L$ to $\theta_L$. We can decompose the difference as
  \begin{align*}
    \hat{\theta}^\eta_L - \theta_L & = \tilde{\theta}^\eta_L - \E[\tilde{\theta}^\eta_L]\\
    & \quad  + \E[\hat{\theta}^\eta_L] - \theta^\eta_L \tag{$\ast$}\\
    & \quad + (\hat{\theta}^\eta_L - \E[\hat{\theta}^\eta_L]) - (\tilde{\theta}^\eta_L - \E[\tilde{\theta}^\eta_L]) \tag{$\ast\ast$}\\
    & \quad + \E[\tilde{\theta}^\eta_L] - \theta_L \tag{$\ast\ast\ast$}
  \end{align*}

We have bounded the bias term ($\ast$) in Lemma~\ref{lem:entropic_est_bias_growing}.   Next, we control the second part of the decomposition, term ($\ast\ast$), following \citet{kennedy_semiparametric_2024}, Lemma 1. Specifically, we can write
  \begin{align*}
    h(O) & = \theta_L^\eta(\bm{b}(X), \bm{c}(X)) + \left\<\nabla_{\bm{b}}\theta_L^\eta(\bm{b}(X), \bm{c}(X)), \bm{\varphi}^{(b)}(O)\right\> +  \left\<\nabla_{\bm{c}}\theta_L^\eta(\bm{b}(X), \bm{c}(X)), \bm{\varphi}^{(c)}(O)\right\> \\
    \hat{h}(O) & = \theta_L^\eta(\hat{\bm{b}}(X), \hat{\bm{c}}(X)) + \left\<\nabla_{\bm{b}}\theta_L^\eta(\hat{\bm{b}}(X), \hat{\bm{c}}(X)), \hat{\bm{\varphi}}^{(b)}(O)\right\> +  \left\<\nabla_{\bm{c}}\theta_L^\eta(\hat{\bm{b}}(X), \hat{\bm{c}}(X)), \hat{\bm{\varphi}}^{(c)}(O)\right\> .
  \end{align*}
  We can write the difference as
  \begin{align*}
    \hat{h}(O) - h(O) & = \theta_L^\eta(\hat{\bm{b}}(X), \hat{\bm{c}}(X)) - \theta_L^\eta(\bm{b}(X), \bm{c}(X)) \tag{i}\\
    & \quad + \left\<\nabla_{\bm{b}}\theta_L^\eta(\hat{\bm{b}}(X), \hat{\bm{c}}(X)) - \nabla_{\bm{b}}\theta_L^\eta(\bm{b}(X), \bm{c}(X)), \bm{\varphi}^{(b)}(O)\right\> \tag{ii}\\
    & \quad + \left\<\nabla_{\bm{c}}\theta_L^\eta(\hat{\bm{b}}(X), \hat{\bm{c}}(X)) - \nabla_{\bm{c}}\theta_L^\eta(\bm{b}(X), \bm{c}(X)), \bm{\varphi}^{(c)}(O)\right\> \tag{iii}\\
    & \quad + \left\<\nabla_{\bm{b}}\theta_L^\eta(\hat{\bm{b}}(X), \hat{\bm{c}}(X)), \hat{\bm{\varphi}}^{(b)}(O) - \bm{\varphi}^{(b)}(O)\right\>\tag{iv}\\
    & \quad + \left\<\nabla_{\bm{c}}\theta_L^\eta(\hat{\bm{b}}(X), \hat{\bm{c}}(X)), \hat{\bm{\varphi}}^{(c)}(O) - \bm{\varphi}^{(c)}(O)\right\> \tag{v}
  \end{align*}
  With strict feasibility under Assumption~\ref{a:strict_feasibility}, because each component is consistent, i.e. $\|\hat{\bm{b}} - \bm{b}\|_\infty = o_p(1)$ and $\|\hat{\bm{c}} - \bm{c}\|_\infty = o_p(1)$, and both $\sup_{o \in \mathcal{O}}\|\hat{\bm{\varphi}}^{(b)}(x) - \bm{\varphi}^{(b)}(o)\|_\infty = o_p(1)$ and $\sup_{o \in \mathcal{O}}\|\hat{\bm{\varphi}}^{(c)}(x) - \bm{\varphi}^{(c)}(o)\|_\infty = o_p(1)$, from Lemma~\ref{lem:jac_hes_bounds} we have that the gradients $\nabla_{\bm{b}}\theta(\hat{\bm{b}}(X), \hat{\bm{c}}(X))$ and $\nabla_{\bm{c}}\theta(\hat{\bm{b}}(X), \hat{\bm{c}}(X))$ are bounded. So for terms (i), (iv), and (v) $\E[(i)^2] \to 0$, $\E[(iv)^2] \to 0$, and $\E[(v)^2] \to 0$.

  Following the proof of Lemma~\ref{lem:entropic_est_bias_growing}, we can decompose the differences in gradients as
  \begin{align*}
    \nabla_{\bm{b}}\theta_L^\eta(\hat{\bm{b}}(X), \hat{\bm{c}}(X)) - \nabla_{\bm{b}}\theta_L^\eta(\bm{b}(X), \bm{c}(X)) & = \nabla_{\bm{b}}\theta_L^\eta(\hat{\bm{b}}(X), \hat{\bm{c}}(X)) - \nabla_{\bm{b}}\theta_L^\eta(\hat{\bm{b}}(X), \bm{c}(X))\\
    & + \nabla_{\bm{b}}\theta_L^\eta(\hat{\bm{b}}(X), \bm{c}(X)) - \nabla_{\bm{b}}\theta_L^\eta(\bm{b}(X), \bm{c}(X))\\
    \nabla_{\bm{c}}\theta_L^\eta(\hat{\bm{b}}(X), \hat{\bm{c}}(X)) - \nabla_{\bm{c}}\theta_L^\eta(\bm{b}(X), \bm{c}(X)) & = \nabla_{\bm{c}}\theta_L^\eta(\hat{\bm{b}}(X), \hat{\bm{c}}(X)) - \nabla_{\bm{c}}\theta_L^\eta(\hat{\bm{b}}(X), \bm{c}(X))\\
    & + \nabla_{\bm{c}}\theta_L^\eta(\hat{\bm{b}}(X), \bm{c}(X)) - \nabla_{\bm{c}}\theta_L^\eta(\bm{b}(X), \bm{c}(X)).
  \end{align*}
  From Lemma~\ref{lem:jac_hes_bounds}, $\nabla_{\bm{b}}\theta_L^\eta(\hat{\bm{b}}(X), \cdot)$, $\nabla_{\bm{b}}\theta_L^\eta(\cdot, \bm{c}(X))$,  $\nabla_{\bm{c}}\theta_L^\eta(\hat{\bm{b}}(X), \cdot)$, $\nabla_{\bm{c}}\theta_L^\eta(\cdot, \bm{c}(X))$ are all Lipchitz functions of their respective missing arguments. Following the same bounding arguments in the proof of Lemma~\ref{lem:entropic_est_bias_growing} gives that
  \begin{align*}
    \E[(ii)^2] & \lesssim \eta^2 \|\hat{\bm{c}} - \bm{c}\|_2^2 + \|\hat{\bm{b}} - \bm{b}\|_\infty^{q}\\
    \E[(iii)^2] & \lesssim \eta^2 \|\hat{\bm{c}} - \bm{c}\|_2^2 +  \eta^2 \|\hat{\bm{b}} - \bm{b}\|_2^2
  \end{align*}

  Putting together the pieces gives that
  \[
    \|\hat{h} - h\|_2 = O_p\left(\eta \|\hat{\bm{b}} - \bm{b}\|_2 + \eta \|\hat{\bm{c}} - \bm{c}\|_2 + \|\hat{\bm{b}} - \bm{b}\|_\infty^{ \frac{q}{2}}\right) + o_p(1).
  \]
  Under the conditions that $\eta \|\hat{\bm{b}} - \bm{b}\|_2 = o_p(1)$ and $\eta \|\hat{\bm{c}} - \bm{c}\|_2 = o_p(1)$, this implies that $\|\hat{h} - h\|_2 = o_p(1)$ and so term $(\ast\ast) = o_p(n^{-1/2})$ by \citet{kennedy_semiparametric_2024}, Lemma 1.

  Finally we turn to term ($\ast\ast\ast$). We can write it as
  \begin{align*}
    (\ast\ast\ast)& = \E\left[\frac{1}{n}\sum_{i=1}^n \<\bm{c}(X_i) + \bm{\varphi}^{(c)}(O_i), \bm{p}^\eta(X_i)\> + \<\bm{c}(X_i), \nabla_{\bm{b}} \bm{p}^\eta(X_i) \bm{\varphi}^{(b)}(O_i) + \nabla_{\bm{c}} \bm{p}^\eta(X_i) \bm{\varphi}^{(c)}(O_i)\>\right] - \theta_L\\
    & = \E\left[\frac{1}{n}\sum_{i=1}^n \<\bm{c}(X_i), \bm{p}^\eta(X_i)\>\right] - \theta_L\\
    & =  \E\left[\frac{1}{n}\sum_{i=1}^n \<\bm{c}(X_i), \bm{p}^\eta(X_i)\> - \theta_L(X_i)\right].
  \end{align*}

  Define  the size of the conditional feasible set as measured by the $\ell^1$ norm as $R_1(x) \equiv \max_{\bm{p} \in \mathcal{P}(x)} \|\bm{p}\|_1$, and the size as measured by the entropy $R_H(x) \equiv \max_{\bm{p}_1, \bm{p}_2 \in \mathcal{P}(X)} \sum_k p_{2k}(\log p_{2k} - 1) - p_{1k}(\log p_{1k} - 1)$.
  From Theorem 5 in \citet{weed_explicit_2018}, we have that
   if $\eta \geq \frac{R_1(x) + R_H(x)}{\Delta_p(x)}$, then
  \begin{align*}
    0 \leq \<\bm{c}(x), \bm{p}^\eta(x)\>  - \<\bm{c}(x), \bm{p}(x)\> & \leq \Delta_p(x) \exp\left(-\eta \frac{\Delta_p(x)}{R_1(x)} + \frac{R_1(x) + R_H(x)}{R_1(x)}\right),
  \end{align*}
  where we note that if $\Delta_p(x) = 0$, then any valid primal solution is optimal and so $\<\bm{c}(x), \bm{p}^\eta(x)\>  - \<\bm{c}(x), \bm{p}(x)\> = 0$.

    Note that we can decompose this difference into the difference when $\eta$ is larger than $\frac{R_1(x) + R_H(x)}{\Delta_p(x)}$ and when $\eta$ is smaller than this value:
    \begin{align*}
      (\ast\ast\ast) & = \E\left[\< \bm{c}(X), \bm{p}^\eta(X)\> - \theta_L(X)\right]\\
      & = \E\left[\left(\< \bm{c}(X), \bm{p}^\eta(X)\> - \theta_L(X)\right) \bbone\left\{\eta \geq \frac{R_1(X) + R_H(X)}{\Delta_p(X)}\right\}\right] \tag{i}\\
      & + \E\left[\left(\< \bm{c}(X), \bm{p}^\eta(X)\> - \theta_L(X)\right) \bbone\left\{\eta < \frac{R_1(X) + R_H(X)}{\Delta_p(X)}\right\}\right] \tag{ii}
    \end{align*}

  When the margin parameter $\alpha < \infty$, since the hyperparameter is large relative to the margin in the first term (i), from Theorem 5 in \citet{weed_explicit_2018}, we have that
  \[
    (i) \leq \E\left[\Delta_p(X) \exp\left(-\eta \frac{\Delta_p(X)}{R_1(X)} + \frac{R_1(X) + R_H(X)}{R_1(X)}\right)\right] \lesssim \E\left[\Delta_p(X) \exp\left(-\eta \frac{\Delta_p(X)}{R_1}\right)\bbone\{\Delta_p(X) > 0\}\right],
  \]
  where $R_1 = \sup_{x \in \mathcal{X}}R_1(x)$. Because this is an expectation of a non-negative random variable with a maximum value of $\frac{R_1}{\eta e}$
  \[
    (i) \lesssim \int_0^{\frac{R_1}{\eta e}} P\left(\Delta_p(X) \exp\left(-\eta \frac{\Delta_p(X)}{R_1}\right)\bbone\{\Delta_p(X) > 0\} > t\right)dt,
  \]
  where note that because each $\mathcal{P}(x)$ is non-empty $\inf_{x \in \mathcal{X}} R_1(X) > 0$.
  Now note that if $\Delta_p(x) \exp\left(-\eta \frac{\Delta_p(x)}{R_1}\right)\bbone\{\Delta_p(x) > 0\} > t$, this implies that $0 < \Delta_p(x) < f(t)$ where $f(t)$ is the largest solution to $xe^{-\eta \frac{x}{R_1}} = t$. So we can re-write the integral in the upper bound:
  \begin{align*}
    \int_0^{\frac{R_1}{\eta e}} P(0 < \Delta_p(X) \leq f(t)) dt & \lesssim \int_0^{\frac{R_1}{\eta e}}f(t)^\alpha dt\\
    & = \int_{\frac{\eta}{R_1}}^\infty u^\alpha\left(\frac{\eta}{R_1}u - 1\right)e^{-\eta\frac{u}{R_1}}du\\
    & = \left(\frac{R_1}{\eta}\right)^{1 + \alpha}\int_1^\infty v^\alpha (v - 1)e^{-v}dv,
  \end{align*}
  where the first inequality uses the margin condition~in Assumption~\ref{a:primal_margin}, the second equality uses the change of variables $u = f(t)$ (noting that via implicit differentiation $\frac{dt}{du} = \left(1 - \frac{\eta}{R_1}u\right)e^{-\eta\frac{u}{R_1}}$), and the final equality uses the change of variables $v= \frac{\eta}{R_1} u$ (so $\frac{du}{dv} = \frac{R_1}{\eta}$). The final integral is bounded and only depends on $\alpha$, so we have that
  \[
    (i) \lesssim \frac{1}{\eta^{1 + \alpha}}
  \]

  For the second term (ii), the sub-optimality gap is small relative to the inverse of $\eta$, $\Delta_p(X) < \frac{R_1(X) + R_H(X)}{\eta}$.
  Note that when $\Delta_p(X) = 0$, then $\<\bm{c}(X), \bm{p}^\eta(X)\> = \theta_L(X)$, regardless of what $\eta$ is set to, and
  so by Proposition 4 in \citet{weed_explicit_2018}, and the margin condition on $\Delta_p(X)$ (Assumption~\ref{a:primal_margin}), we have that
\begin{align*}
  (ii) & \leq \E\left[\left(\frac{R_H(X)}{\eta}\bbone\left\{0 < \Delta_p(X) < \frac{R_1(X) + R_H(X)}{\eta}\right\}\right)\right]\\
  & \leq \frac{\sup_{x \in \mathcal{X}} R_H(x)}{\eta} P\left(0 < \Delta_p(X) \leq \frac{R_1(X) + R_H(X)}{\eta}\right)\\
  & \lesssim \frac{1}{\eta^{1 + \alpha}} = O\left(\eta^{-(1 + \alpha)}\right).
\end{align*}

 If instead the margin parameter $\alpha = \infty$, then as $\eta \to \infty$ at some point we will have $\eta \geq \frac{\sup_{x \in \mathcal{X}} R_1(x) + R_H(x)}{\Delta_0} \geq \frac{R_1(x) + R_H(x)}{\Delta_p(x)}$ for almost every $x \in \mathcal{X}$ such that $\Delta_p(x) > 0$ (again when $\Delta_p(x) = 0$ then all feasible solutions are optimal and there is no difference between the regularized and unregularized versions).

  This implies that
  \begin{align*}
    0  \leq \E\left[\Delta_p(X) \exp\left(-\eta \frac{\Delta_p(X)}{R_1(X)} + \frac{R_1(X) + R_H(X)}{R_1(X)}\right)\right] = O(e^{-C\eta}),
  \end{align*}
  where $C = \frac{\Delta_0}{\sup_{x \in \mathcal{X}} R_1(x)}$.

  Combining terms ($\ast$), ($\ast\ast$) and ($\ast\ast\ast$) gives the result.

\end{proof}

\begin{proof}[Proof of Corollary~\ref{cor:entropic_est_normal}]
  From Theorem~\ref{thm:entropic_est_approx} and the rate conditions in Assumption~\ref{a:rate_entropic}, we have that $\hat{\theta}^\eta_L - \theta_L= \tilde{\theta}^\eta_L - \E[\tilde{\theta}^\eta_L] + o_p(n^{-1/2})$. The result follows by applying the Lindeberg CLT to $\tilde{\theta}^\eta_L - \E[\tilde{\theta}^\eta_L]$ and Slutsky's theorem.

\end{proof}

\begin{proof}[Proof of Theorem~\ref{thm:primal_policy}]

  Define $\hat{\theta}_L(\pi) \equiv \frac{1}{n}\sum_{i=1}^n \hat{\theta}_L(X_i, \pi(X_i))+ \left\< \hat{\bm{\lambda}}(X_i, \pi(X_i)), \hat{\bm{\varphi}}^{(b)}(O_i)\right\>$.
  First, note that we can write the excess regret as
  \begin{align*}
    \theta_L(\pi^\ast) - \theta_L(\hat{\pi}) &= \theta_L(\pi^\ast) - \hat{\theta}_L(\pi^\ast) + \hat{\theta}_L(\pi^\ast) - \hat{\theta}_L(\hat{\pi}) + \hat{\theta}_L(\hat{\pi}) - \theta_L(\hat{\pi}) \\
    & \leq \theta_L(\pi^\ast) - \hat{\theta}_L(\pi^\ast) + \hat{\theta}_L(\hat{\pi}) - \theta_L(\hat{\pi}) \\
    & \leq 2 \sup_{\pi \in \Pi} \left|\hat{\theta}_L(\pi) - \theta_L(\pi)\right|\\
    & \leq 2\sup_{\pi \in \Pi} \left|\hat{\theta}_L(\pi) - \E[\hat{\theta}_L(\pi)]\right| + 2\sup_{\pi \in \Pi} \left|\E[\hat{\theta}_L(\pi)] - \E[\theta_L(\pi)]\right|,
  \end{align*}
  where the first inequality follows from the fact that $\hat{\pi}$ maximizes $\hat{\theta}_L(\pi)$ so $\hat{\theta}_L(\hat{\pi}) \geq \hat{\theta}_L(\pi^\ast)$.
  Note that $\hat{\theta}_L(\pi) - \E[\hat{\theta}_L(\pi)]$ is a mean-zero empirical process indexed by the policy $\pi$.
  Defining the function
  $f_\pi(o) = \hat{\theta}_L(x, \pi(x))+ \left\< \hat{\bm{\lambda}}(x, \pi(x)), \hat{\bm{\varphi}}^{(b)}(o)\right\>$,
  and the function class $\mathcal{F} = \{f_\pi: \pi \in \Pi\}$, we can write the first term in the decomposition above as
  \[
    \sup_{\pi \in \Pi} \left|\hat{\theta}_L(\pi) - \E[\hat{\theta}_L(\pi)]\right| = \sup_{f \in \mathcal{F}} \left|\frac{1}{n}\sum_{i=1}^n f(O_i) - \E[f(O)]\right|.
  \]
  By the regularity conditions in Assumption~\ref{a:regularity}, the function class $\mathcal{F}$ is uniformly bounded, i.e. $\|f_\pi\|_\infty \leq\tilde{C}$ for some $\tilde{C}$, for all $\pi \in \Pi$. Therefore, by Theorem 4.2 in \citet{wainwright_2019}, we have that
  \[
    \sup_{f \in \mathcal{F}} \left|\frac{1}{n}\sum_{i=1}^n f(O_i) - \E[f(O)]\right| \leq 2 \mathcal{R}_n(\mathcal{F}) + \tilde{C} \sqrt{\frac{2 \log(1/\delta)}{n}},
  \]
  with probability at least $1-\delta$. To link the Rademacher complexity of the function class $\mathcal{F}$ to the Rademacher complexity of the policy class $\Pi$, we note that $f$ is a Lipschitz function of $\pi$.
  This is because by assumption the primal solution $\hat{\bm{p}}(x, \pi(x))$ is non-degenerate, meaning that the dual solution $\hat{\bm{\lambda}}(x, \pi(x))$ is unique, and it is a piece-wise linear function of the objective vector $\bm{c}(x, \pi(x))$ and so it is Lipschitz in the objective vector. In addition  $\bm{c}(x, \pi(x))$ is Lipschitz in $\pi(x)$. Furthermore, $f_\pi(o)$ is a linear function of the dual solution $\hat{\bm{\lambda}}(x, \pi(x))$, so it is also Lipschitz in $\pi(x)$. By the Talagrand contraction principle \citep[Theorem 4.12][]{ledoux_probability_1991}, we have that the Rademacher complexity of the function class $\mathcal{F}$ is bounded by the Rademacher complexity of the policy class $\Pi$, $\mathcal{R}_n(\mathcal{F}) \lesssim \mathcal{R}_n(\Pi)$.

  It remains to bound the second term in the decomposition of the excess regret,
  \[
    \sup_{\pi \in \Pi} \left|\E[\hat{\theta}_L(\pi)] - \E[\theta_L(\pi)]\right|.\]
  Note that under the strong optimality detection margin condition in Assumption~\ref{a:primal_infeasible_gap_strong}, we can directly apply Lemma~\ref{lem:direct_est_bias} for each policy $\pi \in \Pi$, noting that because the objective is bounded and known, we can obtain that
  \[
  \sup_{\pi \in \Pi} \left|\E[\hat{\theta}_L(\pi)] - \E[\theta_L(\pi)]\right| = O_p\left(\|\hat{\bm{b}} - \bm{b}\|_\infty^{1 + \alpha} + r_n\right).
  \]
  Putting together the pieces gives the result.

\end{proof}

\begin{proof}[Proof of Theorem~\ref{thm:entropic_policy}]
  Define $\hat{\theta}_L^\eta(\pi) \equiv\frac{1}{n}\sum_{i=1}^n \hat{\theta}_L^\eta(X_i, \pi(X_i)) + \left\< \nabla_{\bm{b}} \hat{\theta}_L^\eta(X_i, \pi(X_i)), \hat{\bm{\varphi}}^{(b)}(O_i) \right\>$. We begin by considering the entropic lower bound. Note that
  \begin{align*}
    \theta_L^\eta(\pi^\ast_\eta) - \theta_L^\eta(\hat{\pi}^\eta) & = \theta_L^\eta(\pi^\ast_\eta) - \hat{\theta}_L^\eta(\pi^\ast_\eta) + \hat{\theta}_L^\eta(\pi^\ast_\eta) - \hat{\theta}_L^\eta(\hat{\pi}^\eta) + \hat{\theta}_L^\eta(\hat{\pi}^\eta) - \theta_L^\eta(\hat{\pi}^\eta)\\
    & \leq \theta_L^\eta(\pi^\ast_\eta) - \hat{\theta}_L^\eta(\pi^\ast_\eta) + \hat{\theta}_L^\eta(\hat{\pi}^\eta) - \theta_L^\eta(\hat{\pi}^\eta)\\
    & \leq 2 \sup_{\pi \in \Pi} \left|\hat{\theta}_L^\eta(\pi) - \theta_L^\eta(\pi)\right|\\
    & \leq 2\sup_{\pi \in \Pi} \left|\hat{\theta}_L^\eta(\pi) - \E\left[\hat{\theta}_L^\eta(\pi)\right]\right| + 2\sup_{\pi \in \Pi} \left|\E\left[\hat{\theta}_L^\eta(\pi)\right] - \theta_L^\eta(\pi)\right|,
  \end{align*}
  where the first inequality follows from the fact that $\hat{\pi}^\eta$ maximizes $\hat{\theta}_L^\eta(\pi)$.

  For the first term, define the function $f_\pi(o) = \hat{\theta}_L^\eta(x, \pi(x)) + \left\< \nabla_{\bm{b}} \hat{\theta}_L^\eta(x, \pi(x)), \hat{\bm{\varphi}}^{(b)}(o) \right\>$, and the function class $\mathcal{F} = \{f_\pi \mid \pi \in \Pi\}$. Assumption \ref{a:regularity} ensures that $\mathcal{F}$ is uniformly bounded, and so as with the proof of Theorem~\ref{thm:primal_policy} by Theorem 4.2 in \citet{wainwright_2019}, we have that
  \[
    \sup_{f \in \mathcal{F}} \left|\frac{1}{n}\sum_{i=1}^n f(O_i) - \E[f(O)]\right| \leq 2 \mathcal{R}_n(\mathcal{F}) + \tilde{C} \sqrt{\frac{2 \log(1/\delta)}{n}},
  \]
  with probability at least $1-\delta$.
  Lemma~\ref{lem:jac_hes_bounds} shows that the norm of the Jacobian of the value $\theta_L^\eta(\bm{b}, \bm{c})$ with respect to the objective $\bm{c}$ is bounded by a constant, and that the operator norm of the Hessian is bounded by a constant times $\eta$.
  By Assumption~\ref{a:regularity}, the objective function $\bm{c}(x, \pi(x))$ is Lipschitz in $\pi(x)$, and so $f_\pi(o)$ is Lipschitz  in $\pi(x)$ for any $x$ as well, with a Lipschitz constant scaling with $\eta$. So by Talagrand's contraction principle we have that the Rademacher complexity of $\mathcal{F}$ is bounded by $\mathcal{R}_n(\mathcal{F}) \lesssim \eta \mathcal{R}_n(\Pi)$.
  To complete the first claim we must bound the second term in the decomposition above, $\sup_{\pi \in \Pi} \left|\E\left[\hat{\theta}_L^\eta(\pi)\right] - \theta_L^\eta(\pi)\right|$. We can do so by directly applying Lemma~\ref{lem:entropic_est_bias_fixed} to get that
  \[
    \sup_{\pi \in \Pi} \left|\E\left[\hat{\theta}_L^\eta(\pi)\right] - \theta_L^\eta(\pi)\right| = O_p\left(r_n + \|\hat{\bm{b}} - \bm{b}\|_2^2\right),
  \]
  noting that the objective vector $\bm{c}(x, \pi(x))$ is known in this case.

  For the second claim, we can similarly write the excess regret as
    \begin{align*}
    \theta_L(\pi^\ast) - \theta_L(\hat{\pi}^\eta) & = \theta_L(\pi^\ast) - \hat{\theta}_L^\eta(\pi^\ast) + \hat{\theta}_L^\eta(\pi^\ast) - \hat{\theta}_L^\eta(\hat{\pi}^\eta) + \hat{\theta}_L^\eta(\hat{\pi}^\eta) - \theta_L(\hat{\pi}^\eta)\\
    & \leq \theta_L(\pi^\ast) - \hat{\theta}_L^\eta(\pi^\ast) + \hat{\theta}_L^\eta(\hat{\pi}^\eta) - \theta_L^\eta(\hat{\pi}^\eta)\\
    & \leq 2 \sup_{\pi \in \Pi} \left|\hat{\theta}_L^\eta(\pi) - \theta_L(\pi)\right|\\
    & \leq 2\sup_{\pi \in \Pi} \left|\hat{\theta}_L^\eta(\pi) - \E\left[\hat{\theta}_L^\eta(\pi)\right]\right| + 2\sup_{\pi \in \Pi} \left|\E\left[\hat{\theta}_L^\eta(\pi)\right] - \theta_L(\pi)\right|.
  \end{align*}
  We have just controlled the first term above. To control the second term we can apply Lemma~\ref{lem:entropic_est_bias_growing} and the argument in Theorem~\ref{thm:entropic_est_approx} for each $\pi \in \Pi$, where the various constants are uniform over $\pi \in \Pi$,  to obtain that

  \[
    \sup_{\pi \in \Pi} \left|\E\left[\hat{\theta}_L^\eta(\pi)\right] - \theta_L(\pi)\right| = O_p\left(\rho_\alpha(\eta)+  r_n + \|\hat{\bm{b}} - \bm{b}\|_\infty^{1+q}\right).
  \]
  Combining these bounds gives the result.

\end{proof}

\clearpage
\singlespacing
\bibliographystyle{apalike}
\bibliography{citations}

@article{Pu2021,
	title = {Estimating optimal treatment rules with an instrumental variable: {A} partial identification learning approach},
	volume = {83},
	issn = {23318422},
	doi = {10.1111/rssb.12413},
	number = {2},
	journal = {Journal of the Royal Statistical Society Series B},
	author = {Pu, Hongming and Zhang, Bo},
	year = {2021},
	pages = {318--345}
}

@book{wainwright_2019, place={Cambridge}, series={Cambridge Series in Statistical and Probabilistic Mathematics}, title={High-Dimensional Statistics: A Non-Asymptotic Viewpoint}, DOI={10.1017/9781108627771}, publisher={Cambridge University Press}, author={Wainwright, Martin J.}, year={2019}, collection={Cambridge Series in Statistical and Probabilistic Mathematics}}

@article{Athey2021,
archivePrefix = {arXiv},
arxivId = {1702.02896},
author = {Athey, Susan and Wager, Stefan},
doi = {10.3982/ecta15732},
eprint = {1702.02896},
issn = {0012-9682},
journal = {Econometrica},
number = {1},
pages = {133--161},
title = {Policy Learning With Observational Data},
volume = {89},
year = {2021}
}

@article{Kallus2021,
author = {Kallus, Nathan and Zhou, Angela},
doi = {10.1287/mnsc.2020.3699},
issn = {15265501},
journal = {Management Science},
number = {5},
pages = {2870--2890},
title = {{Minimax-optimal policy learning under unobserved confounding}},
volume = {67},
year = {2021}
}

@article{Luedtke2016,
author = {Luedtke, Alexander R. and {Van Der Laan}, Mark J.},
doi = {10.1214/15-AOS1384},
issn = {00905364},
journal = {Annals of Statistics},
number = {2},
pages = {713--742},
title = {{Statistical inference for the mean outcome under a possibly non-unique optimal treatment strategy}},
volume = {44},
year = {2016}
}

@book{Manski2005_partial,
 ISBN = {9780691121536},
 URL = {http://www.jstor.org/stable/j.ctv14164s1},
 author = {Charles F. Manski},
 publisher = {Princeton University Press},
 title = {Social Choice with Partial Knowledge of Treatment Response},
 year = {2005}
}

@article{Cui2021_partial,
author = {Cui, Yifan},
journal = {Harvard Data Science Review},
title = {{Individualized decision making under partial identification: three perspectives, two optimality results, and one paradox}},
year = {2021},
note = {Just accepted}
}

@article{Rosenbaum1983,
author = {Rosenbaum, Paul R and Rubin, Donald B},
journal = {Biometrika},
number = {1},
pages = {41--55},
title = {{The Central Role of the Propensity Score in Observational Studies for Causal Effects}},
url = {http://links.jstor.org/sici?sici=0006-3444%28198304%2970%3A1%3C41%3ATCROTP%3E2.0.CO%3B2-Q},
volume = {70},
year = {1983}
}

@article{Han2021_partial,
archivePrefix = {arXiv},
arxivId = {1912.10014},
author = {Han, Sukjin},
eprint = {1912.10014},
journal = {Journal of the American Statistical Association},
title = {{Optimal Dynamic Treatment Regimes and Partial Welfare Ordering}},
url = {http://arxiv.org/abs/1912.10014},
year = {2021}
}

@article{benmichael2021_safe,
	title = {Safe {Policy} {Learning} through {Extrapolation}: {Application} to {Pre}-trial {Risk} {Assessment}},
	volume = {120},
	issn = {0162-1459},
	shorttitle = {Safe {Policy} {Learning} through {Extrapolation}},
	doi = {10.1080/01621459.2025.2489135},
	number = {551},
	urldate = {2025-12-15},
	journal = {Journal of the American Statistical Association},
	author = {Ben-Michael, Eli and Greiner, D. James and Imai, Kosuke and Jiang, Zhichao},
	year = {2025},
	pages = {1386--1399}
}

@inproceedings{Kallus2022_harm,
 author = {Kallus, Nathan},
 booktitle = {Advances in Neural Information Processing Systems},
 pages = {15996--16009},
 title = {What's the Harm? Sharp Bounds on the Fraction Negatively Affected by Treatment},
 url = {https://proceedings.neurips.cc/paper_files/paper/2022/file/666cccc6376058e251315b4de7e085b9-Paper-Conference.pdf},
 volume = {35},
 year = {2022}
}

@article{Audibert2007,
author = {Audibert, Jean Yves and Tsybakov, Alexandre B.},
doi = {10.1214/009053606000001217},
issn = {00905364},
journal = {Annals of Statistics},
number = {2},
pages = {608--633},
title = {{Fast learning rates for plug-in classifiers}},
volume = {35},
year = {2007}
}

@unpublished{DAdamo2023,
archivePrefix = {arXiv},
arxivId = {2111.10904},
author = {D'Adamo, Riccardo},
eprint = {2111.10904},
title = {{Orthogonal Policy Learning Under Ambiguity}},
url = {http://arxiv.org/abs/2111.10904},
year = {2023}
}

@article{Kennedy2022_drlearner,
	title = {Towards optimal doubly robust estimation of heterogeneous causal effects},
	volume = {17},
	issn = {1935-7524, 1935-7524},
	doi = {10.1214/23-EJS2157},
	number = {2},
	urldate = {2024-10-15},
	journal = {Electronic Journal of Statistics},
	author = {Kennedy, Edward H.},
	year = {2023},
	pages = {3008--3049},
}

@article{Manski2011,
author = {Manski, Charles F},
doi = {10.1146/annurev-economics-061109-080359},
issn = {19411383},
journal = {Annual Review of Economics},
pages = {25--49},
title = {{Choosing treatment policies under ambiguity}},
volume = {3},
year = {2011}
}

@unpublished{Zhang2022_safe,
archivePrefix = {arXiv},
arxivId = {2208.13323},
author = {Zhang, Yi and Ben-Michael, Eli and Imai, Kosuke},
eprint = {2208.13323},
title = {{Safe Policy Learning under Regression Discontinuity Designs with Multiple Cutoffs}},
url = {http://arxiv.org/abs/2208.13323},
year = {2023}
}

@book{moulin_axioms_1988,
	address = {Cambridge},
	series = {Econometric {Society} {Monographs}},
	title = {Axioms of {Cooperative} {Decision} {Making}},
	isbn = {978-0-521-36055-5},
	url = {https://www.cambridge.org/core/product/481FCBCDD15F3CCEE2FE381E7BF17B3D},
	publisher = {Cambridge University Press},
	author = {Moulin, Herv\'{e}},
	year = {1988},
	doi = {10.1017/CCOL0521360552}
}

@incollection{kennedy_semiparametric_2024,
	title = {Semiparametric {Doubly} {Robust} {Targeted} {Double} {Machine} {Learning}: {A} {Review}},
	isbn = {978-1-00-321622-3},
	shorttitle = {Semiparametric {Doubly} {Robust} {Targeted} {Double} {Machine} {Learning}},
	booktitle = {Handbook of {Statistical} {Methods} for {Precision} {Medicine}},
	publisher = {Chapman and Hall/CRC},
	author = {Kennedy, Edward H.},
	year = {2024},
}

@inproceedings{weed_explicit_2018,
	title = {An explicit analysis of the entropic penalty in linear programming},
	url = {https://proceedings.mlr.press/v75/weed18a.html},
	language = {en},
	urldate = {2025-03-06},
	booktitle = {Proceedings of the 31st  {Conference} {On} {Learning} {Theory}},
	publisher = {PMLR},
	author = {Weed, Jonathan},
	month = jul,
	year = {2018},
	note = {ISSN: 2640-3498},
	pages = {1841--1855},
}

@article{imbens_confidence_2004,
	title = {Confidence {Intervals} for {Partially} {Identified} {Parameters}},
	volume = {72},
	issn = {1468-0262},
	url = {https://onlinelibrary.wiley.com/doi/abs/10.1111/j.1468-0262.2004.00555.x},
	doi = {10.1111/j.1468-0262.2004.00555.x},
	language = {en},
	number = {6},
	urldate = {2024-04-01},
	journal = {Econometrica},
	author = {Imbens, Guido W. and Manski, Charles F.},
	year = {2004},
	pages = {1845--1857},
}

@article{balke_bounds_1997,
	title = {Bounds on treatment effects from studies with imperfect compliance},
	volume = {92},
	issn = {1537274X},
	doi = {10.1080/01621459.1997.10474074},
	number = {439},
	journal = {Journal of the American Statistical Association},
	author = {Balke, Alexander and Pearl, Judea},
	year = {1997},
	pages = {1171--1176},
}

@misc{levis_assisted_2023,
	title = {Covariate-assisted bounds on causal effects with instrumental variables},
	url = {http://arxiv.org/abs/2301.12106},
	urldate = {2025-03-20},
	publisher = {arXiv},
	author = {Levis, Alexander W. and Bonvini, Matteo and Zeng, Zhenghao and Keele, Luke and Kennedy, Edward H.},
	month = sep,
	year = {2023},
	note = {arXiv:2301.12106},
}

@misc{ji_model-agnostic_2024,
	title = {Model-{Agnostic} {Covariate}-{Assisted} {Inference} on {Partially} {Identified} {Causal} {Effects}},
	url = {http://arxiv.org/abs/2310.08115},
	doi = {10.48550/arXiv.2310.08115},
	urldate = {2025-02-11},
	publisher = {arXiv},
	author = {Ji, Wenlong and Lei, Lihua and Spector, Asher},
	month = nov,
	year = {2024},
	note = {arXiv:2310.08115},
}

@misc{jiang_semiparametric_2025,
	title = {Semiparametric {Inference} for {Partially} {Identifiable} {Data} {Fusion} {Estimands} via {Double} {Machine} {Learning}},
	url = {https://arxiv.org/abs/2502.05319v1},
	author = {Jiang, Yicong and Janson, Lucas},
	year = {2025},
  note = {arXiv:2502.05319},
}

@article{kennedy_sharp_2020,
	title = {Sharp instruments for classifying compliers and generalizing causal effects},
	volume = {48},
	issn = {0090-5364, 2168-8966},
	url = {https://projecteuclid.org/journals/annals-of-statistics/volume-48/issue-4/Sharp-instruments-for-classifying-compliers-and-generalizing-causal-effects/10.1214/19-AOS1874.full},
	doi = {10.1214/19-AOS1874},
	number = {4},
	urldate = {2024-06-24},
	journal = {The Annals of Statistics},
	author = {Kennedy, Edward H. and Balakrishnan, Sivaraman and G'Sell, Max},
	month = aug,
	year = {2020},
	pages = {2008--2030}
}

@incollection{dantzig_maximization_1951,
	title = {Maximization of a {Linear} {Function} of {Variables} {Subject} to {Linear} {Inequalities}},
	booktitle = {Activity {Analysis} of {Production} and {Allocation}},
	publisher = {Wiley \& Chapman-Hall},
	author = {Dantzig, George Bernard},
	editor = {Koopmans, Tjalling Charles},
	year = {1951},
	pages = {339--347},
}

@article{spielman_smoothed_2004,
	title = {Smoothed analysis of algorithms: {Why} the simplex algorithm usually takes polynomial time},
	volume = {51},
	issn = {0004-5411},
	shorttitle = {Smoothed analysis of algorithms},
	url = {https://dl.acm.org/doi/10.1145/990308.990310},
	doi = {10.1145/990308.990310},
	number = {3},
	urldate = {2025-04-17},
	journal = {Journal of the ACM},
	author = {Spielman, Daniel A. and Teng, Shang-Hua},
	month = may,
	year = {2004},
	pages = {385--463},
}

@misc{semenova_aggregated_2024,
	title = {Debiased {Machine} {Learning} of {Aggregated} {Intersection} {Bounds} and {Other} {Causal} {Parameters}},
	url = {http://arxiv.org/abs/2303.00982},
	doi = {10.48550/arXiv.2303.00982},
	author = {Semenova, Vira},
	year = {2025},
	note = {arXiv:2303.00982},
}

@inproceedings{cuturi_sinkhorn_2013,
	title = {Sinkhorn {Distances}: {Lightspeed} {Computation} of {Optimal} {Transport}},
	volume = {26},
	shorttitle = {Sinkhorn {Distances}},
	url = {https://papers.nips.cc/paper_files/paper/2013/hash/af21d0c97db2e27e13572cbf59eb343d-Abstract.html},
	urldate = {2025-04-18},
	booktitle = {Advances in {Neural} {Information} {Processing} {Systems}},
	author = {Cuturi, Marco},
	year = {2013},
}

@article{sinkhorn_diagonal_1967,
	title = {Diagonal {Equivalence} to {Matrices} with {Prescribed} {Row} and {Column} {Sums}},
	volume = {74},
	issn = {0002-9890},
	url = {https://www.jstor.org/stable/2314570},
	doi = {10.2307/2314570},
	number = {4},
	urldate = {2025-04-18},
	journal = {The American Mathematical Monthly},
	author = {Sinkhorn, Richard},
	year = {1967},
	pages = {402--405},
}

@article{benmichael_asymmetric_2024,
	title = {Policy {Learning} with {Asymmetric} {Counterfactual} {Utilities}},
	volume = {119},
	issn = {0162-1459},
	doi = {10.1080/01621459.2023.2300507},
	number = {548},
	urldate = {2025-01-31},
	journal = {Journal of the American Statistical Association},
	publisher = {ASA Website},
	author = {Ben-Michael, Eli and Imai, Kosuke and Jiang, Zhichao},
	month = oct,
	year = {2024},
	pages = {3045--3058}
}

@article{finkelstein_oregon_2012,
	title = {The {Oregon} {Health} {Insurance} {Experiment}: {Evidence} from the {First} {Year}},
	volume = {127},
	issn = {0033-5533},
	shorttitle = {The {Oregon} {Health} {Insurance} {Experiment}},
	url = {https://doi.org/10.1093/qje/qjs020},
	doi = {10.1093/qje/qjs020},
	number = {3},
	urldate = {2024-07-29},
	journal = {The Quarterly Journal of Economics},
	author = {Finkelstein, Amy and Taubman, Sarah and Wright, Bill and Bernstein, Mira and Gruber, Jonathan and Newhouse, Joseph P. and Allen, Heidi and Baicker, Katherine and {Oregon Health Study Group}},
	month = aug,
	year = {2012},
	pages = {1057--1106}
}

@article{taubman_medicaid_2014,
	title = {Medicaid {Increases} {Emergency}-{Department} {Use}: {Evidence} from {Oregon}'s {Health} {Insurance} {Experiment}},
	volume = {343},
	shorttitle = {Medicaid {Increases} {Emergency}-{Department} {Use}},
	url = {https://www.science.org/doi/full/10.1126/science.1246183},
	doi = {10.1126/science.1246183},
	number = {6168},
	urldate = {2025-04-24},
	journal = {Science},
	author = {Taubman, Sarah L. and Allen, Heidi L. and Wright, Bill J. and Baicker, Katherine and Finkelstein, Amy N.},
	month = jan,
	year = {2014},
	pages = {263--268}
}

@article{chernozhukov_doubledebiased_2018,
	title = {Double/debiased machine learning for treatment and structural parameters},
	volume = {21},
	issn = {1368-4221},
	url = {https://doi.org/10.1111/ectj.12097},
	doi = {10.1111/ectj.12097},
	number = {1},
	urldate = {2025-05-13},
	journal = {The Econometrics Journal},
	author = {Chernozhukov, Victor and Chetverikov, Denis and Demirer, Mert and Duflo, Esther and Hansen, Christian and Newey, Whitney and Robins, James},
	month = feb,
	year = {2018},
	pages = {C1--C68},
}

@article{Robins1994,
	title = {Estimation of {Regression} {Coefficients} {When} {Some} {Regressors} are not {Always} {Observed}},
	volume = {89},
	url = {http://www.tandfonline.com/action/journalInformation?journalCode=uasa20},
	doi = {10.1080/01621459.1994.10476818doi.org/10.1080/01621459.1994.10476818},
	number = {427},
	urldate = {2018-04-26},
	journal = {Journal of the American Statistical Association},
	author = {Robins, James M and Rotnitzky, Andrea and Ping Zhao, Lue},
	year = {1994},
	pages = {846--866},
}

@book{ledoux_probability_1991,
	address = {Berlin, Heidelberg},
	title = {Probability in {Banach} {Spaces}},
	copyright = {http://www.springer.com/tdm},
	isbn = {978-3-642-20211-7 978-3-642-20212-4},
	url = {http://link.springer.com/10.1007/978-3-642-20212-4},
	urldate = {2025-03-18},
	publisher = {Springer},
	author = {Ledoux, Michel and Talagrand, Michel},
	year = {1991},
	doi = {10.1007/978-3-642-20212-4}
}

@article{gabriel_sharp_2024,
	title = {Sharp symbolic nonparametric bounds for measures of benefit in observational and imperfect randomized studies with ordinal outcomes},
	volume = {111},
	issn = {1464-3510},
	url = {https://doi.org/10.1093/biomet/asae020},
	doi = {10.1093/biomet/asae020},
	number = {4},
	urldate = {2025-06-02},
	journal = {Biometrika},
	author = {Gabriel, Erin E and Sachs, Michael C and Jensen, Andreas Kryger},
	month = dec,
	year = {2024},
	pages = {1429--1436},
}

@article{sachs_general_2023,
	title = {A {General} {Method} for {Deriving} {Tight} {Symbolic} {Bounds} on {Causal} {Effects}},
	volume = {32},
	issn = {1061-8600},
	url = {https://doi.org/10.1080/10618600.2022.2071905},
	doi = {10.1080/10618600.2022.2071905},
	number = {2},
	urldate = {2025-06-02},
	journal = {Journal of Computational and Graphical Statistics},
	author = {Sachs, Michael C. and Jonzon, Gustav and Sjölander, Arvid and Gabriel, Erin E.},
	month = apr,
	year = {2023},
	pages = {567--576},
}

@article{benmichael_ai_2024,
	title = {Does {AI} help humans make better decisions? {A} statistical evaluation framework for experimental and observational studies},
	volume = {122},
	shorttitle = {Does {AI} help humans make better decisions?},
	doi = {10.1073/pnas.2505106122},
	number = {38},
	urldate = {2025-11-12},
	journal = {Proceedings of the National Academy of Sciences},
	publisher = {Proceedings of the National Academy of Sciences},
	author = {Ben-Michael, Eli and Greiner, D. James and Huang, Melody and Imai, Kosuke and Jiang, Zhichao and Shin, Sooahn},
	year = {2025},
	pages = {e2505106122}
}

@misc{lanners_data_2025,
	title = {Data {Fusion} for {Partial} {Identification} of {Causal} {Effects}},
	url = {http://arxiv.org/abs/2505.24296},
	doi = {10.48550/arXiv.2505.24296},
	urldate = {2025-06-02},
	publisher = {arXiv},
	author = {Lanners, Quinn and Rudin, Cynthia and Volfovsky, Alexander and Parikh, Harsh},
	month = may,
	year = {2025},
	note = {arXiv:2505.24296 [stat]},
}

@misc{liu_beyond_2025,
	title = {Beyond entropic regularization: {Debiased} {Gaussian} estimators for discrete optimal transport and general linear programs},
	shorttitle = {Beyond entropic regularization},
	url = {http://arxiv.org/abs/2505.04312},
	doi = {10.48550/arXiv.2505.04312},
	urldate = {2025-05-08},
	publisher = {arXiv},
	author = {Liu, Shuyu and Bunea, Florentina and Niles-Weed, Jonathan},
	month = may,
	year = {2025},
	note = {arXiv:2505.04312 [math]},
}

@article{duarte_automated_2024,
	title = {An {Automated} {Approach} to {Causal} {Inference} in {Discrete} {Settings}},
	volume = {119},
	issn = {0162-1459},
	url = {https://doi.org/10.1080/01621459.2023.2216909},
	doi = {10.1080/01621459.2023.2216909},
	number = {547},
	urldate = {2025-06-11},
	journal = {Journal of the American Statistical Association},
	author = {Duarte, Guilherme and Finkelstein, Noam and Knox, Dean and Mummolo, Jonathan and Shpitser, Ilya},
	year = {2024},
	pages = {1778--1793},
}

@article{andrews_inference_2023,
	title = {Inference for {Linear} {Conditional} {Moment} {Inequalities}},
	volume = {90},
	issn = {0034-6527},
	number = {6},
	journal = {The Review of Economic Studies},
	author = {Andrews, Isaiah and Roth, Jonathan and Pakes, Ariel},
	year = {2023},
	pages = {2763--2791}
}

@article{fang_inference_2023,
	title = {Inference for {Large}-{Scale} {Linear} {Systems} {With} {Known} {Coefficients}},
	volume = {91},
	issn = {1468-0262},
	url = {https://onlinelibrary.wiley.com/doi/abs/10.3982/ECTA18979},
	doi = {10.3982/ECTA18979},
	language = {en},
	number = {1},
	urldate = {2025-03-12},
	journal = {Econometrica},
	author = {Fang, Zheng and Santos, Andres and Shaikh, Azeem M. and Torgovitsky, Alexander},
	year = {2023},
	pages = {299--327}
}

@article{gafarov_simple_2025,
  title = {Simple subvector inference on sharp identified set in affine models},
  journal = {Journal of Econometrics},
  volume = {249},
  pages = {105952},
  year = {2025},
  issn = {0304-4076},
  doi = {https://doi.org/10.1016/j.jeconom.2025.105952},
  url = {https://www.sciencedirect.com/science/article/pii/S0304407625000065},
  author = {Bulat Gafarov}
}

@misc{goff_inference_2025,
	title = {Inference on the value of linear programs},
	url = {http://arxiv.org/abs/2506.06776},
	doi = {10.48550/arXiv.2506.06776},
	author = {Goff, Leonard and Mbakop, Eric},
	year = {2025},
	note = {arXiv:2506.06776 [econ]}
}

@article{hsieh_inference_2022,
	title = {Inference on estimators defined by mathematical programming},
	volume = {226},
	issn = {03044076},
	doi = {10.1016/j.jeconom.2021.06.001},
	number = {2},
	urldate = {2026-03-16},
	journal = {Journal of Econometrics},
	author = {Hsieh, Yu-Wei and Shi, Xiaoxia and Shum, Matthew},
	year = {2022},
	pages = {248--268}
}

@article{cho_simple_2024,
	title = {Simple {Inference} on {Functionals} of {Set}-{Identified} {Parameters} {Defined} by {Linear} {Moments}},
	volume = {42},
	issn = {0735-0015},
	url = {https://doi.org/10.1080/07350015.2023.2203768},
	doi = {10.1080/07350015.2023.2203768},
	number = {2},
	journal = {Journal of Business \& Economic Statistics},
	publisher = {Taylor \& Francis},
	author = {Cho, JoonHwan and Russell, Thomas M.},
	year = {2024},
	pages = {563--578}
}

@article{cox_simple_2023,
	title = {Simple {Adaptive} {Size}-{Exact} {Testing} for {Full}-{Vector} and {Subvector} {Inference} in {Moment} {Inequality} {Models}},
	volume = {90},
	copyright = {https://academic.oup.com/journals/pages/open\_access/funder\_policies/chorus/standard\_publication\_model},
	issn = {0034-6527, 1467-937X},
	url = {https://academic.oup.com/restud/article/90/1/201/6548809},
	doi = {10.1093/restud/rdac015},
	language = {en},
	number = {1},
	urldate = {2026-03-16},
	journal = {The Review of Economic Studies},
	author = {Cox, Gregory and Shi, Xiaoxia},
	year = {2023},
	pages = {201--228}
}

@article{freyberger_identification_2015,
	title = {Identification and shape restrictions in nonparametric instrumental variables estimation},
	volume = {189},
	issn = {03044076},
	url = {https://linkinghub.elsevier.com/retrieve/pii/S0304407615001918},
	doi = {10.1016/j.jeconom.2015.06.020},
	language = {en},
	number = {1},
	urldate = {2026-03-16},
	journal = {Journal of Econometrics},
	author = {Freyberger, Joachim and Horowitz, Joel L.},
	year = {2015},
	pages = {41--53}
}

@book{vandervaart_weak_1996,
	address = {New York, NY},
	series = {Springer {Series} in {Statistics}},
	title = {Weak {Convergence} and {Empirical} {Processes}},
	copyright = {http://www.springer.com/tdm},
	isbn = {978-1-4757-2547-6 978-1-4757-2545-2},
	doi = {10.1007/978-1-4757-2545-2},
	language = {en},
	urldate = {2026-03-19},
	publisher = {Springer New York},
	author = {Van Der Vaart, Aad W. and Wellner, Jon A.},
	year = {1996},
}

@article{mogstad_using_2018,
	title = {Using {Instrumental} {Variables} for {Inference} {About} {Policy} {Relevant} {Treatment} {Parameters}},
	volume = {86},
	issn = {1468-0262},
	url = {https://onlinelibrary.wiley.com/doi/abs/10.3982/ECTA15463},
	doi = {10.3982/ECTA15463},
	number = {5},
	urldate = {2026-03-16},
	journal = {Econometrica},
	author = {Mogstad, Magne and Santos, Andres and Torgovitsky, Alexander},
	year = {2018},
	pages = {1589--1619}
}

@misc{whitehouse_inference_2026,
	title = {Inference on {Optimal} {Policy} {Values} and {Other} {Irregular} {Functionals} via {Softmax} {Smoothing}},
	url = {http://arxiv.org/abs/2507.11780},
	doi = {10.48550/arXiv.2507.11780},
	urldate = {2026-08-03},
	publisher = {arXiv},
	author = {Whitehouse, Justin and Chen, Qizhao and Austern, Morgane and Syrgkanis, Vasilis},
	year = {2026},
	note = {arXiv:2507.11780}
}

@misc{muller_fisher-rao_2025,
	title = {Fisher-{Rao} {Gradient} {Flows} of {Linear} {Programs} and {State}-{Action} {Natural} {Policy} {Gradients}},
	url = {http://arxiv.org/abs/2403.19448},
	doi = {10.48550/arXiv.2403.19448},
	urldate = {2026-07-22},
	publisher = {arXiv},
	author = {Müller, Johannes and Çaycı, Semih and Montúfar, Guido},
	month = feb,
	year = {2025},
	note = {arXiv:2403.19448 [math.OC]},
}

@article{alvarez_hessian_2004,
	title = {Hessian {Riemannian} {Gradient} {Flows} in {Convex} {Programming}},
	volume = {43},
	issn = {0363-0129, 1095-7138},
	url = {http://epubs.siam.org/doi/10.1137/S0363012902419977},
	doi = {10.1137/S0363012902419977},
	language = {en},
	number = {2},
	urldate = {2026-08-03},
	journal = {SIAM Journal on Control and Optimization},
	author = {Alvarez, Felipe and Bolte, Jérôme and Brahic, Olivier},
	month = jan,
	year = {2004},
	pages = {477--501}
}

@misc{khandamiryan_adaptive_2026,
	title = {Adaptive {Estimation} of {Aggregated} {Values} of {Conditional} {Linear} {Programs}},
	url = {http://arxiv.org/abs/2606.08359},
	doi = {10.48550/arXiv.2606.08359},
	urldate = {2026-08-06},
	publisher = {arXiv},
	author = {Khandamiryan, Gevorg and Semenova, Vira},
	month = jun,
	year = {2026},
	note = {arXiv:2606.08359},
}

@article{ben-tal_geometric_1990,
	title = {A geometric property of the least squares solution of linear equations},
	volume = {139},
	issn = {0024-3795},
	doi = {10.1016/0024-3795(90)90395-S},
	urldate = {2026-10-16},
	journal = {Linear Algebra and its Applications},
	author = {Ben-Tal, Aharon and Teboulle, Marc},
	month = oct,
	year = {1990},
	pages = {165--170}
}

\end{document}